\documentclass[12pt,onecolumn,twoside]{IEEEtran}

\usepackage{calc}

\usepackage{multirow}
\usepackage{cite}
\usepackage{graphicx,subfigure}
\usepackage{psfrag}
\usepackage{amsmath,amssymb,amsthm}
\usepackage{color}
\usepackage{tikz} 
\usepackage{bm}
\usepackage{bbm}
\usepackage{verbatim}
\usepackage[T1]{fontenc}

\usepackage{cases}
\usepackage[noend]{algpseudocode}

\usepackage{xcolor}
\usepackage[most]{tcolorbox}
\usepackage{bm}

\usepackage{tabu}
\usepackage{mathtools}
\usepackage{xifthen}
\usepackage{booktabs}
\usepackage{algorithm}
\usepackage{algpseudocode}
\usepackage{arydshln}

\DeclarePairedDelimiter\ceil{\lceil}{\rceil}
\DeclarePairedDelimiter\floor{\lfloor}{\rfloor}

\usepackage{algorithm}
\makeatletter
\def\BState{\State\hskip-\ALG@thistlm}
\makeatother

\DeclareMathOperator*{\defeq}{\triangleq}

\newtheorem{theorem}{Theorem}

\newtheorem{lemma}{Lemma}

\newcommand{\bit}{\begin{itemize}}
\newcommand{\eit}{\end{itemize}}

\newcommand{\bc}{\begin{center}}
\newcommand{\ec}{\end{center}}

\newcommand{\ba}{\begin{array}}
\newcommand{\ea}{\end{array}}

\newcommand{\beq}{\begin{equation}}
\newcommand{\eeq}{\end{equation}}

\newcommand{\beqn}{\begin{equation*}}
\newcommand{\eeqn}{\end{equation*}}

\newcommand{\bean}{\begin{eqnarray*}}
\newcommand{\eean}{\end{eqnarray*}}
\newcommand{\bea}{\begin{eqnarray}}
\newcommand{\eea}{\end{eqnarray}}

\def\hv{\boldsymbol{h}}

\def\wv{\boldsymbol{w}}
\def\xv{\boldsymbol{x}}

\def\zv{\boldsymbol{z}}

\def\Hm{\boldsymbol{H}}

\newcommand{\Ac}{{\mathcal A}}
\newcommand{\Bc}{{\mathcal B}}
\newcommand{\Cc}{{\mathcal C}}
\newcommand{\Dc}{{\mathcal D}}
\newcommand{\Ec}{{\mathcal E}}
\newcommand{\Fc}{{\mathcal F}}

\newcommand{\Pc}{{\mathcal P}}

\newcommand{\Sc}{{\mathcal S}}

\newcommand{\T}{{\scriptscriptstyle\mathsf{T}}}

\newtheorem{remark}{Remark}

\renewcommand{\Bmatrix}[1]{\begin{bmatrix}#1\end{bmatrix}}

\newcommand{\non}{\nonumber}

\newcommand{\Pros}{\text{processors}}
\newcommand{\Pro}{\text{processor}}
\newcommand{\PRO}{\text{Processor}}
\newcommand{\Lh}{\ell}
\newcommand{\Me}{\wv}
\newcommand{\Ry}{\mathrm{s}}

\newcommand{\Vr}{\mathrm{v}} 
\newcommand{\gn}{\eta} 
 
\newcommand{\Meg}{\bar{\wv}} 
 
\newcommand{\Bit}{b} 
\newcommand{\Rd}{r} 
\newcommand{\cb}{c} 
\newcommand{\Ss}{\Sc}

\newcommand{\Lin}{l}
\newcommand{\Lk}{\mathrm{u}}
\newcommand{\Ce}{\beta}
\newcommand{\Le}{\alpha}
\newcommand{\Fe}{\delta}
\newcommand{\BT}{\Bit}

\algdef{SE}[SUBALG]{Indent}{EndIndent}{}{\algorithmicend\ }%
\algtext*{Indent}
\algtext*{EndIndent}

\begin{document}
\sloppy
\title{Fundamental Limits of Byzantine Agreement}

\author{Jinyuan Chen 
\thanks{Jinyuan Chen is with Louisiana Tech University, Department of Electrical Engineering, Ruston, USA (email:  jinyuan@latech.edu). 
} 
}

\maketitle
\pagestyle{headings}

\begin{abstract}

Byzantine agreement (BA)  is a distributed consensus problem where $n$ $\Pros$ want to reach agreement on an $\Lh$-bit message or value, but up to $t$ of the $\Pros$ are dishonest or faulty.  The challenge of this BA problem lies in achieving agreement despite the presence of dishonest $\Pros$ who may arbitrarily deviate from the designed protocol.   
 The quality of a BA protocol is measured primarily by using the following three parameters: the number of processors $n$ as a function of $t$ allowed (\emph{resilience});   the number of rounds (\emph{round complexity}, denoted by $\Rd$); and  the total number of communication bits (\emph{communication complexity},  denoted by $\BT$). 
 For any error-free BA protocol, the known lower bounds on those three parameters  are  $n\geq 3t + 1$, $\Rd \geq t + 1$ and $\BT \geq \Omega(\max\{n\Lh, nt\})$, respectively, where a protocol that is guaranteed to be correct in all executions is said to be \emph{error free}.

In this work   by using  coding theory, together with graph theory and linear algebra,   we  design a {\bf co}ded  BA  protoc{\bf ol} (termed as COOL)   that achieves consensus on an $\Lh$-bit message  with optimal \emph{resilience},  asymptotically optimal \emph{round complexity}, and asymptotically optimal \emph{communication complexity}  when $\ell \geq t   \log t$, simultaneously.  
 The proposed COOL is an \emph{error-free}  and \emph{deterministic} BA protocol  that does not rely on cryptographic technique such as signatures, hashing, authentication and secret sharing (\emph{signature free}). It is secure  against computationally unbounded adversary who takes full control over the dishonest $\Pros$ (\emph{information-theoretic secure}). The main idea of the proposed COOL is to use a carefully-crafted error correction code that provides an efficient way of exchanging ``compressed'' information among distributed nodes, while keeping the ability of detecting errors, masking errors, and making a consistent and validated agreement at honest distributed nodes. 
 With the achievable performance by the proposed COOL and the known lower bounds, we characterize the optimal \emph{communication complexity exponent} as 
 \[  \Ce^*(\Le, \Fe) = \max\{1+  \Le, 1+  \Fe  \}  \]
 under the optimal resilience and asymptotically optimal round complexity, for $\Ce (\Le, \Fe)\defeq $ $ \lim_{n \to \infty} \log \BT (n, \Fe, \Le)/\log n$, $\Le  \defeq  \lim_{n \to \infty} \log \ell/\log n$ and $\Fe \defeq  \lim_{n \to \infty} \log t/\log n$.    
 We show that our results can also be extended to the setting of  Byzantine broadcast, aka Byzantine generals problem, where the honest $\Pros$ want to agree on the message sent by a leader who is potentially dishonest. 
This work reveals that \emph{coding} is an effective approach for achieving the fundamental limits of Byzantine agreement and its variants.  Our protocol analysis borrows tools from coding theory, graph theory and linear algebra.

\end{abstract}


\section{Introduction}

Byzantine agreement (BA), as originally proposed by Pease, Shostak and Lamport in 1980,  is a distributed consensus problem where $n$ $\Pros$ want to reach agreement on some message (or value), but up to $t$ of the $\Pros$ are dishonest or faulty \cite{PSL:80}.  The challenge of this BA problem lies in achieving agreement despite the presence of dishonest $\Pros$ who may arbitrarily deviate from the designed protocol.  One variant of the problem is  Byzantine broadcast (BB),  aka Byzantine generals problem, where the honest $\Pros$ want to agree on the message sent by a leader who is potentially dishonest \cite{LSP:82}.  Byzantine agreement and its variants are considered to be the fundamental building blocks for distributed systems and cryptography including Byzantine-fault-tolerant (BFT) distributed computing, distributed storage,  blockchain protocols, state machine replication and voting, just to name a few \cite{PSL:80, LSP:82, FH:06, LV:11, GP:20, LDK:20, NRSVX:20, Patra:11, CT:05, MXCSS:16,  ADDNR:17, CS:20, Shi:19, PS:18, QZJZ:20, GWGR:04, SB:12}.

To solve the Byzantine agreement problem, a designed protocol needs to satisfy the following  conditions: every honest $\Pro$  eventually outputs a message and terminates (\emph{termination});  all honest $\Pros$  output the same message (\emph{consistency}); and  if all honest $\Pros$ hold the same initial message  then they output this initial message (\emph{validity}).   
 A protocol that satisfies the above three conditions in all executions is said to be \emph{error free}. 
 The quality of a BA protocol is measured primarily by using  three parameters: 
     \begin{itemize}
\item  Resilience:  the number of processors $n$ as a function of $t$ allowed.
\item   Round complexity:  the number of rounds of exchanging information, denoted by $\Rd$.
\item  Communication complexity:    the total number of communication bits, denoted by $\BT$.
\end{itemize} 
For any error-free BA protocol, the known lower bounds on those three parameters  are   respectively
 \begin{align}
 n\geq 3t + 1   \quad  \text{(cf.~\cite{PSL:80, LSP:82})} , \quad   
 \Rd \geq t + 1  \quad   \text{(cf.~\cite{FL:82, DS:83})}  , \quad 
 \BT \geq \Omega(\max\{n\Lh, nt\})  \quad  \text{(cf.~\cite{DR:85, FH:06})}    \non
  \end{align} 
where $\ell$ denotes the length of message being agreed upon.  
 It is worth mentioning that, in practice the consensus is often required for \emph{multi-valued} message rather than just single-bit message\cite{FH:06, LV:11,  GP:20, LDK:20, NRSVX:20,Patra:11, PR:11}.  
   For example,  in BFT consensus protocols of   Libra (or Diem) and  Hyperledger Fabric  Blockchains proposed by Facebook and IBM respectively, the message being agreed upon could be a transaction or transaction block with size scaled from  $1KB$ to $1MB$ \cite{Libra19,SBV:18, GLGK:19, Hyperledgerfabric20, Hyperledger18, Hyperledgerfabric18,hypfab:19}. 
Also, in practice the consensus  is often expected to be 100\% secure and error free in mission-critical applications such as online banking and smart contracts \cite{Ethereum15, HGBR:19, MJPS:19, CDP:19, DATF:19}.

{\renewcommand{\arraystretch}{1.2}
\begin{table}
\small
\begin{center}
\caption{Comparison of the proposed and some other error-free  synchronous BA protocols.} \label{tb:comp}
      \vspace{-.05 in}
\begin{tabular}{||c||c|c|c|c|c|}
\hline
Protocols & resilience &   communication complexity &  round complexity &     error free & signature free     \\ 
\hline
$\ell$-$1$-bit   & $  n   \geq 3t + 1 $ &  $\Omega(n^2 \ell)$     &  $\Omega(\ell t )$  &yes &yes  \\
\hline
 \cite{LV:11}   & $  n   \geq 3t + 1 $ & $O( n \ell + n^4  \sqrt{\ell}+n^6)$  &  $\Omega (\sqrt{\ell}+n^2)$  &  yes  & yes  \\
\hline
 \cite{GP:20}   & $  n   \geq 3t + 1 $ & $O( n \ell + n^4 )$  &  $O (t)$  &  yes  & yes  \\
\hline
 \cite{LDK:20}   & $  n   \geq 3t + 1 $ & $O( n \ell + n^4 )$   &  $O(t)$  &  yes  & yes  \\
\hline
 \cite{NRSVX:20}   & $  n   \geq 3t + 1 $ & $O(n \ell + n^3)$  &  $O(t)$  &  yes  & yes  \\
\hline
 {\color{blue}Proposed } &   {\color{blue}$  n   \geq 3t + 1 $}   & {\color{blue}$O(\max\{n\Lh, n t \log t \})$}  &   {\color{blue}$O(t)$}     &    {\color{blue}yes}  &  {\color{blue}yes} \\
\hline
\end{tabular}
      \vspace{-.2 in}
\end{center}
\end{table}
}
 
Achieving an \emph{error-free} \emph{multi-valued}  consensus  with   optimal resilience,  optimal round complexity and  optimal communication complexity simultaneously is widely believed to be the ``holy-grail'' in  the BA problem. 
The multi-valued BA problem of achieving consensus on an $\ell$-bit message could be solved by invoking $\ell$ instances of $1$-bit consensus in sequence, which is termed as $\ell$-$1$-bit scheme.   
However, this  scheme will result in communication complexity of $\Omega(n^2 \ell)$ bits, because $\Omega(n^2)$ is the lower bound on communication complexity  of $1$-bit consensus  given $n\geq 3t + 1$  \cite{CW:92, BGP:92, DR:85}.   
In 2006, Fitzi and Hirt  provided a \emph{probabilistically correct} multi-valued BA protocol by using a hashing technique, which results in communication complexity of $O(n\ell+n^3(n+\kappa))$ bits for some $\kappa$ serving as a security parameter \cite{FH:06}. 
This improves the communication complexity to $O(n\ell)$ when $\ell \geq n^3$, but at the cost of \emph{non-zero} error probability.  
In 2011, Liang and Vaidya  provided an \emph{error-free} BA protocol with communication complexity $O( n \ell + n^4  \sqrt{\ell}+n^6)$ bits, which is optimal when $\ell \geq n^6$ \cite{LV:11}. However, in the practical regime of $\ell < n^6$, this communication complexity is sub-optimal. Furthermore, the round complexity of the BA protocol in  \cite{LV:11} is  at least  $\Omega (\sqrt{\ell}+n^2)$,  which is sub-optimal. 
The result of  \cite{LV:11}  was  improved recently in  \cite{GP:20}, \cite{LDK:20} and  \cite{NRSVX:20}. 
Specifically, the communication complexities of the  BA protocols proposed in \cite{GP:20}, \cite{LDK:20} and \cite{NRSVX:20} are  $O( n \ell + n^4 )$ bits, $O( n \ell + n^4 )$ bits, and  $O( n \ell + n^3 )$ bits, respectively. 
Although the communication complexity has been improved in \cite{GP:20}, \cite{LDK:20} and \cite{NRSVX:20}, the achievable performance is still sub-optimal  in the practical regime of $\ell < n^2$. 
In some previous work,  randomized algorithms were proposed to reduce communication and round complexities but the termination   cannot be 100\% guaranteed \cite{FM:97, Patra:11, ADH:08, CR:93, CC:85, KK:09}. 
In some other work, the protocols were designed with  cryptographic technique such as signatures, hashing, authentication and secret sharing \cite{Rabin:83,KK:09, ADDNR:17,DS:83}. 
 However, the protocols with such cryptographic technique are vulnerable to attacks from the adversary with very high computation power, e.g.,  using supercomputer or quantum computer possibly available in the future, and hence not error free.   
 A protocol that  doesn't rely on cryptographic technique mentioned above is said to be \emph{signature free}. 
A protocol is said to be \emph{information-theoretic secure} if it is secure  against computationally unbounded adversary who takes full control over the dishonest $\Pros$.

In this work we show that  \emph{coding} is a promising approach for solving  the long-standing open problem in BA, i.e., achieving the fundamental limits of Byzantine agreement.  
Specifically by using coding theory, together with graph theory and linear algebra, we are able to design  an  error-free signature-free information-theoretic-secure multi-valued  BA  protocol (named as COOL) with optimal \emph{resilience},  asymptotically optimal \emph{round complexity}, and asymptotically optimal \emph{communication complexity}  when $\ell \geq t   \log t$, simultaneously (see Table~\ref{tb:comp}),  focusing on the BA setting with synchronous communication network. 
In a nutshell, carefully-crafted error correction codes  provide an efficient way of exchanging ``compressed'' information among distributed $\Pros$, while keeping the ability of detecting errors, masking errors, and making a consistent agreement at honest distributed $\Pros$. 
In this work the proposed protocol is developed first for  the Byzantine agreement problem. We show that the proposed protocol can also be generalized to the  Byzantine broadcast problem with the same performance as in the Byzantine agreement problem.

The main difference between the protocols  of  \cite{LV:11, GP:20, LDK:20, NRSVX:20} and our proposed protocol is that, while coding is also used, 
  in the protocols  of \cite{LV:11, GP:20, LDK:20, NRSVX:20} each distributed $\Pro$ needs to send an $n$-bit (graph) information to all other $\Pros$ after generating a graph  based on the exchanged information, which results in a communication complexity for this step of at least  $\Omega(n^3)$ bits or even $\Omega(n^4)$ bits if Byzantine broadcast is used.  This is a limitation in the  protocols  of   \cite{LV:11, GP:20, LDK:20, NRSVX:20}. Our proposed protocol, which is carefully designed by using coding theory, graph theory and linear algebra,  avoids the limitation appeared in  \cite{LV:11, GP:20, LDK:20, NRSVX:20}.  Specifically, our proposed protocol consists of at most four phases. After the third phase, it is guaranteed that at most one group of honest $\Pros$ output the same message, while the remaining honest $\Pros$  output an empty or default message.  In this way, all of the honest $\Pros$ can calibrate their messages and finally output a consistent and validated value.   We prove that our protocol  satisfies    termination, consistency and validity conditions. The proof borrows tools from coding theory, graph theory and linear algebra.
  
 \begin{figure}
\centering
\includegraphics[width=7cm]{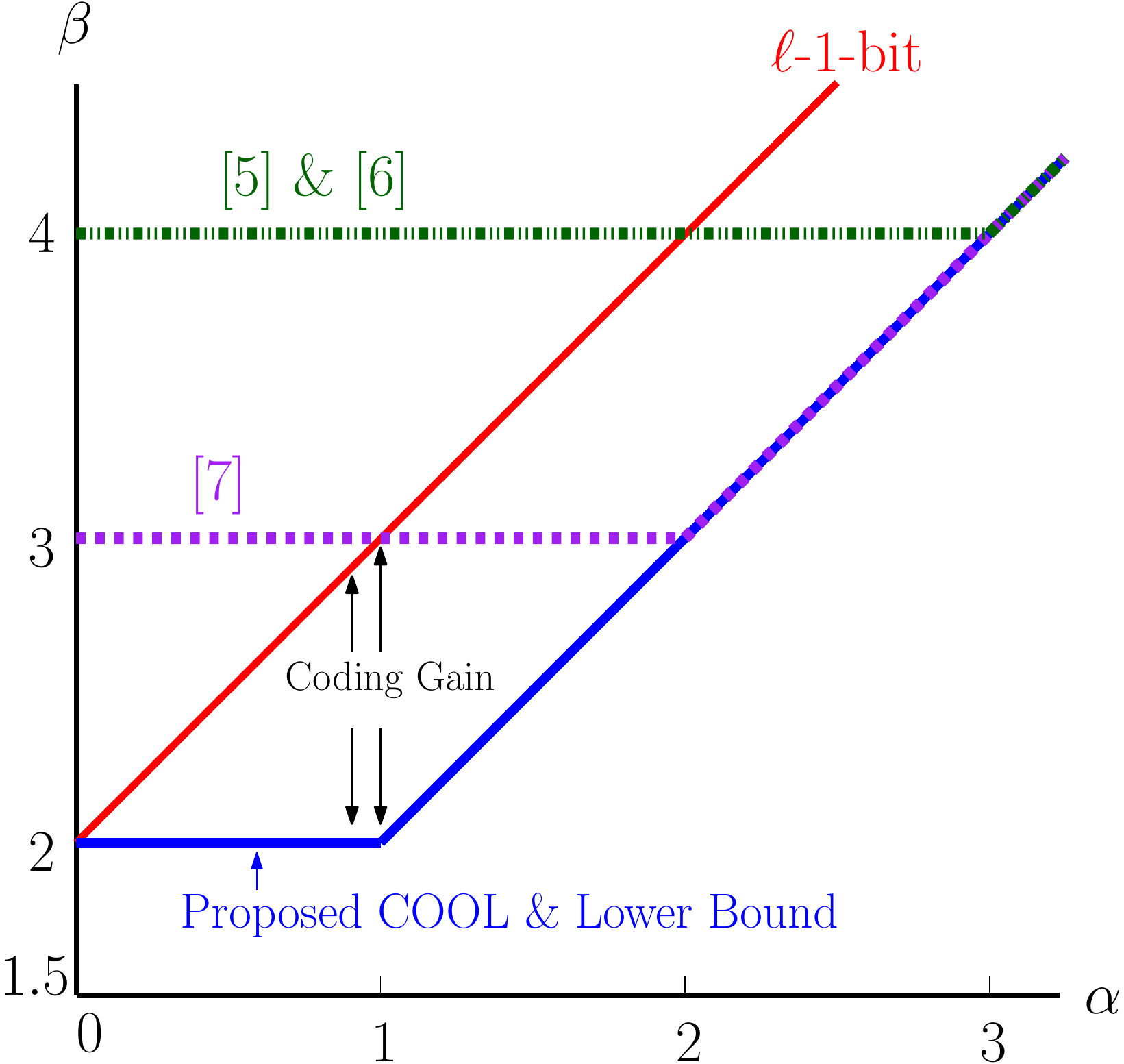}
\caption{Communication complexity exponent $\Ce$ vs. message size exponent $\Le$ of the proposed COOL, the protocols in \cite{GP:20, LDK:20, NRSVX:20}, and $\ell$-$1$-bit scheme, focusing on the case with $\Fe=1$. Note that for this case with $\Fe=1$, the communication complexity exponent of the protocol in \cite{LV:11} is $\Ce(\Le) =  \max\{1+  \Le, 4 +\Le/2, 6\}$, which is equal to $6$, $4 +\Le/2$ and $1+  \Le$ in the regimes of  $0\leq \Le \leq 4$, $4\leq \Le \leq 6$ and $  \Le \geq 6$, respectively, and hence outside the range of this figure. Compared to the protocol in \cite{LV:11} (resp. \cite{GP:20, LDK:20, NRSVX:20}), our proposed COOL provides an additive gain up to $4$ (resp. $2$, $2$, $1$), in terms of communication complexity exponent performance.   Compared to the $\ell$-$1$-bit scheme without using coding,  our proposed COOL provides  a  gain of  $1$ for any $\Le \geq 1$,  which can be considered as a coding gain resulted from carefully-crafted coding.  } 
\label{fig:com} 
      \vspace{-.1 in}
\end{figure}

Table~\ref{tb:comp} and Fig.~\ref{fig:com} provide some comparison  between our proposed protocol and some other  error-free  synchronous BA protocols. In Fig.~\ref{fig:com}, the comparison is focused on the communication complexity. 
For an error-free BA protocol, here we define the notions of \emph{communication complexity exponent}, \emph{message size exponent} and  \emph{faulty size exponent} as 
\begin{align}
  \Ce(\Le, \Fe)  \defeq \lim_{n \to \infty} \frac{\log \BT (n, \Fe, \Le)}{\log n}      \quad \quad
    \Le  \defeq  \lim_{n \to \infty} \frac{\log \ell}{\log n}        \quad \quad
  \Fe  \defeq  \lim_{n \to \infty} \frac{\log t}{\log n}    \label{eq:comexp}     
 \end{align}
  respectively. Intuitively, $\Ce$ (resp. $\Le$ and $\Fe$) captures the exponent of communication complexity $\BT$ (resp. message size $\ell$ and faulty size $t$) with $n$ as the base, when $n$ is large.   
In this work we characterize the \emph{optimal} communication complexity exponent $\Ce^{*}(\Le, \Fe)$ achievable by \emph{any} error-free BA protocol when $n\geq 3t+1$, that is, 
\begin{align}
   \Ce^*(\Le, \Fe) = \max\{1+  \Le, 1+  \Fe  \}  .  \label{eq:Ocomexp}     
 \end{align}
 As shown in Fig.~\ref{fig:com}, the proposed COOL achieves the \emph{optimal} communication complexity exponent. Compared to the protocols in \cite{LV:11, GP:20, LDK:20, NRSVX:20}, our proposed COOL provides  additive gains up to $4$, $2$, $2$, $1$, respectively, in terms of communication complexity exponent performance.   Compared to the $\ell$-$1$-bit scheme without coding,  our proposed COOL provides  a communication complexity exponent  gain of  $1$ for any $\Le \geq 1$, which can be considered as a coding gain resulted from carefully-crafted coding.

Coding has been used previously as an exciting approach in network communication and cooperative data exchange for improving throughput and tolerating attacks or failures \cite{LYC:03, KM:01, FLMP:07, CY:06, HLKMEK:04, JLKHKME:08, KHEA:10, KFM:04, YWRG:08,HM:05, HHYD:14, YH:15, YS:14, CH:16}. However, one common assumption in those previous works is that  the source of the data is always \emph{fault-free}, i.e., the source node is always honest or trustworthy, or the initial messages of the honest nodes are consistent and generated from the trustworthy source node. 
 This is very different from the BA and BB problems considered here, in which the leader (or the source node) can be dishonest and the original messages of the distributed nodes can be  controlled by the Byzantine adversary. One of the main difficulties of the BA and BB problems lies in the unknown knowledge about the leader (honest or dishonest)  and about initial messages (consistent or inconsistent).

 The remainder of this paper is organized as follows. 
Section~\ref{sec:system} describes the  system models.  
Section~\ref{sec:mainresult} provides the  main results of this work. 
Section~\ref{sec:code}  discusses some   advantages,   issues, and challenges of coding for BA.
The proposed COOL developed from coding theory, together with graph theory and linear algebra, is described in  Section~\ref{sec:COOL} for the BA setting. 
In Section~\ref{sec:COOLbb} the proposed COOL is shown to be  extended to the setting of Byzantine broadcast. 
The work is concluded in Section~\ref{sec:concl}.  
 Throughout this paper,  $|\bullet|$ denotes the magnitude of a scalar or the cardinality of a set. 
  $[m_1:  m_2 ]$ denotes the set of integers from $m_1$ to $m_2$, for some nonnegative integers $m_1 \le m_2$.  If  $m_1 > m_2$, $[m_1:  m_2 ]$ denotes an empty set.  
Logarithms are in base~$2$.   $\ceil*{m}$ denotes the least integer that is no less than $m$, and $\floor*{m}$ denotes the greatest integer that is no larger than $m$. 
$f(x)=O(g(x))$ implies that ${\lim\sup}_{x \to \infty} |f(x)|/g(x) < \infty$. 
$f(x)=\Omega(g(x))$ implies that ${\lim\inf}_{x \to \infty} f(x)/g(x) > 0$. 
$f(x)=\Theta(g(x))$  implies that $f(x)=O(g(x))$ and $f(x)=\Omega(g(x))$.

\section{System models  \label{sec:system} }

In the BA problem,  $n$ $\Pros$ want to reach agreement on an $\ell$-bit message (or value), but up to $t$ of the $\Pros$ are dishonest (or faulty).  
$\PRO$~$i$  holds an $\ell$-bit initial message $\Me_i$, $\forall i \in [1:n]$. 
To solve this BA problem, a designed protocol needs to satisfy the  \emph{termination},  \emph{consistency} and \emph{validity} conditions mentioned in the previous section.

We consider  the synchronous BA, in which every two $\Pros$ are connected via a reliable and private communication channel, and the messages sent on a channel  are guaranteed to reach to the destination on time.  The index of  each $\Pro$ is known to all other $\Pros$. 
We assume that a  Byzantine adversary takes full control over the dishonest $\Pros$ and  has complete knowledge of the state of the other $\Pros$, including the $\Lh$-bit initial messages. 

As mentioned,  a protocol that satisfies the  termination, consistency and validity conditions in all executions is said to be \emph{error free}. 
A protocol that  doesn't rely on the cryptographic technique such as signatures, hashing, authentication and secret sharing  is said to be \emph{signature free}. 
A protocol that  is secure (satisfying   the termination,   consistency  and  validity conditions) against computationally unbounded adversary is said to be \emph{information-theoretic secure}.

We also consider  the BB problem, where the honest $\Pros$ want to agree on the message sent by a leader who is potentially dishonest.   
To solve the BB problem, a designed protocol needs to satisfy the termination, consistency and validity conditions. 
In  the BB problem,  the termination and consistency conditions remain the same as that in the BA problem, while the validity condition is slightly different from that in  the BA problem.  Specifically, the validity condition of the BB problem requires that,  if  the leader is honest then all honest $\Pros$  should agree on the message sent by the leader.    
Other definitions follow similarly from that of the BA problem.

\section{Main results  \label{sec:mainresult}}

In this work  by using   coding theory, together with graph theory and linear algebra,  we  design a {\bf co}ded  BA  protoc{\bf ol}, termed as COOL,   that achieves consensus on an $\Lh$-bit message  with promising performance in resilience,  round complexity,  and communication complexity. This proposed COOL is also extended to the Byzantine broadcast setting.  
The main results of this work are summarized in the following theorems.  The proofs are provided in Sections~\ref{sec:COOL}-\ref{sec:COOLbb}.

\begin{theorem} [BA problem]  \label{thm:BAs}
The proposed COOL  is an error-free signature-free information-theoretic-secure multi-valued  BA  protocol that achieves the consensus on an $\Lh$-bit message with optimal resilience,  asymptotically optimal round complexity, and asymptotically  optimal communication complexity when $\ell \geq t   \log t$, simultaneously.  
\end{theorem}

The proposed COOL achieves the consensus on an $\Lh$-bit message with  resilience of  $n\geq  3t+1$ (optimal),   round complexity  of $O(t)$ rounds  (asymptotically optimal), and  communication complexity of  $O(\max\{n\Lh, n t \log t \})$  bits (asymptotically  optimal when $\ell \geq t   \log t$), simultaneously. 
The description of  COOL is provided in Section~\ref{sec:COOL}.
Note that, for any error-free BA protocol, the known lower bounds on resilience, round complexity and communication complexity  are   
$ 3t + 1$ (cf.~\cite{PSL:80, LSP:82}),  $t + 1$  (cf.~\cite{FL:82, DS:83}), and $\Omega(\max\{n\Lh, nt\})$ (cf.~\cite{DR:85, FH:06}), respectively.

\begin{theorem} [BB problem]  \label{thm:BBs}
The proposed adapted COOL is an error-free signature-free information-theoretic-secure multi-valued  BB  protocol that achieves the consensus on an $\Lh$-bit message  with optimal resilience,  asymptotically optimal round complexity, and asymptotically  optimal communication complexity  when $\ell \geq t   \log t$, simultaneously.  
\end{theorem}

The proposed COOL designed for the BA setting can be adapted into the BB setting, which achieves the consensus on an $\Lh$-bit message with  the same performance of  resilience, round complexity  and communication complexity as in the BA setting.  
 Note that  the known lower bounds on resilience, round complexity and communication complexity for error-free BA protocols, can also be applied to any error-free BB protocols. 
The description of  adapted COOL for the BB setting is provided in Section~\ref{sec:COOLbb}.

In this work, with the achievable performance by the proposed COOL and the known lower bounds, we characterize the \emph{optimal} communication complexity exponent $\Ce^{*}(\Le, \Fe)$ achievable by \emph{any} error-free BA (or BB) protocol when $n\geq 3t+1$, that is, 
\begin{align}
   \Ce^*(\Le, \Fe) = \max\{1+  \Le, 1+  \Fe  \}    \non
 \end{align}
 where the \emph{communication complexity exponent} $\Ce$, the \emph{message size exponent} $\Le$ and  the \emph{faulty size exponent} $\Fe$ are defined in \eqref{eq:comexp}. Note that, based on the know lower bound $b\geq \Omega(\max\{n\Lh, nt\})$,   the optimal communication complexity exponent is lower bounded by $\Ce^{*}(\Le, \Fe) \geq  \lim_{n \to \infty} \frac{\log \Omega(\max\{n\Lh, nt\}) }{\log n} = \max\{1+  \Le, 1+  \Fe  \}$.  This lower bound is achievable by the proposed COOL, as the communication complexity exponent of COOL, denoted by $\Ce^{[cool]}$, is $\Ce^{[cool]}(\Le, \Fe) =  \lim_{n \to \infty} \frac{\log O(\max\{n\Lh, n t \log t \}) }{\log n} = \max\{1+  \Le, 1+  \Fe  \}$. 
  As shown in Fig.~\ref{fig:com}, our proposed COOL provides certain gains in communication complexity exponent performance, compared to the protocols in \cite{LV:11, GP:20, LDK:20, NRSVX:20} and the $\ell$-$1$-bit scheme. 
Before describing the proposed COOL, let us first describe the advantages,   issues, and challenges of coding for BA in the following section.

\section{Coding for BA:   advantages,   issues, and challenges}   \label{sec:code}

Coding and information theory provide very elegant and powerful ways of detecting or correcting errors in possibly corrupted data via carefully constructed codes with certain properties, which have revolutionized the digital era. This work seeks to \emph{optimally} use these concepts and techniques to bring the \emph{maximal} benefits
in the field of Byzantine agreement.

\subsection{Advantages of using coding}

For the multi-valued Byzantine agreement problem, error correction codes can be utilized to achieve consensus for the \emph{entire string} of message but not bit by bit.
In order to see the advantages of coding for Byzantine agreement, let us discuss the following three schemes.  \\
$\bullet$ \emph{Scheme with full data transmission:}  For  the $\ell$-$1$-bit scheme with  bit-by-bit consensus  in sequence, since  distributed $\Pros$ need to exchange the whole message data,  it is not surprising that  this scheme results in a sub-optimal communication complexity (see Table~\ref{tb:comp}). \\ 
$\bullet$ \emph{Non-coding  scheme with reduced data transmission:}  To reduce the communication complexity,  one possible solution is to let distributed $\Pros$ exchange only one piece of data, instead of the whole data. 
However, if the data is \emph{uncoded}, this reduction in data exchange could possibly lead to a failure, i.e.,   not satisfying termination, consistency or validity conditions. This is because   honest $\Pros$ might not get enough information from others   and hence this protocol is more vulnerable to the attacks from dishonest $\Pros$. \\  
 $\bullet$ \emph{Coding  scheme with reduced data transmission:}   If the data is \emph{coded} with error correction code, the distributed $\Pros$ could possibly be able to detect and correct errors, even with reduced data transmission. Hence,   the protocol using  carefully-crafted  coding  scheme with reduced data transmission could be robust against attacks from dishonest $\Pros$ and could be error free. 
 Table~\ref{tb:coding}  provides some comparison between the   three schemes.

\begin{table}[h!] 
\small
\begin{center}
\caption{Comparison of three BA schemes. DT  stands for  data transmission.} \label{tb:coding}
      \vspace{-.1 in}
\begin{tabular}{||c||c|c|  }
\hline
Schemes & communication  complexity &   		   error free   \\ 
\hline
scheme with full DT      &  high    &      yes      \\
\hline
non-coding  scheme with reduced DT   &  low        & no  \\
\hline
  coding  scheme with reduced DT  &    {\color{blue}low} &         {\color{blue}yes}   \\
\hline
\end{tabular}
      \vspace{-.2 in}
\end{center}
\end{table}

\subsection{Error correction codes}   \label{sec:ecc}

The $(n, k)$ Reed-Solomon error correction code encodes $k$ data symbols from Galois Field $GF(2^c)$ into a codeword that is consisting of $n$ symbols from $GF(2^c)$, for  \[n \leq 2^c  -1\] (cf.~\cite{RS:60})\footnote{For the extended Reed-Solomon codes, the constraint can be relaxed to $n \leq 2^c +1$, and to $n \leq 2^c +2$ in some cases.}. 
We can use $c$ bits to represent  each symbol from $GF(2^c)$, which implies that a  vector consisting of $k$ symbols from $GF(2^c)$ can be represented using $kc$ bits of data. 
The error correction code can be constructed by \emph{Lagrange polynomial interpolation} \cite{rabin:89}.  This constructed code is  a variant of Reed-Solomon code with minimum distance $d=n-k+1$, which is optimal according to the Singleton bound. 
For  an input  vector $\xv \defeq [x_1,  x_2, \cdots, x_k ]^\T$  and  an output vector $\zv \defeq [z_1,  z_2, \cdots, z_n]^\T$ with  data symbols  from $GF(2^c)$,  the encoding  is described below
\begin{align}
  z_i =   \hv_i^\T  \xv    \quad    \quad i\in [1:n]  \label{eq:zidef}     
 \end{align}
where  
\begin{align}
 \hv_i \defeq  [h_{i,1},  h_{i,2}, \cdots, h_{i,k} ]^\T  \quad \text{and} \quad  h_{i,j} \defeq \prod_{\substack{p=1 \\ p \neq j} }^{k} \frac{i-p}{j-p} , \quad   i\in [1:n], \quad j\in [1:k].  \label{eq:zidefh}     
 \end{align}
 Note that when $k=1$, it becomes a non-coding scheme, with coefficients  set as $\hv_i =  h_{i,1} =1$ for   $i\in [1:n]$. 
An $(n, k)$ error correction code  can correct up to $\lfloor \frac{n-k}{2}\rfloor$ Byzantine errors, that is, 
up to $\lfloor \frac{n-k}{2}\rfloor$  Byzantine errors in $n$  symbol observations (see \eqref{eq:zidef}) can be corrected by applying some efficient decoding algorithms for Reed-Solomon code, such as, Berlekamp-Welch algorithm and Euclid's algorithm\cite{roth:06, Berlekamp:68, RS:60}.  
More generally,  an $(n, k)$ error correction code can correct up to  $f_c$ Byzantine errors and simultaneously detect up to $f_d$ Byzantine errors in $n'$ symbol observations if and only if the  conditions of $2f_c + f_d \leq   n'- k$  and  $n' \leq  n$ are satisfied.

\subsection{Coding  scheme with reduced DT: what could possibly go wrong?}

In the BA problem, if distributed $\Pros$ exchange information with reduced data transmission,  the honest $\Pros$ might not be able to get enough information from others  and might be vulnerable to attacks from dishonest $\Pros$. In this sub-section, we will describe some possible issues when using coding  scheme with reduced data transmission. In the next section, we will show how to handle the issues by carefully designing the  protocol.

\subsubsection{Attack example: violating validity condition} \label{sec:vvc}

Let us consider  an example  where  all   honest $\Pros$  hold the \emph{same} initial message $\Meg$, i.e., $\Me_i = \Meg, \forall i \notin \Fc$, where $\Fc$ denotes a set of dishonest $\Pros$, for some non-empty $\ell$-bit value of  $\Meg$. 
By using coding scheme, each  $\Pro$ could encode its initial message $\Me_i$ into some symbols with reduced size,  as in \eqref{eq:zidef}. In this way, the  $\Pros$ could exchange information with reduced data transmission.
For example, $\PRO$~$i$, $ i \in [1:n]$, could encode its initial message $\Me_i$ into some symbols as  $y_j^{(i)}  =    \hv_j^\T \Me_i$,  $j \in  [1:n]$, where $\hv_i$ is defined as in \eqref{eq:zidefh}.  The size of the symbol $y_j^{(i)}$ is $c=\ell/k$ bits, which is relatively small compared to the original message size, when $k$  is sufficiently large.  By sending a pair of  the $\ell/k$-bit information symbols $(y_j^{(i)}, y_i^{(i)})$, instead of the  $\ell$-bit message $\Me_i$,  from $\PRO$~$i$ to $\PRO$~$j$ for $j\neq  i$, the communication complexity could be significantly reduced.

However, if the protocol is not well designed, the honest $\Pros$ might not be able to make a correct consensus output that is supposed to be   $\Meg$ in this case under the validity condition.
For example, if $k$ is designed to be too large, then the honest $\Pros$ might not be able to correct the errors injected by the dishonest $\Pros$, and hence make a wrong consensus output, e.g., as a default value $\phi$, violating the validity condition (see Fig.~\ref{fig:attackv}).

\begin{figure}[h!]
\centering
\includegraphics[width=15cm]{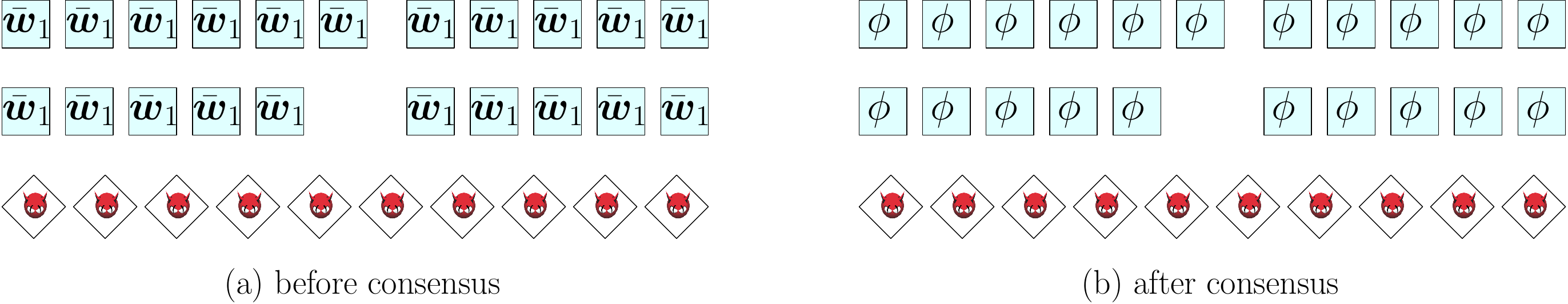}
      \vspace{-.1 in}
\caption{An attack example violating  the validity condition with $(t=10, n=31)$.  
The light-cyan  square nodes  refer to  the  honest processors, while the rest nodes refer to the dishonest processors.} 
\label{fig:attackv}
\end{figure}

\subsubsection{Attack example: violating consistency condition}   \label{sec:vcc}

 Let us  consider another example with three disjoint groups in $n$-$\Pro$ network:  Group $\Fc$,  Group $\Ac_{1}$ and Group $\Ac_{2}$. Group $\Fc$  is a set of $t$  dishonest $\Pros$.    Group $\Ac_{1}$ is a set  of  honest $\Pros$ holding the same initial message $\Meg_{1}$, i.e., $\Me_i = \Meg_{1}, \forall i \in \Ac_{1}$, while Group $\Ac_{2}$ is a set  of  honest $\Pros$ holding the same initial message $\Meg_{2}$, i.e., $\Me_{i'} = \Meg_{2}, \forall i' \in \Ac_{2}$,  given $|\Ac_{1}| =t+1$, $|\Ac_{2}| =t$, and for some  $\Meg_{1}$ and  $\Meg_{2}$, $\Meg_{1} \neq \Meg_{2}$  (see an example in Fig.~\ref{fig:attackc}).  
By using coding scheme,  $\PRO$~$i$, for $i \in  [1:n]$, could encode its initial $\ell$-bit message $\Me_i$ into $\ell/k$-bit symbols as  $y_j^{(i)}  =    \hv_j^\T \Me_i$,  $j \in  [1:n]$. 
Then $\PRO$~$i$ could send  a pair of ``compressed'' information symbols  $(y_j^{(i)}, y_i^{(i)})$ to $\PRO$~$j$,  for $j\neq  i$. 
By sending ``compressed'' information symbols to other $\Pros$, the communication complexity could be significantly reduced.  
However, if the protocol is not well designed,  the honest $\Pros$ might  make inconsistent consensus outputs.

To attack the protocol, $\PRO$~$i$, $i\in \Fc$,  could send out inconsistent information, e.g., sending a pair of symbols $(\hv_j^\T  \Meg_{1},  \hv_i^\T  \Meg_{1})$ to $\PRO$~$j$ for $j \in \Ac_{1}$ and sending another pair of symbols $(\hv_{j'}^\T  \Meg_{2}, \hv_i^\T  \Meg_{2})$ to $\PRO$~$j'$ for $j'\in \Ac_{2}$, respectively. 
In this scenario,  each $\Pro$ in Group $\Ac_{1}$ might make a consensus output as $\Meg_{1}$. This is because, for  $\PRO$~$j$, $j \in \Ac_{1}$, at least $2t+1$  observations of $(y_j^{(i)}, y_i^{(i)})$ received from $\PRO$~$i$, $i\in \Ac_{1}\cup \Fc$, are encoded from message $\Meg_{1}$. Hence,  each $\Pro$  in Group $\Ac_{1}$ has an illusion  that at least $2t+1$ $\Pros$ hold the same initial message $\Meg_{1}$. 
At the same time, some $\Pros$ in Group $\Ac_{2}$ might make a consensus output as $\Meg_{2}$. This is because,   for  $\PRO$~$j$, $j \in \Ac_{2}$, $2t$  observations of $(y_j^{(i)}, y_i^{(i)})$ received from Group $\Ac_{2}$  and Group $\Fc$ are encoded from message $\Meg_{2}$. Furthermore,  some observations of $(y_j^{(i)}, y_i^{(i)})$ received from Group $\Ac_{1}$ could possibly be expressed as $(y_j^{(i)}, y_i^{(i)}) = ( \hv_j^\T  \Meg_{2}, \hv_i^\T  \Meg_{2})$ when
\begin{align}
 \hv_i^\T  \Meg_{1}  =  \hv_i^\T  \Meg_{2} \quad    \text{and}  \quad \hv_j^\T  \Meg_{1}  =  \hv_j^\T  \Meg_{2},\quad   \text{for some} \  i \in  \Ac_{1},  j \in \Ac_{2} \label{eq:equal12}   
 \end{align}
 i.e., when $(\Meg_{1} - \Meg_{2})$ is in the null space of $\hv_i$ and $\hv_j$.   
With the condition in \eqref{eq:equal12}, some $\Pros$  in Group $\Ac_{2}$ have an illusion  that at least $2t+1$ $\Pros$ hold the same initial message $\Meg_{2}$.  At the end, some honest $\Pros$ might  make inconsistent consensus outputs   (see Fig.~\ref{fig:attackc}).

 \begin{figure}[h!]
\centering
\includegraphics[width=15cm]{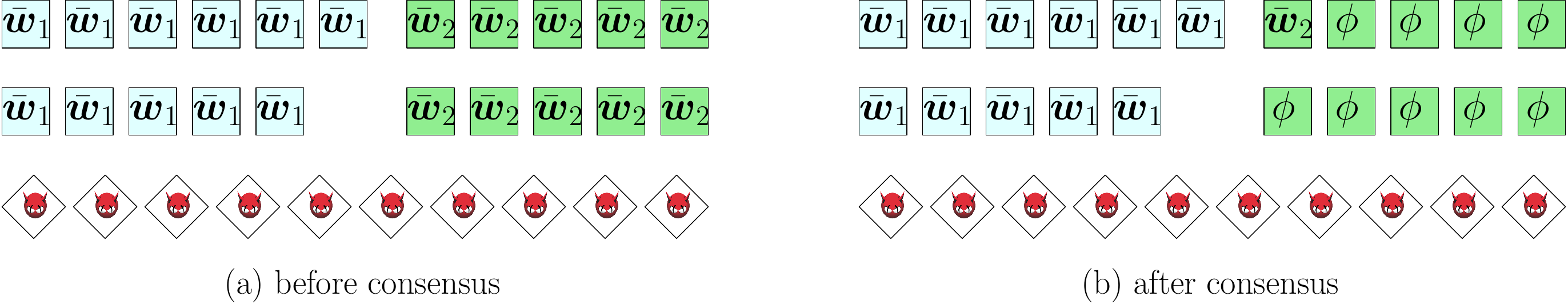}
      \vspace{-.1 in}
\caption{An attack example violating   consistency condition with $(t=10, n=31)$.  The cyan and green  square nodes  refer to  Group $\Ac_{1}$  and Group $\Ac_{2}$ of honest processors, respectively.} 
\label{fig:attackc}
\end{figure}

 \subsubsection{Attack example: violating termination condition}

For the coding scheme with reduced data transmission, if the protocol is not designed well, it might not be able to terminate.  For example, one possible termination condition at a $\Pro$ is that it needs to hear at least $2t+1$ ``ready'' signals from other $\Pros$, where the ``ready'' signal could be sent when a $\Pro$ is able to correct up to $t$ errors from the received observations. However, due to the inconsistent information injected from dishonest $\Pros$,  there might be some events  that the termination condition will never be satisfied and hence the protocol won't terminate.

 \subsection{Coding  scheme with reduced DT: what are the challenges?}

 For the coding  scheme with reduced date transmission, since $\Pros$  exchange  coded information with reduced size (cf.~\eqref{eq:zidef}), it might create some illusion at some  $\Pros$ as described in Section~\ref{sec:vcc}.  
 As shown in  \eqref{eq:equal12}, when $\PRO$~$i$, for some $i\in \Ac_{1}$, sends  a pair of coded information symbols $(\hv_j^\T  \Meg_{1}, \hv_i^\T  \Meg_{1})$ to $\PRO$~$j$, for some $j\in \Ac_{2}$, this pair of coded information symbols could possibly be expressed as $(\hv_j^\T  \Meg_{1}, \hv_i^\T  \Meg_{1})= (\hv_j^\T  \Meg_{2}, \hv_i^\T  \Meg_{2})$,  for some different $\Meg_{1}$ and $\Meg_{2}$, which creates an illusion at $\PRO$~$j$ that  $\PRO$~$i$ might initially hold the message $\Meg_{2}$, but not $\Meg_{1}$.  Removing such illusion is one of the challenges in designing the coding scheme with reduced date transmission. This is because, sending the coded information with reduced dimension, e.g.,  $(\hv_j^\T  \Meg_{1}, \hv_i^\T  \Meg_{1})$, could not reveal enough information about the initial message $\Meg_{1}$. To make it efficient and secure, i.e., satisfying termination, consistency and validity conditions, the coding scheme  needs to be carefully crafted.

Another challenge lies in a constraint of Reed-Solomon error correction code used in the coding  scheme with reduced date transmission.
The $(n, k)$ Reed-Solomon error correction code encodes $k$ data symbols from Galois Field $GF(2^c)$ into a codeword that is consisting of $n$ symbols from $GF(2^c)$, but under the constraint of  $n \leq 2^c  -1$, or equivalently $\log (n+1) \leq c$ (see Section~\ref{sec:ecc}).   
This constraint implies that each symbol  from $GF(2^c)$ is represented by at least $\log (n+1)$ bits.  
If every $\Pro$ sends one such symbol to all other $\Pros$, it will result in a communication complexity of at least $\Omega(n^2 \log n)$ bits.  
As it will be shown later on, this is the main reason that there is still a multiplicative  gap (within $\log t$) between the  communication complexity of  proposed protocol and the known lower bound, when $\Lh < t \log t$. When $\Lh \geq  t \log t$, the proposed protocol achieves the asymptotically optimal communication complexity of   $\Theta(  n \Lh )$ bits.  To close the gap in the regime of $\Lh < t \log t$, one might need to overcome the challenge incurred by the constraint $n \leq 2^c  -1$ of Reed-Solomon error correction code.

\section{COOL:   coded Byzantine agreement protocol}  \label{sec:COOL}

This section will describe   the proposed COOL: {\bf co}ded Byzantine agreement protoc{\bf ol}. 
The design of COOL is based on the coding  scheme with reduced data transmission, and is crystallized from coding theory, graph theory and linear algebra. 
In a nutshell, carefully-crafted error correction codes  provide an efficient way of exchanging ``compressed'' information among distributed nodes, while keeping the ability of detecting errors, masking errors, and making a consistent and validated agreement at honest distributed nodes.
To this end, it will be shown that the proposed COOL is an  error-free signature-free information-theoretic-secure multi-valued  BA  protocol that achieves the consensus on an $\Lh$-bit message with  resilience of  $n\geq  3t+1$ (optimal),   round complexity  of $O(t)$ rounds  (asymptotically optimal), and  communication complexity of  $O(\max\{n\Lh, n t \log t \})$  bits (asymptotically  optimal when $\ell \geq t   \log t$), simultaneously. This result will serve as the achievability proof of Theorem~\ref{thm:BAs}.

For this BA problem, $\PRO$~$i$, $i \in [1:n]$, initially has the $\ell$-bit input message $\Me_i$. 
At first  the parameters $k$  and $c$  are designed as 
 \begin{align}
k \defeq    \Bigl \lfloor   \frac{ t  }{5 } \Bigr\rfloor    +1, \quad  \cb \defeq   \Bigl \lceil \frac{ \max\{ \ell, \ (t/5 +1) \cdot \log (n+1) \} }{k} \Bigr\rceil  .   \label{eq:qdef} 
\end{align}  
As will be shown later, our design of $k$ and $c$ as above is one of the key elements in the proposal protocol, which guarantees that the proposed protocol satisfies  the  termination, consistency and validity conditions. 
With the above values of $k$ and $c$, it holds true that the  condition of the Reed-Solomon error correction code, i.e.,   $ n \leq 2^c -1$, is satisfied (see Section~\ref{sec:ecc}). 
The $(n, k)$ Reed-Solomon error correction code  will be used to encode the $\ell$-bit initial message $\Me_i$, $\forall i \in [1:n]$.   When $\ell$ is less than $kc$ bits,  the $\ell$-bit  message $\Me_i$ will be first extended to a $kc$-bit data by adding $(kc - \ell)$ bits of redundant zeros (zero padding). 
The proposal COOL will work as long as $t \leq \frac{n-1}{3}$. 
In the description of the proposed COOL,  $t$ is considered such that $t \leq \frac{n-1}{3}$ and $t=\Omega(n)$. Later on we will discuss the case when $t$ is relatively small compared to $n$. 

The proposed COOL consists of at most four phases (see Fig.~\ref{fig:cool}), which are described in the following sub-sections. 
The proposed COOL is also described in Algorithm~\ref{algm:BAs}. 
Two examples of COOL are provided in Fig.~\ref{fig:coolex} and Fig.~\ref{fig:coolexv}, which   address the issues described in Section~\ref{sec:vcc} and Section~\ref{sec:vvc}, respectively.  Some notations of COOL are summarized in Table~\ref{tb:notation}.

 \begin{figure}
\centering
\includegraphics[width=11cm]{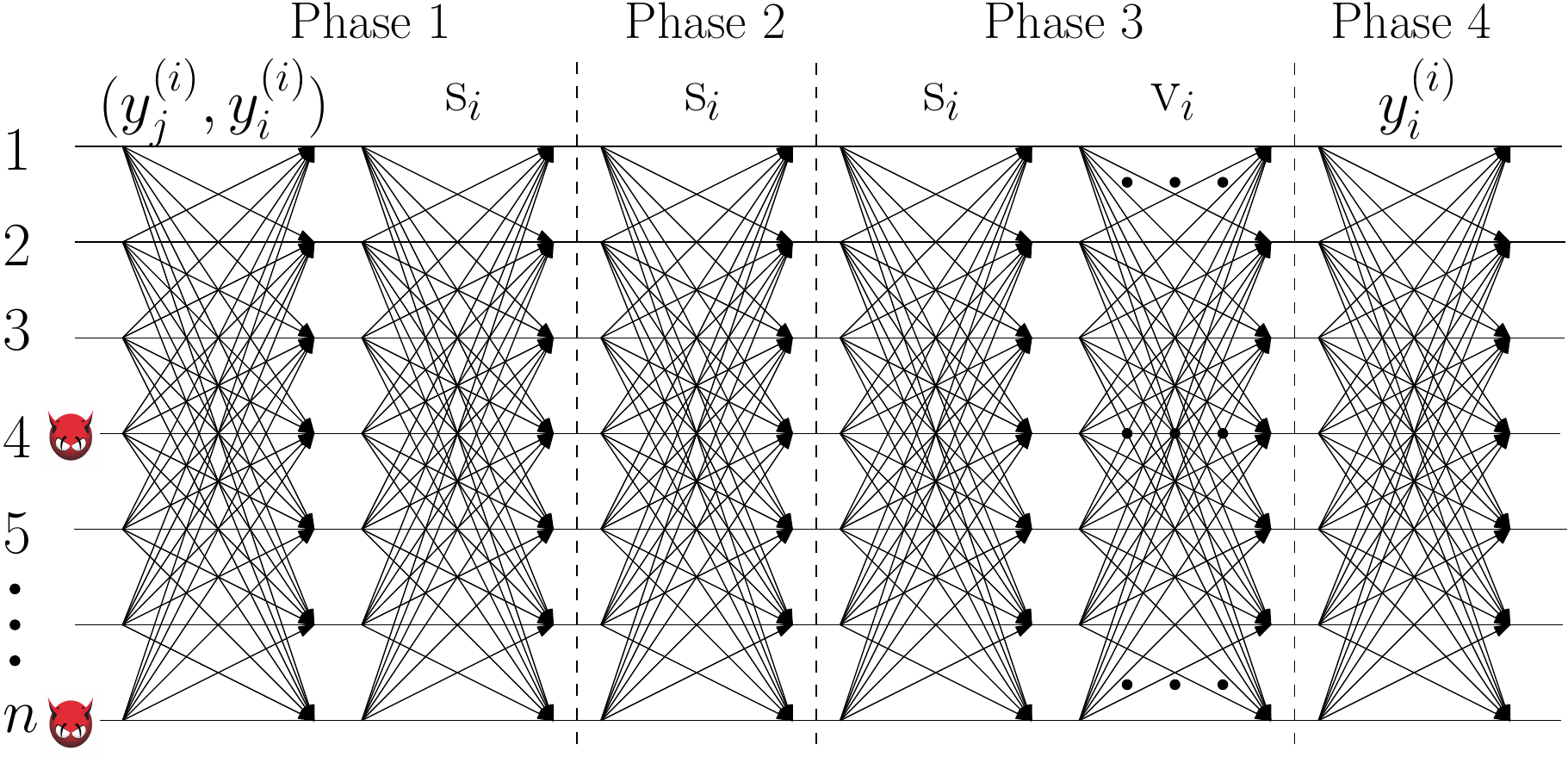}
\caption{Four-phase operation of the proposed COOL. In each phase, the distributed $\Pros$ exchange some kinds of information symbols (see the symbols above the first line) that are defined later and summarized in Table~\ref{tb:notation}.}  
\label{fig:cool}
\end{figure}

Let us first define  $\Me^{(i)}$ as the updated (or decoded) message at $\PRO$~$i$, $i \in [1:n]$.  $\Me^{(i)}$ can be updated via decoding, or via comparing its own information and the obtained information. 
In this proposed protocol, the decoding is required at Phase~4 only. 
The  value of $\Me^{(i)}$ is initially set as $\Me^{(i)} = \Me_{i}$, $i \in [1:n]$.

\subsection{Phase 1: exchange compressed  information  and update message \label{sec:Ph1} }

Phase 1 has three steps. The idea is to exchange ``compressed''  information and learn it.    
   
\emph{1) Exchange compressed  information:} $\PRO$~$i$, $i \in [1:n]$,  first encodes its $\ell$-bit initial message $\Me_{i}$ into $\ell/k$-bit symbols as    \begin{align}
y_j^{(i)}  \defeq    \hv_j^\T\Me_{i}, \quad  j \in  [1:n]    \label{eq:yi11} 
\end{align}
where $\hv_i$ is defined as in \eqref{eq:zidefh}. 
Then, $\PRO$~$i$, $i \in [1:n]$, sends a pair of coded symbols $(y_j^{(i)}, y_i^{(i)})$ to $\PRO$~$j$ for $ j \in  [1:n],  j\neq i$.

\emph{2) Update information:}  
 $\PRO$~$i$, $i \in [1:n]$, compares the observation $(y_i^{(j)}, y_j^{(j)})$ received from $\PRO$~$j$ with its  observation $(y_i^{(i)}, y_j^{(i)})$ and sets a binary indicator for the link between $\PRO$~$i$ and $\PRO$~$j$, denoted by $\Lk_i(j)$, as 
\begin{numcases}  
{ \Lk_i (j)=} 
     1   &    if   \    $ (y_i^{(j)}, y_j^{(j)}) = (y_i^{(i)}, y_j^{(i)})$           			 \label{eq:lkindicator}  \\
  0  &  else            		\non  
\end{numcases}
 for $ j \in  [1:n]$.  $\Lk_i (j)$ can be considered as a \emph{link indicator} for $\PRO$~$i$ and $\PRO$~$j$.  The value of $\Lk_i (j)=0$ reveals that $\PRO$~$i$ and $\PRO$~$j$ have mismatched messages, i.e.,  $ \Me^{(i)} \neq \Me^{(j)}$. However, the value of $\Lk_i (j)=1$ does \emph{not} mean that  $\PRO$~$i$ and $\PRO$~$j$ have matched messages; it just means that $\PRO$~$i$ and $\PRO$~$j$ share a common information at a certain degree, i.e.,  $ (y_i^{(j)}, y_j^{(j)}) = (y_i^{(i)}, y_j^{(i)})$.  
 When $\Lk_i (j)=0$,  the  observation of $(y_i^{(j)}, y_j^{(j)})$ received from $\PRO$~$j$ is considered as a \emph{mismatched observation} at  $\PRO$~$i$.
 In this step $\PRO$~$i$, $i \in [1:n]$,  checks if its own initial message   successfully matches the majority of  other $\Pros$' initial messages, by counting the number of mismatched observations. Specifically, $\PRO$~$i$ sets a binary \emph{success indicator}, denoted by $\Ry_i$, as 
\begin{numcases}  
{ \Ry_i =} 
     1   &    if   \    $ \sum_{j=1}^n \Lk_i (j)  \geq n - t$           			 \label{eq:sindicator}  \\
  0  &  else   .         		\non  
\end{numcases}
The event of $\Ry_i =0$ means that the number of mismatched observations is more than $t$, which implies that the initial message of $\PRO$~$i$ doesn't  match the majority of  other $\Pros$' initial messages. 
If  $\Ry_i =0$,    $\PRO$~$i$ updates the message as  $\Me^{(i)} =\phi$  (a default value), else keeps the original value of $\Me^{(i)}$.

\emph{3) Exchange success indicators:}  
 $\PRO$~$i$, $i \in [1:n]$,  sends the binary value of  success indicator $\Ry_i$ to  all other $\Pros$. 
Based on the received success indicators  $\{\Ry_{i}\}_{i=1}^n$, each $\Pro$ creates the following two sets: 
   \begin{align}
\Ss_1  \defeq   \{ i:   \Ry_{i}=1,   i \in [1:n ]\} ,\quad 
\Ss_0    \defeq   \{ i:  \Ry_{i}=0,   i \in [1:n ]\} .  \label{eq:vr0} 
 \end{align}
Note that different $\Pros$ might have different views on $\Ss_1 $ and $\Ss_0$, due to the inconsistent information possibly
sent from dishonest $\Pros$.

   \begin{remark}  \label{rk:ph1a} 
  \emph{Since $y_j^{(i)}$ defined in \eqref{eq:yi11} has only $\cb$ bits, the total communication complexity among the network for the first step of Phase~1, denoted by $\Bit_1$,  is $\Bit_1 = 2 \cb n (n-1)$ bits. 
  If  $\PRO$~$i$, $i \in [1:n]$,  sends the whole message $\Me_i$ to other $\Pros$, then the total communication complexity  would be $\ell n (n-1)$ bits. 
 Compared to the whole message $\Me_i$,  the value of $y_j^{(i)}$ can be considered as a compressed  information.  By exchanging compressed  information, instead of whole messages,  the communication complexity is significantly reduced in this step.
 Since the success indicator $\Ry_i$ has only $1$ bit,   the total communication complexity among the network  for the third step of Phase~1, denoted by $\Bit_2$,  is $\Bit_2 =  n (n-1)$ bits.}  
\end{remark}

\begin{remark}
\emph{In Phase~1, exchanging ``compressed''  information reduces the communication complexity, however,  it also creates some potential issues due to the lack of full  original information.  
As shown in Fig.~\ref{fig:coolex}, each dishonest $\Pro$ could send inconsistent information to two different groups of honest $\Pros$, that is,   
sending a pair of symbols $(\hv_j^\T  \Meg_{1},  \hv_i^\T  \Meg_{1})$ to $\PRO$~$j$ for $j \in \Ac_{1}$ and sending another pair of symbols $(\hv_{j'}^\T  \Meg_{2}, \hv_i^\T  \Meg_{2})$ to $\PRO$~$j'$ for $j'\in \Ac_{2}$, respectively, for  $i\in \Fc$. 
 Due to this inconsistent information, together with the condition of $\hv_1^\T  \Meg_{1}  =  \hv_1^\T  \Meg_{2}$ and $\hv_{12}^\T  \Meg_{1}  =  \hv_{12}^\T  \Meg_{2}$, the honest $\Pros$ from different groups consequently have different updated messages, which  might  lead to inconsistent consensus outputs as described in Section~\ref{sec:vcc}. 
In Phase~1 $\PRO$~$i$, $i\in [13:21] \subset \Ac_{2}$, sets $\Lk_i (j) =0, \forall j \in [1:11]$ and $\Ry_i =0$, from which it identifies that its initial message doesn't  match the majority of  other $\Pros$' initial messages and then updates its message  as $\Me^{(i)}=\phi$. 
However, $\PRO$~$12$ still thinks that its initial message does  match the majority of  other $\Pros$' initial messages because of the condition $\hv_1^\T  \Meg_{1}  =  \hv_1^\T  \Meg_{2}$ and $\hv_{12}^\T  \Meg_{1}  =  \hv_{12}^\T  \Meg_{2}$.  This condition implies  that $ (y_{12}^{(j)}, y_j^{(j)}) = (y_{12}^{(12)}, y_j^{(12)})$ for any $j \in \{1\}\cup\Ac_{2}$ from the view of $\PRO$~$12$, and hence results in a ``wrong'' (i.e., mismatched)  output of $\Ry_{12} =1$ at  $\PRO$~$12$. 
In the next phases the effort is to detect  errors (mismatched information),  mask errors, and identify ``trusted'' information.} 
\end{remark}

 \begin{figure}[t!]
\centering
\includegraphics[width=15cm]{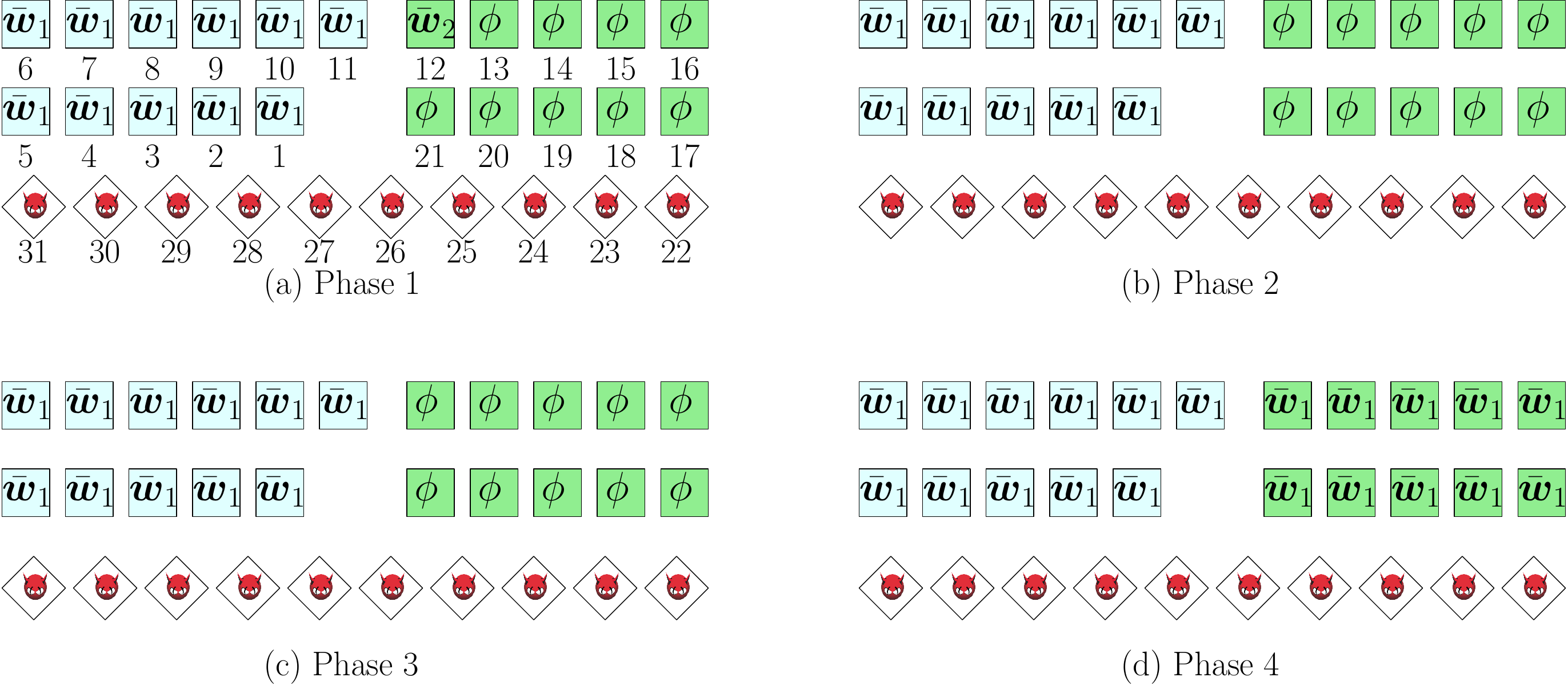}
 \caption{Illustration of  COOL addressing the issue described in Section~\ref{sec:vcc}, with $(t=10, n=31)$. The number under the node indicates the identity. The value inside the $i$th node denotes the value of updated message $\Me^{(i)}$ at the end of the corresponding phase.  Here $\Ac_{1} = [1:11]$, $\Ac_{2} = [12:21]$,  $\Fc = [22:31]$, $\Me_i = \Meg_{1}, \forall i \in \Ac_{1}$ and $\Me_{i'} = \Meg_{2}, \forall i' \in \Ac_{2}$.  It is assumed that $\hv_1^\T  \Meg_{1}  =  \hv_1^\T  \Meg_{2}$ and $\hv_{12}^\T  \Meg_{1}  =  \hv_{12}^\T  \Meg_{2}$, for $\Meg_{1} \neq \Meg_{2}$. 
   All honest  $\Pros$ eventually make the same consensus output.
} 
\label{fig:coolex}
\end{figure}

\subsection{Phase 2: mask errors, and update success indicator and message  \label{sec:Ph2} }

Phase 2 has three steps. The goal is to mask errors from the honest $\Pros$.

\emph{1) Mask errors identified in the previous phase:} $\PRO$~$i$, $i  \in \Ss_1$,  sets  
 \begin{align}
\Lk_i (j) = 0 ,   \   \forall j \in \Ss_0.   \label{eq:Thetaset2826} 
 \end{align}

\emph{2) Update and send success indicator:} $\PRO$~$i$, $i  \in \Ss_1$,  updates $\Ry_i$ as in \eqref{eq:sindicator} using updated values of $\{\Lk_i (1), \cdots, \Lk_i (n)\}$.   If the updated value of  success indicator is $\Ry_i =0$,  then $\PRO$~$i$ sends  the updated success indicator of $\Ry_i =0$ to others and  updates the message as  $\Me^{(i)} =\phi$.

\emph{3) Update $\Ss_1$ and $\Ss_0$:} $\PRO$~$i$, $i \in [1:n] $, updates the sets of $\Ss_1$ and $\Ss_0$ as in \eqref{eq:vr0} based on the newly received success indicators  $\{\Ry_{i}\}_{i=1}^n$.

  \begin{remark}   \label{rk:ph2a}
 \emph{Since the success indicator $\Ry_i$ has only $1$ bit,   the total communication complexity  for the second step of Phase~2, denoted by $\Bit_3$,  is bounded by $\Bit_3 \leq   n (n-1)$ bits.}  
 \end{remark}
 
  \begin{remark}
\emph{The idea of  Phase~2 is to mask errors from  honest  $\Pros$ whose initial messages don't  match the majority of  other $\Pros$' initial messages, but could not be detected out in   Phase~1 due to the condition as in \eqref{eq:equal12}.  
For the example in Fig.~\ref{fig:coolex}, at the second step of Phase 2, $\PRO$~$12$   updates the message as $\Me^{(12)}=\phi$. This is because at the second step of Phase 2, from the view of $\PRO$~$12$, the number of mismatched observations is at least $19$,  since $\Lk_{12} (j) =0, \forall j \in[2:11] \cup [13:21]$ based on the updated information in \eqref{eq:Thetaset2826}.}  
\end{remark}    

\subsection{Phase 3: mask errors, update information, and vote  \label{sec:Ph3}}

Phase 3 has five steps. The goal is to mask the rest of errors from the honest $\Pros$, and then vote for going to next phase or stop in this phase.

\emph{1) Mask errors identified in the previous phase:} $\PRO$~$i$, $i  \in \Ss_1$,  sets  
 \begin{align}
\Lk_i (j) = 0 ,   \   \forall j \in \Ss_0.   \label{eq:Thetaset2826P3} 
 \end{align}

\emph{2) Update and send success indicator:} $\PRO$~$i$, $i  \in \Ss_1$,  updates $\Ry_i$ as in \eqref{eq:sindicator} using updated values of $\{\Lk_i (1), \cdots, \Lk_i (n)\}$.   If the updated value of  success indicator is $\Ry_i =0$,  then $\PRO$~$i$ sends  the updated success indicator of $\Ry_i =0$ to others and  updates the message as  $\Me^{(i)} =\phi$.

\emph{3) Update $\Ss_1$ and $\Ss_0$:} $\PRO$~$i$, $i \in [1:n] $, updates the sets of $\Ss_1$ and $\Ss_0$ as in \eqref{eq:vr0} based on the newly received success indicators  $\{\Ry_{i}\}_{i=1}^n$.

\emph{4) Vote:} $\PRO$~$i$, $i \in [1:n]$, sets a binary vote as 
\begin{numcases}  
{ \Vr_i  =} 
     1   &    if   \    $ \sum_{j=1}^n \Ry_j  \geq  2t+1$           			\label{eq:vindicator}  \\
  0  &  else   .         		\non  
\end{numcases}
The value of $\Vr_i  =1$ means that  $\PRO$~$i$ receives greater than or equal to  $2t+1$ number of ones from $n$ success indicators $\{ \Ry_1, \Ry_2, \cdots, \Ry_n \}$.  
 The  indicator $\Vr_i$ can be considered as a vote from  $\PRO$~$i$  for going to next phase or stopping in this phase.      
          
\emph{5) One-bit consensus  on the $n$ votes:} In this step the system runs one-bit consensus  on the $n$ votes $\{\Vr_1, \Vr_2, \cdots, \Vr_n\}$  from all $\Pros$.
By using the one-bit consensus from \cite{BGP:92,CW:92}, the correct consensus can be made by using $O(nt)$ bits of   communication complexity, and $O(t)$ rounds of round complexity, for $t<n/3$. 
If the consensus of the votes $\{\Vr_1, \Vr_2, \cdots, \Vr_n\}$ is $1$, then every honest $\Pro$  \emph{goes to next phase}, else every honest $\Pro$   sets $\Me^{(i)} =\phi$ and considers it as a final consensus   and \emph{stops here}. 

 \begin{remark}  \label{rk:ph3a}
\emph{Since $\Ry_i$  has only $1$ bit,  the total communication complexity  for the second step of Phase~3, denoted by $\Bit_4$,  is bounded by $\Bit_4 \leq  n (n-1)$ bits.} 
\end{remark}

  \begin{remark}   \label{rk:ph3b}
\emph{The one-bit consensus  from \cite{BGP:92,CW:92} ensures that: (a) every honest $\Pro$  eventually outputs a consensus value; (b) all  honest $\Pros$ should have the same consensus output; (c) if  all honest $\Pros$ have the same vote then the consensus should be the same as the vote of the honest $\Pros$. 
Since we run the  one-bit consensus  from \cite{BGP:92,CW:92}, the total communication complexity  for the last step of Phase~3, denoted by $\Bit_5$,  is $\Bit_5  =O( n t)$ bits,    
while the round complexity of this step is $O(t)$ rounds, which dominates the total round complexity of the proposed protocol.}     
\end{remark}
 
 \begin{remark}
\emph{The goal of the first two steps of  Phase~3 is to mask the remaining errors, if any, from  honest  $\Pros$ whose initial messages don't  match the majority of  other $\Pros$' initial messages, but could not be detected out in the previous phases. As it will be shown in the protocol analysis in the next section,  these two steps, together with the  steps in the previous phases,  guarantee that at the end of Phase~3 there exists \emph{at most}  1  group of honest $\Pros$, where the honest $\Pros$ within this group  have the same  \emph{non-empty} updated  message, and  the  honest $\Pros$ outside this group  have  the same  \emph{empty} updated  message (i.e., equal to the default value $\phi$).}  
\end{remark}

\subsection{Phase 4: identify trusted information and make consensus    \label{sec:Ph4}}

This phase is taken place only when one-bit consensus of the votes $\{\Vr_1, \Vr_2, \cdots, \Vr_n\}$ is $1$.     This phase has four steps. 
     
 \emph{1) Update information with majority rule:} $\PRO$~$i$, $i \in \Ss_0$, updates the value of $y_i^{(i)}$ as 
 \begin{align}
y_i^{(i)}    \leftarrow  \text{Majority}( \{y_i^{(j)}:   j \in \Ss_1\})    \label{eq:yiph4} 
 \end{align}
 where the symbols of $y_i^{(j)}$ were received in Phase~1. $\text{Majority}(\bullet)$ is a function that returns the most frequent value in the list, based on majority rule. For example, $\text{Majority}(1,2,2) = 2$.

\emph{2) Broadcast updated information:}  $\PRO$~$i$,  $i \in \Ss_0$,  sends the updated value of $y_i^{(i)}$ to $\PRO$~$j$, $\forall j \in\Ss_0, j\neq i$.

\emph{3)  Decode the message:} $\PRO$~$i$, $i \in \Ss_0$, decodes its message using the  observations $\{y_1^{(1)}, y_2^{(2)}, \cdots, y_n^{(n)}\}$, where  $\{y_j^{(j)}: j \in \Ss_0\}$  were updated and received in this phase and  $\{y_j^{(j)}: j \in \Ss_1\}$ were received in the Phase~1.  The value of $\Me^{(i)}$  is updated as the decoded  message at $\PRO$~$i$, $i \in \Ss_0$. For $\PRO$~$i$, $i \in \Ss_1$, it keeps the original value of $\Me^{(i)}$.

\emph{4) Stop:} $\PRO$~$i$,  $i \in [1:n] $,  outputs consensus   as the message $\Me^{(i)}$   and stops.

 \begin{figure}[t!]
\centering
\includegraphics[width=15cm]{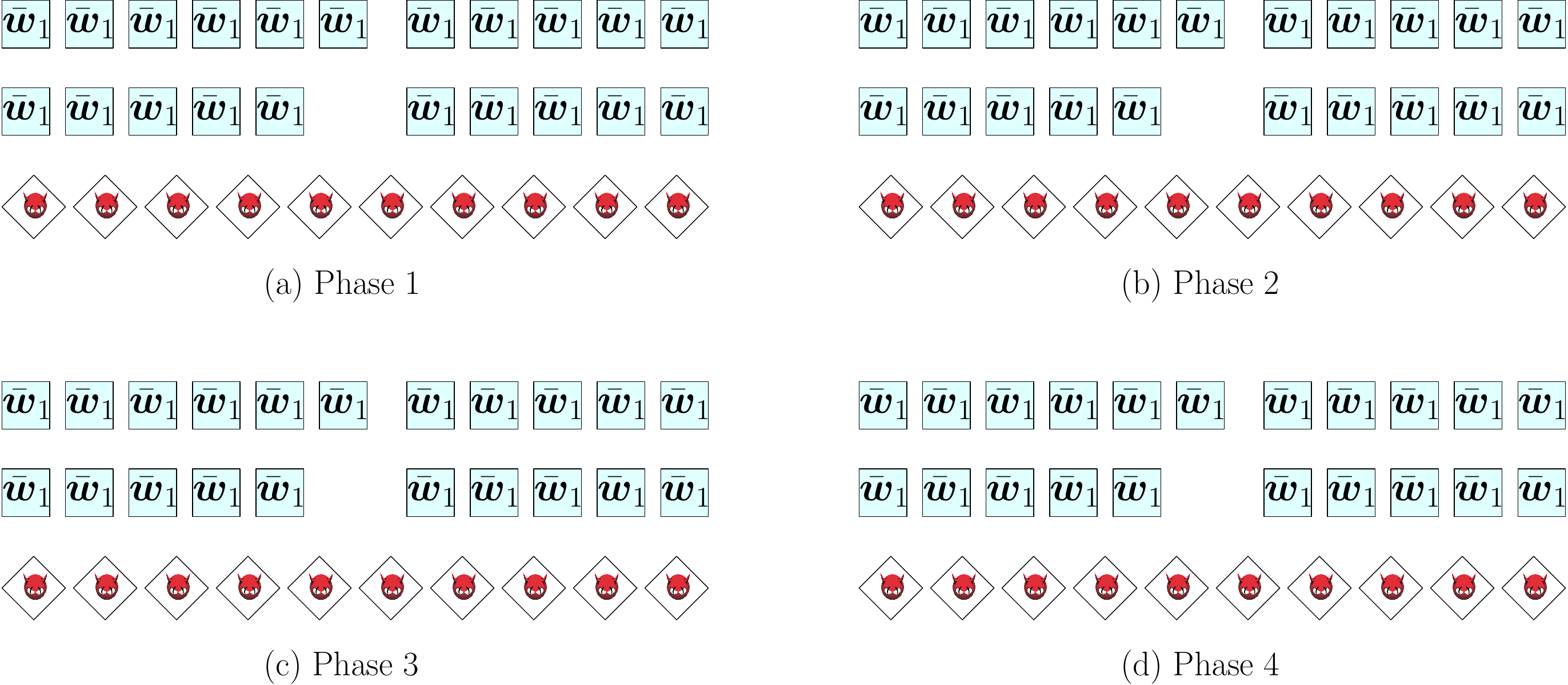}
 \caption{Illustration of  COOL addressing the issue described in Section~\ref{sec:vvc}, with $(t=10, n=31)$. The light-cyan  square nodes  refer to   the honest $\Pros$, while the rest nodes refer to the dishonest processors. In this example, all honest $\Pros$  have the same initial message as  $ \Meg_{1}$.  In each phase, the value inside each honest node denotes the value of its updated message.  In each phase  all honest  $\Pros$  output the same updated message as $ \Meg_{1}$, no matter what information is sent from the dishonest $\Pros$. Eventually all honest  $\Pros$  make the same consensus output, which is validated.
} 
\label{fig:coolexv}
\end{figure}

  \begin{remark}  \label{rk:ph4a}
   \emph{ Since $y_i^{(i)}$ has only $\cb$ bits, the total communication complexity  for the second step of Phase~4, denoted by $\Bit_6$,  is upper bounded by $\Bit_6 \leq  \cb n (n-1)$ bits.} 
\end{remark}

  \begin{remark}
  \emph{The idea of Phase~4 is to identify ``trusted'' information from the set of honest $\Pros$ whose initial messages  match the majority of  other $\Pros$' initial messages. In this way, the set of honest  $\Pros$ whose initial messages don't  match the majority of  other $\Pros$' initial messages could calibrate and update their information. 
 In some scenarios, the protocol could terminate at Phase~3, depending on the information sent from dishonest $\Pros$.  In any scenario,  all honest  $\Pros$ eventually make the same consensus output.  For the example in Fig.~\ref{fig:coolex},  since  the size of $\Ac_{1}$ is bigger than the size of $\Fc$, it guarantees that  Group $\Ac_{2}$ could calibrate and update their information successfully in the first step of Phase~4, eventually making the same consensus output as Group $\Ac_{1}$.}
\end{remark}

  \begin{remark}
  \emph{Fig.~\ref{fig:coolexv} depicts another example of COOL,  addressing the issue described in Section~\ref{sec:vvc}. 
 In this example, all honest $\Pros$  have the same initial message as  $ \Meg_{1}$.  In each phase  all honest  $\Pros$  output the same  message as $ \Meg_{1}$, no matter what information sent by the dishonest $\Pros$. All honest  $\Pros$  eventually make the consistent and validated consensus outputs. 
 }
   \end{remark}

\subsection{Provable performance of COOL}

For the proposed COOL, the provable performance is summarized in the following lemmas.

\begin{lemma}    \label{lm:comr}
 COOL achieves the consensus on an $\Lh$-bit message with the resilience of  $n\geq  3t+1$,  the round complexity  of $O(t)$ rounds, and  the communication complexity of  $O(\max\{n\Lh, n t \log t \})$  bits. 
\end{lemma}
\begin{proof}
The proposed COOL achieves the consensus on an $\Lh$-bit message as long as   $n\geq  3t+1$ (see Lemma~\ref{lm:ef}). 
 For the proposed COOL, the total communication complexity, denoted by $\BT$, is computed as 
   \begin{align*}
\BT  &= \sum_{i=1}^6 \Bit_i \\
&= O( \cb n (n-1) + n^2 )  \\
&= O\Bigl( \Bigl \lceil \frac{ \max\{ \ell, \ (t/5 +1) \cdot \log (n+1) \} }{ \bigl \lfloor   \frac{ t  }{5 } \bigr\rfloor    +1} \Bigr\rceil \cdot  n (n-1) + n^2 \Bigr)  \\
&= O(  \max\{ \ell/t,  \log n \}  \cdot n (n-1) + n^2 )   \\
& =O(  \max\{ \ell n^2/t ,  n^2 \log n \}  )  \quad \text{bits}  
 \end{align*}
where  $ \Bit_1,  \Bit_2, \cdots,  \Bit_6$ are expressed in Remarks~\ref{rk:ph1a}, \ref{rk:ph2a}, \ref{rk:ph3a}, \ref{rk:ph3b} and \ref{rk:ph4a}, and  the parameters $\cb$ and   $k$ are defined in  \eqref{eq:qdef}, i.e., $k =    \bigl \lfloor   \frac{ t  }{5 } \bigr\rfloor    +1$ and  $\cb =   \bigl \lceil \frac{ \max\{ \ell, \ (t/5 +1) \cdot \log (n+1) \} }{k} \bigr\rceil$.
In the description of the proposed protocol,  $t$ is considered such that $t \leq (n-1)/3$ and $t=\Omega(n)$. In this case, the total communication complexity computed as above can be rewritten in the form of  
\[\BT =  O(  \max\{ \ell n ,  n t \log t \}  ) \quad \text{bits}.\]  

If $t$ is relatively small, we can simply randomly select $n'\defeq 3t+1$ $\Pros$ for running the consensus as described in the proposed protocol (just replacing $n$ with $n'$, and replacing $\cb$ with $\cb'\defeq    \lceil \frac{ \max\{ \ell, \ (t/5 +1) \cdot \log (n'+1) \} }{k} \rceil$). After the consensus, those selected $n'$ $\Pros$ will send the coded information of the agreed message  to the rest of the $\Pros$, in a way that each of the selected $n'$ $\Pros$ sends a coded message with $\cb'$ bits (like in Phase~1). 
The proposed protocol guarantees that the consensus among the selected $n'$ $\Pros$ satisfies termination, consistency and validity conditions. After the consensus, the $n'$ $\Pros$ uses $(n', k)$ error correction code as described in Phase~1 to send the agreed message to each of the rest of the $\Pros$, where $k$ keeps the same form as in \eqref{eq:qdef} , i.e., $k =    \bigl \lfloor   \frac{ t  }{5 } \bigr\rfloor    +1$.  In this communication, each selected $\Pro$ just sends $\cb'$ bits of information coded from the agreed message, that is, $ \hv_i^\T   \Meg$, where  $i$ belongs to the set of indices of the selected $n'$ $\Pros$, and $\Meg$ denotes the agreed message.  The communication complexity in this step is $n' (n-n')\cb'$ bits. 
In this scenario, the total communication complexity is 
   \begin{align*}
 \BT &= \sum_{i=1}^6 \Bit_i' + n' (n-n')\cb' \\
  &= O( \cb' n' (n'-1) + {n'}^2 )  + n' (n-n')\cb'\\
    &= O( \cb' t^2 + t^2 )  + O(  t (n-t)\cb' )\\
 &= O( \cb' t^2 + t^2 +   t (n-t)\cb' )  \\
 &= O( \cb' nt  + t^2  )   \\
 & = O(  \max\{ \ell/t,  \log n' \}  n t + t^2 )   \\
 &= O(  \max\{ \ell n ,  nt \log n' \}  )        \\
  &= O(  \max\{ \ell n ,  nt \log t \}  )  \  \text{bits}
  \end{align*}
  where $\Bit_i'$ has the same form as $\Bit_i$ (see  Remarks~\ref{rk:ph1a}, \ref{rk:ph2a}, \ref{rk:ph3a}, \ref{rk:ph3b} and \ref{rk:ph4a}) but in the expression of $\Bit_i'$,  the parameter $n$ is replaced with $n'$ and  $\cb$ is replaced with $\cb'$, for $i\in \{1,2,\cdots, 6\}$. Recall that  $n'= 3t+1$, $\cb'=    \lceil \frac{ \max\{ \ell, \ (t/5 +1) \cdot \log (n'+1) \} }{k} \rceil$, and $k =    \bigl \lfloor   \frac{ t  }{5 } \bigr\rfloor    +1$.

 Thus, by combining the above two cases, it is concluded that the    total communication complexity of COOL is 
\[\BT = O(  \max\{ \ell n ,  nt \log t \}  )  \quad \text{bits}. \]

For the proposed protocol, the round complexity is dominated by the round complexity of the one-bit consensus in Phase~3, which is $O(t)$ rounds. Therefore, the round complexity of the proposed protocol is  $O(t)$ rounds.
\end{proof}

        \vspace{.1 in}
        
\begin{lemma}    \label{lm:ef}
For $n\geq  3t+1$, the proposed COOL is an   error-free  BA protocol, i.e., it satisfies the  termination, consistency and validity conditions in all executions. 
\end{lemma}
\begin{proof}
The proof  of Lemma~\ref{lm:ef} borrows tools from coding theory, graph theory and linear algebra, with key steps  given below.           
\begin{itemize} 
\item   Given $n\geq 3t+1$, and given  the parameter $k$ designed  in \eqref{eq:qdef}, at the end of Phase~3   there exists \emph{at most}  1  group of honest $\Pros$,  where the honest $\Pros$ within this group  have   the same  non-empty updated  message (like  $\Ac_{1}$   in Fig.~\ref{fig:coolex}), and  the  honest $\Pros$ outside this group  have  the same  empty updated  message (like  $\Ac_{2}$   in Fig.~\ref{fig:coolex}).  
\item    Given $n\geq 3t+1$, when the protocol stops at Phase~3, it is happened only when the honest $\Pros$ have inconsistent initial messages, and eventually they have the same consensus output $\phi$. 
\item   Given $n\geq 3t+1$, when the protocol goes to Phase~4,  it is ensured that there exists \emph{exactly}  1  group of honest $\Pros$ having   the same  non-empty updated  message (like $\Ac_{1}$   in Fig.~\ref{fig:coolex}), and the size of this group is at least $t+1$.  This makes sure that the messages of the rest of the honest $\Pros$ (like $\Ac_{2}$   in Fig.~\ref{fig:coolex}) can be calibrated,   in the presence of inconsistent information sent from dishonest $\Pros$.  In this way,  all honest $\Pros$ eventually have the same consensus output.  
\item   Given $n\geq 3t+1$, if  all  honest $\Pros$ have  the same initial message, it is guaranteed that they will go to Phase~4 and  have  the  validated consensus output.  
  \end{itemize}

Specifically, the proof  of Lemma~\ref{lm:ef} will use the following lemmas whose proofs  are provided in Section~\ref{sec:COOLproof}.

\begin{lemma}    \label{lm:ph32group}
Given $n\geq 3t+1$, at the end of Phase~3 of COOL there exists \emph{at most}  1  group of honest $\Pros$,  
where the honest $\Pros$ within this group  have   the same  non-empty updated  message, and  the  honest $\Pros$ outside this group  have  the same  empty updated  message. 
\end{lemma}
\begin{proof}
See Section~\ref{sec:ph32group}. Proving this lemma is the most important step in the protocol analysis.  The proof borrows tools from coding theory, graph theory and linear algebra.  
\end{proof}

\begin{lemma}    \label{lm:p4samem}
Given $n\geq 3t+1$,  all  honest $\Pros$  reach  the same agreement in COOL. 
\end{lemma}
\begin{proof}
See Section~\ref{sec:p4samem}. 
\end{proof}

\begin{lemma}    \label{lm:p4samemsih}
Given $n\geq 3t+1$, if  all  honest $\Pros$ have  the same initial message, then at the end of COOL  all honest $\Pros$ agree on this  initial message. 
\end{lemma}
\begin{proof}
See Section~\ref{sec:p4samemsih}. 
\end{proof}

\begin{lemma}    \label{lm:terminate}
Given $n\geq 3t+1$,  all  honest $\Pros$  eventually  output  messages and terminate   in COOL. 
\end{lemma}
\begin{proof}
See Section~\ref{sec:terminate}. 
\end{proof}
        
From Lemmas~\ref{lm:p4samem}-\ref{lm:terminate} it reveals that, given $n\geq 3t+1$, the  consistency, validity and termination  conditions are all satisfied in COOL, which completes the proof of Lemma~\ref{lm:ef}.   The above   Lemma~\ref{lm:ph32group} is used for the proofs of Lemmas~\ref{lm:p4samem}-\ref{lm:terminate}.  
 \end{proof}

Note that Lemma~\ref{lm:ef} holds true without using any assumptions on the cryptographic technique such as signatures, hashing, authentication and secret sharing, which implies that COOL is a \emph{signature-free}  protocol.  Lemma~\ref{lm:ef} also holds true even when the adversary, who takes full control over the dishonest $\Pros$,  has unbounded computational power, which implies that COOL is an \emph{information-theoretic-secure}  protocol.   
The results of Lemma~\ref{lm:comr} and Lemma~\ref{lm:ef} serve as the achievability proof of Theorem~\ref{thm:BAs}. 

From the proposed COOL, it reveals that  \emph{coding} is an effective approach for solving  the  BA problem.  
In a nutshell, carefully-crafted error correction codes  provide an efficient way of exchanging ``compressed'' information among distributed nodes, while keeping the ability of detecting errors, masking errors, and making a consistent and validated agreement at honest distributed nodes.

                \vspace{.1 in}

\begin{table} [h!] 
\small
\begin{center}
\caption{Summary of some notations for the proposed COOL.} \label{tb:notation}
\begin{tabu}{cc}
\toprule
notation & interpretation \\
\midrule
$k$   & error correction code parameter  (see \eqref{eq:qdef}) \\
\midrule
$\cb$   & error correction code parameter  (see \eqref{eq:qdef}) \\
\midrule
$\Me_i$   & initial message at $\PRO$~$i$ \\
\midrule
$\Me^{(i)}$   & updated or decoded message at $\PRO$~$i$ \\
\midrule
$\hv_i$   &  an encoding vector defined as in \eqref{eq:zidefh} \\
\midrule
$y_j^{(i)}$   & a value encoded with $\hv_j$ and updated message $\Me^{(i)}$ of $\PRO$~$i$  (see \eqref{eq:yi11})\\
\midrule
$\Lk_i (j)$   & a  link indicator  for  $\PRO$~$i$ and $\PRO$~$j$    (see \eqref{eq:lkindicator})   \\
\midrule
$\Ry_i$   &  a success indicator    at  $\PRO$~$i$  (see \eqref{eq:sindicator})   \\
\midrule
$\Ss_1$   &    a set of indices of  success indicators $\{\Ry_1, \Ry_2, \cdots, \Ry_n\}$ whose values are ones  (see \eqref{eq:vr0})   \\
\midrule
$\Ss_0$   &    a set of indices of  success indicators $\{\Ry_1, \Ry_2, \cdots, \Ry_n\}$ whose values are zeros  (see \eqref{eq:vr0})   \\
\midrule
$\Vr_i$   &    a vote from  $\PRO$~$i$  for going to next phase or stopping in the current phase  (see \eqref{eq:vindicator})   \\
\midrule
\end{tabu}
\end{center}
\end{table}

\begin{algorithm}
\caption{\textbf{:  COOL protocol, code for $\PRO$~$i$, $i \in [1:n]$}}  \label{algm:BAs}
\begin{algorithmic}[1]

\State  Initially set   $\Me^{(i)} = \Me_{i}$.     \quad \quad\quad \quad\quad\quad\quad  \quad\quad\quad\quad \quad\quad\quad\quad \quad\quad\quad\quad  \quad \emph{// Set the initial value of $\Me^{(i)}$}

\State $\PRO$~$i$  encodes its message into $n$ symbols as  $y_j^{(i)}  \defeq    \hv_j^\T \Me_{i}, \  j \in  [1:n]$.  \quad\quad\quad\quad\quad  \emph{// Encoding}

  \vspace{5pt} 
  
\Statex {\bf \emph{Phase~1}}
 
\State  $\PRO$~$i$ sends $(y_j^{(i)}, y_i^{(i)})$ to $\PRO$~$j$,    $\forall j \in  [1:n],  j\neq i$.   \quad\quad\quad \quad \emph{// Exchange coded  symbols}

 \For     {$j =1:n$}      \quad  \quad \quad \quad\quad \quad \quad \quad \quad \quad\quad \quad \quad  \quad \quad \quad \quad \quad \quad \quad \quad  \quad \quad \quad \quad \emph{//Update link indicator}
    \If {$ ((y_i^{(j)}, y_j^{(j)}) = = (y_i^{(i)}, y_j^{(i)}))$ }     
    
            \State      $\PRO$~$i$  sets   $\Lk_i (j) =1$.
                  
         \Else 
         
           \State	$\PRO$~$i$ sets   $\Lk_i (j) =0$.

        \EndIf
   \EndFor

   \If {$(\sum_{j=1}^n \Lk_i (j)  >= n - t )$}        \quad  \quad \quad \quad \quad \quad \quad \quad \quad  \quad \quad \quad \quad \emph{//Update success indicator and message}
    
            \State      $\PRO$~$i$   sets its success indicator as  $\Ry_i =1$.
                  
         \Else 
         
           \State	 $\PRO$~$i$	sets $\Ry_i =0$ and $\Me^{(i)} =\phi$.

        \EndIf
 
  \State  $\PRO$~$i$ sends the value of  $\Ry_i$  to all other $\Pros$. \quad\quad\quad\quad\quad \quad  \emph{//  Exchange success indicators}
 
 \State $\PRO$~$i$ creates the sets  $ \Ss_p  =   \{ j:  \Ry_{j}=p,   j \in [1:n ]\}$, $p\in \{0,1\}$, based on received   $ \{\Ry_{j}\}_ {j=1}^{n}$.

  \vspace{5pt} 
\Statex  {\bf \emph{Phase~2}}

 \If {$(\Ry_i  ==1)$}  

    \State $\PRO$~$i$ sets   $\Lk_i (j) = 0,  \forall j \in \Ss_0$.\quad\quad\quad\quad   \quad\quad\quad\quad \quad\quad\quad \quad  \quad  \quad\quad \emph{//  Mask identified errors} 
    
              \If {$(\sum_{j=1}^n \Lk_i (j)  < n - t)$}           
         
           	 \State $\PRO$~$i$	sets $\Ry_i =0$ and $\Me^{(i)} =\phi$.       \ \quad \quad  \quad \quad \quad \quad \emph{//Update success indicator and message}
           	 \State  $\PRO$~$i$ sends the value of  $\Ry_i$  to all other $\Pros$.    \quad\quad \quad  \emph{//  Exchange success indicators}     
              \EndIf                  
\EndIf
 \State $\PRO$~$i$ updates   $\Ss_0$ and $\Ss_1$ based on the newly received success indicators.

   \vspace{5pt} 
\Statex  {\bf \emph{Phase~3}}

 \If {$(\Ry_i  ==1)$}  

    \State $\PRO$~$i$ sets   $\Lk_i (j) = 0,  \forall j \in \Ss_0$.\quad\quad\quad\quad   \quad\quad\quad\quad \quad\quad\quad \quad  \quad  \quad\quad \emph{//  Mask identified errors} 
    
              \If {$(\sum_{j=1}^n \Lk_i (j)  < n - t)$}            
         
           	 \State $\PRO$~$i$	sets $\Ry_i =0$ and $\Me^{(i)} =\phi$.     	      \ \quad \quad  \quad \quad \quad \quad \emph{//Update success indicator and message}
            	 \State  $\PRO$~$i$ sends the value of  $\Ry_i$  to all other $\Pros$.    \quad\quad \quad  \emph{//  Exchange success indicators}                     
              \EndIf                  
\EndIf
  \State $\PRO$~$i$ updates   $\Ss_0$ and $\Ss_1$ based on the newly received success indicators.

     \If {   $(\sum_{j=1}^{n}  \Ry_j  >= 2t+1 )$}   \quad \quad\quad\quad\quad\quad \quad\quad\quad\quad\quad\quad\quad\quad\quad\quad\quad\quad\quad\quad\quad \quad\quad\quad\quad\quad  \emph{// Vote}
    
       \State   $\PRO$~$i$ sets the binary vote as $\Vr_i =1$.
                  
                  \Else 
         
                \State   $\PRO$~$i$ sets the binary vote as $\Vr_i =0$.  	 
                 
                 \EndIf
      
  \State   $\PRO$~$i$    runs the one-bit consensus with all  other $\Pros$  on the $n$ votes $\{\Vr_1, \Vr_2, \cdots, \Vr_n\}$, by using the  one-bit consensus protocol from \cite{BGP:92,CW:92}.  \quad \quad\quad\quad  \quad \emph{// One-bit consensus  on the $n$ votes} 
 
 \If {(the consensus of the votes $\{\Vr_1, \Vr_2, \cdots, \Vr_n\}$ is $1$)}    
    
              \State    $\PRO$~$i$   goes to next phase.
                  
         \Else 
         
           \State  $\PRO$~$i$  sets $\Me^{(i)} =\phi$ and considers it as a final consensus   and stops here.               
           	              
        \EndIf
 
    \vspace{5pt} 
\Statex  {\bf \emph{Phase~4}}
    
 \If {($ \Ry_{i} ==0$)}   
    
\State $\PRO$~$i$ updates  $y_i^{(i)}     \leftarrow  \text{Majority}( \{y_i^{(j)}:   \Ry_{j} =1, j\in [1:n ]\})$.   \quad\quad   \emph{// Update with majority rule} 
 
 \State  $\PRO$~$i$  sends  updated  $y_i^{(i)} $ to $\PRO$~$j$, $\forall j \in\Ss_0, j\neq i$.     \quad  \emph{//  Broadcast updated information} 

\State $\PRO$~$i$ decodes   message    with  new observations $\{y_1^{(1)},   \cdots, y_n^{(n)}\}$ and updates $\Me^{(i)}$.        \emph{// Decode}  

 \EndIf

\State  $\PRO$~$i$  outputs consensus   as the updated message $\Me^{(i)}$   and stops.    \quad\quad  \quad\quad\quad\quad \quad\quad\quad\quad \emph{// Stop} 
  
\end{algorithmic}
\end{algorithm}

\section{Analysis on termination, consistency and validity properties of   COOL}     \label{sec:COOLproof}

In this section we will provide  analysis on the termination, consistency and validity properties of  the proposed COOL. Specifically, we will prove Lemmas~\ref{lm:ph32group}-\ref{lm:terminate}, by  borrowing tools from coding theory, graph theory and linear algebra.   
Before proving Lemmas~\ref{lm:ph32group}-\ref{lm:terminate}, we will first define network groups in Section~\ref{sec:groupdef}, and provide Lemmas~\ref{lm:sizeboundMatrix}-\ref{lm:eta1231} in  Sections~\ref{sec:sizeboundMatrix}-\ref{sec:eta1231}, which will be used in the analysis.

\subsection{Groups in network}     \label{sec:groupdef}

At first we will define some groups of $\Pros$ in the $n$-$\Pro$ network, which will help us analyze the proposed protocol. 
The group definition is based on the values of  the initial  messages and the values of the success indicators $\{\Ry_i\}_{i=1}^{n}$.  Note that  the value of $\Ry_i$ could be updated in Phase~1,  Phase~2 and Phase~3. For the ease of notation, let us use $\Ry_i^{[1]}$, $\Ry_i^{[2]}$ and $\Ry_i^{[3]}$ to denote the values of $\Ry_i$ updated in Phase~1, Phase~2  and Phase~3, respectively. Similarly, let us use $\Lk_i^{[1]} (j)$, $\Lk_i^{[2]} (j)$ and $\Lk_i^{[3]} (j)$ to denote the values of $\Lk_i (j)$ updated in Phase~1, Phase~2  and Phase~3, respectively. 
Let us first define Group $\Fc$  as  the  indices of  all of the  dishonest $\Pros$, for $| \Fc |   = t$.
Let us then define some groups of honest $\Pros$ as
  \begin{align}
 \Ac_{\Lin} \defeq &  \{  i:    \Me_i =  \Meg_{\Lin},  \  i \notin  \Fc  , \ i \in [1:n]\}, \quad   \Lin \in [1 : \gn]     \label{eq:Aell00}   \\
\Ac_{\Lin}^{[1]} \defeq&   \{  i:  \Ry_i^{[1]} =1, \Me_i =  \Meg_{\Lin},  \  i \notin  \Fc, \ i \in [1:n]\}, \quad   \Lin \in [1 : \gn^{[1]}]     \label{eq:Aell01}   \\ 
   \Ac_{\Lin}^{[2]} \defeq &  \{  i:  \Ry_i^{[2]} =1, \Me_i =  \Meg_{\Lin},  \  i \notin  \Fc, \ i \in [1:n]\}, \quad   \Lin \in [1 : \gn^{[2]}]     \label{eq:Aell03}  \\
      \Ac_{\Lin}^{[3]} \defeq &  \{  i:  \Ry_i^{[3]} =1, \Me_i =  \Meg_{\Lin},  \  i \notin  \Fc, \ i \in [1:n]\}, \quad   \Lin \in [1 : \gn^{[3]}]     \label{eq:Aell03new}    
 \end{align} 
 for  some different non-empty $\ell$-bit  values $\Meg_{1}, \Meg_{2}, \cdots, \Meg_{\gn}$ and some non-negative integers $\gn, \gn^{[1]},\gn^{[2]}, \gn^{[3]}$ such that $\gn^{[3]} \leq \gn^{[2]} \leq \gn^{[1]} \leq \gn$.  
The above definition implies that   Group~$\Ac_{\Lin}$ is a subset of honest $\Pros$ who have the same value of  initial  messages.  $\Ac_{\Lin}^{[1]}$ is a subset of $\Ac_{\Lin}$ who have the same non-empty value of updated  messages at the end of Phase~1.  Note that at the end of Phase~1, if the updated  message of  honest $\PRO~i$ is non-empty, then it implies that its updated message remains the same as the initial message and that $\Ry_i^{[1]} =1$.   
Similarly, $\Ac_{\Lin}^{[2]}$ is a subset of $\Ac_{\Lin}^{[1]}$ who have the same non-empty value of updated  messages at the end of Phase~2 for $\Lin \in [1 : \gn^{[2]}]$, while $\Ac_{\Lin}^{[3]}$ is a subset of $\Ac_{\Lin}^{[2]}$ who have the same non-empty value of updated  messages at the end of Phase~3 for $\Lin \in [1 : \gn^{[3]}]$.
In our setting, when  $1\leq \gn^{[3]} \leq \gn^{[2]} \leq \gn^{[1]} \leq \gn$, the sets $\Ac_{\Lin}, \Ac_{\Lin_1}^{[1]}, \Ac_{\Lin_2}^{[2]},  \Ac_{\Lin_3}^{[3]}$ defined in  \eqref{eq:Aell00}-\eqref{eq:Aell03} are all non-empty for  any $\Lin \in [1 : \gn], \Lin_1 \in [1 : \gn^{[1]}], \Lin_2 \in [1 : \gn^{[2]}], \Lin_3 \in [1 : \gn^{[3]}]$. 
One example of the $n$-$\Pro$ network is provided in Fig.~\ref{fig:nNetwork}.

 \begin{figure}
\centering
\includegraphics[width=7.5cm]{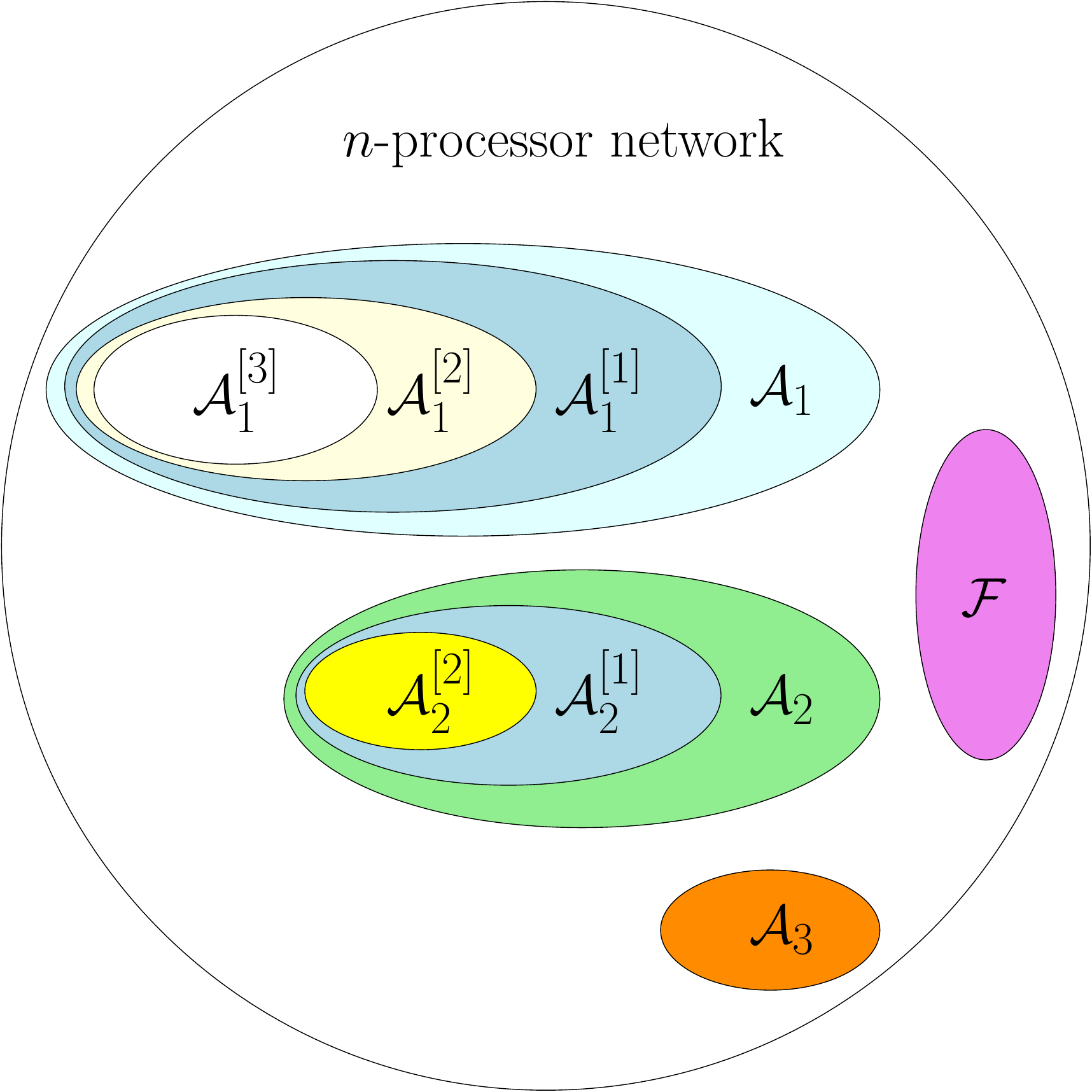}
\caption{One example of $n$-$\Pro$ network with $( \gn =3, \gn^{[1]} =\gn^{[2]} =2, \gn^{[3]} =1)$, where $\Ac_{1}^{[3]} \subseteq \Ac_{1}^{[2]} \subseteq \Ac_{1}^{[1]} \subseteq \Ac_{1}$ and  $\Ac_{2}^{[2]} \subseteq \Ac_{2}^{[1]} \subseteq\Ac_{2}$.} 
\label{fig:nNetwork} 
\end{figure}

 Based on our definition, it holds true that 
 \begin{align}
  \sum_{\Lin=1}^{\gn}  |   \Ac_{\Lin} |      +  |\Fc|   = n  .      \label{eq:sumn001}  
 \end{align}
Let us also define   $\Bc^{[p]}$  as 
  \begin{align}
\Bc^{[p]} \defeq  & \{  i:  \Ry_i^{[p]} =0, \  i \notin  \Fc, \ i \in [1:n] \}, \quad  p\in \{1,2, 3\}.    \label{eq:Bdef01} 
 \end{align} 
 Based on our definitions, it holds true that 
 \begin{align}
  \sum_{\Lin=1}^{\gn^{[p]}}  |   \Ac_{\Lin}^{[p]} |     +  |\Bc^{[p]}|  +  |\Fc|   = n , \quad  p\in \{1,2, 3\}.       \label{eq:sumn01}  
 \end{align} 
For  some $i \in \Ac_{\Lin}$,   the equality of  $\hv_i^\T  \Meg_{\Lin} = \hv_i^\T  \Meg_j$ might be satisfied for some $j$ and  $\Lin$.  
With this motivation, Group $\Ac_{\Lin}$  can be further divided into some (possibly overlapping) sub-groups  defined as 
    \begin{align}
 \Ac_{\Lin,j} \defeq   & \{  i:  \   i\in  \Ac_{\Lin},   \   \hv_i^\T  \Meg_{\Lin}  = \hv_i^\T  \Meg_j \} , \quad  j\neq \Lin ,   \  j, \Lin \in [1: \gn]   \label{eq:Alj}  \\     
 \Ac_{\Lin,\Lin} \defeq  &  \Ac_{\Lin}   \setminus  \{\cup_{j=1, j\neq \Lin}^{\gn}\Ac_{\Lin,j}\}   , \quad  \Lin \in [1: \gn] .  \label{eq:All}       
 \end{align} 
 Similarly, Group $\Ac_{\Lin}^{[p]}$  can be further divided into some   sub-groups  defined as
    \begin{align}
 \Ac_{\Lin,j}^{[p]} \defeq   & \{  i:  \   i\in  \Ac_{\Lin}^{[p]},   \   \hv_i^\T  \Meg_{\Lin}  = \hv_i^\T  \Meg_j \} , \quad  j\neq \Lin ,   \  j, \Lin  \in [1: \gn^{[p]}]   \label{eq:Alj11}  \\     
 \Ac_{\Lin,\Lin}^{[p]} \defeq  &  \Ac_{\Lin}^{[p]}   \setminus  \{\cup_{j=1, j\neq \Lin}^{\gn^{[p]}}\Ac_{\Lin,j}^{[p]}\}   , \quad  \Lin \in [1: \gn^{[p]}] .  \label{eq:All11}       
 \end{align}  
 for $p\in \{1,2, 3\}$. 
 
 \subsection{Lemma~\ref{lm:sizeboundMatrix} and its proof}  \label{sec:sizeboundMatrix}

In this sub-section we  provide Lemma~\ref{lm:sizeboundMatrix} that will be used later for the analysis of the proposed protocol. 
 
\begin{lemma}    \label{lm:sizeboundMatrix}
For  $\Ac_{\Lin,j}$ and  $\Ac_{\Lin,j}^{[1]}$ defined in \eqref{eq:Alj} and \eqref{eq:Alj11}, and for $\gn \geq \gn^{[1]} \geq 2$, the following inequalities  hold true 
    \begin{align}
  |\Ac_{\Lin,j}| + |\Ac_{j,\Lin}|  <  &  k,    \quad \forall   j \neq \Lin,  \  j, \Lin \in [1:\gn]    \label{eq: Aljb00}  \\
    |\Ac_{\Lin,j}^{[1]}| + |\Ac_{j,\Lin}^{[1]}|  <  &  k,  \quad \forall   j \neq \Lin,  \  j, \Lin \in [1:\gn^{[1]}]      \label{eq: Aljb01}  
 \end{align} 
 where $k$ is defined in \eqref{eq:qdef}.
  \end{lemma}
\begin{proof}
The proof  of Lemma~\ref{lm:sizeboundMatrix} borrows tool from linear algebra.  We will focus on the proof of \eqref{eq: Aljb00}, as the proof of \eqref{eq: Aljb01} is very similar.  
 The proof will use the fact that 
      \begin{align}
  \Meg_{\Lin}-  \Meg_j   \neq  \boldsymbol{0} , \quad \forall   j \neq \Lin,  \  j, \Lin \in [1:\gn]   \label{eq:messagediff}      
 \end{align}
 as well as the fact that  
     \begin{align}
  \hv_i^\T ( \Meg_{\Lin}   -  \Meg_j  ) = \boldsymbol{0}  \quad    \forall  i\in   \Ac_{\Lin,j} \cup \Ac_{j,\Lin}, \quad   j\neq \Lin,   \  \Lin, j \in [1: \gn]   \label{eq:zerodiff}      
 \end{align}
 which follows from the definition of  $\Ac_{\Lin,j}$   in \eqref{eq:Alj}.  
 Let us define   $\Hm_{\Lin, j} $  as   an  $|\Ac_{\Lin,j}| \times k$ matrix  such that 
    \begin{align}
\Hm_{\Lin, j} \defeq  \Bmatrix{ \hv_{i_1}^\T \\  \hv_{i_2}^\T \\ \vdots \\ \hv_{i_{|\Ac_{\Lin,j}|}}^\T }     \quad \text{for}\quad   i_{1}, i_{2}, \cdots,  i_{|\Ac_{\Lin,j}|}    \in  \Ac_{\Lin,j}  ,  \quad    i_{1} < i_{2} < \cdots < i_{|\Ac_{\Lin,j}|}      \label{eq:Hej}      
 \end{align}
 for   $j\neq \Lin,   \  \Lin, j \in [1: \gn]$. Note that $\Hm_{\Lin, j}$ is full rank, based on the definition of $\hv_i$ as in \eqref{eq:zidefh}.  
 This is because any matrix, whose rows are different and generated as in \eqref{eq:zidefh}, is full rank.  
 From \eqref{eq:zerodiff} and \eqref{eq:Hej}, it implies that the following equality holds true that 
 \begin{align}
   \Bmatrix{ \Hm_{\Lin, j} \\ \Hm_{j, \Lin}} ( \Meg_{\Lin}   -  \Meg_j  ) = \boldsymbol{0}    .  \label{eq:HX0lj}   
 \end{align}
 Note that $ \Bmatrix{ \Hm_{\Lin, j} \\ \Hm_{j, \Lin}}$ is  a  full rank matrix, based on the definition of $\hv_i$ as in \eqref{eq:zidefh} and the fact of  
 $\Ac_{\Lin,j}  \cap \Ac_{j,\Lin} = \phi$  for $ j\neq \Lin$.  The fact of $\Ac_{\Lin,j}  \cap \Ac_{j,\Lin} = \phi$ results from the identities of  $\Ac_{\Lin,j} \subseteq \Ac_{\Lin}$, $\Ac_{j,\Lin} \subseteq\Ac_{j}$ and  $\Ac_{j} \cap \Ac_{\Lin} = \phi$  for $ j\neq \Lin$.  This fact implies that every row of  $\Hm_{\Lin, j}$ is different from every row of $\Hm_{j, \Lin}$ for $ j\neq \Lin$, where the rows of those two matrices are generated as in \eqref{eq:zidefh}. 
 
Based on the definition in \eqref{eq:Hej}, the dimension of the matrix  $ \Bmatrix{ \Hm_{\Lin, j} \\ \Hm_{j, \Lin}}$ is  $(|\Ac_{\Lin,j}| + |\Ac_{j,\Lin}|)\times k$. 
For the equality of  \eqref{eq:HX0lj},   by combining the fact  that  $\Meg_{\Lin}   -  \Meg_j  \neq  \boldsymbol{0}$  (see \eqref{eq:messagediff}) and the fact that  the $(|\Ac_{\Lin,j}| + |\Ac_{j,\Lin}|)\times k$ matrix $ \Bmatrix{ \Hm_{\Lin, j} \\ \Hm_{j, \Lin}}$  is full rank, we can conclude that the following inequality must hold true 
 \begin{align}
|\Ac_{\Lin,j}| + |\Ac_{j,\Lin}|  < k ,   \quad  j\neq \Lin , \Lin, j \in [1: \gn]  \label{eq:m1m2}      
 \end{align}
 otherwise we have $\Meg_{\Lin}   -  \Meg_j = \boldsymbol{0}$, which contradicts with  the fact  in \eqref{eq:messagediff}. 
 At this point we complete the proof of \eqref{eq: Aljb00}, as well as the proof of \eqref{eq: Aljb01} by replacing  $\Ac_{\Lin,j}$ and $\Ac_{j,\Lin}$ with $\Ac_{\Lin,j}^{[1]}$ and $\Ac_{j,\Lin}^{[1]}$ respectively. 
 \end{proof}

 \subsection{Lemma~\ref{lm:graph} and its proof}  \label{sec:graph}

By using graph theory we  will derive a lemma that will be used later for the analysis of the proposed protocol. 
Let us consider a graph $G=(\Pc, \Ec)$, where $\Pc$ includes $n -t$  vertices,  $\Pc= [1: n-t]$,  and $\Ec$ is a set of edges. 
For this graph,  we consider a given vertex $i^{\star}$ for $i^{\star}\in \Pc$, and a set of vertices $\Cc$ for $\Cc \subseteq \Pc \setminus \{i^{\star}\}$ and  $|\Cc|  \geq n-2t -1$,  such that each vertex in $\Cc$ is connected with at least $n-2t$ edges and one of the edges is connected to vertex $i^{\star}$. In our setting, a loop connecting to itself is also counted as an edge. 
We use $E_{i,j}= 1$ (resp. $E_{i,j}= 0$) to indicate that there is an edge (resp. no edge) between vertex $i$ and vertex $j$, for $E_{i,j}= E_{j,i}, \forall i, j \in \Pc$.  
Mathematically, for the graph $G=(\Pc, \Ec)$ we considered,  there exists a set  $\Cc \subseteq \Pc \setminus \{i^{\star}\}$ such that the following conditions are satisfied: 
 \begin{align}
E_{i, i^{\star}}&= 1                          \quad     \forall i \in \Cc   \label{eq:graph01}  \\
\sum_{j\in \Pc}E_{i, j} &\geq   n-2t   \quad     \forall i \in \Cc  \label{eq:graph02}   \\
|\Cc| &   \geq n-2t -1     \label{eq:graph03} 
 \end{align}
for a given $i^{\star}\in \Pc= [1: n-t]$.

 \begin{figure}
\centering
\includegraphics[width=5cm]{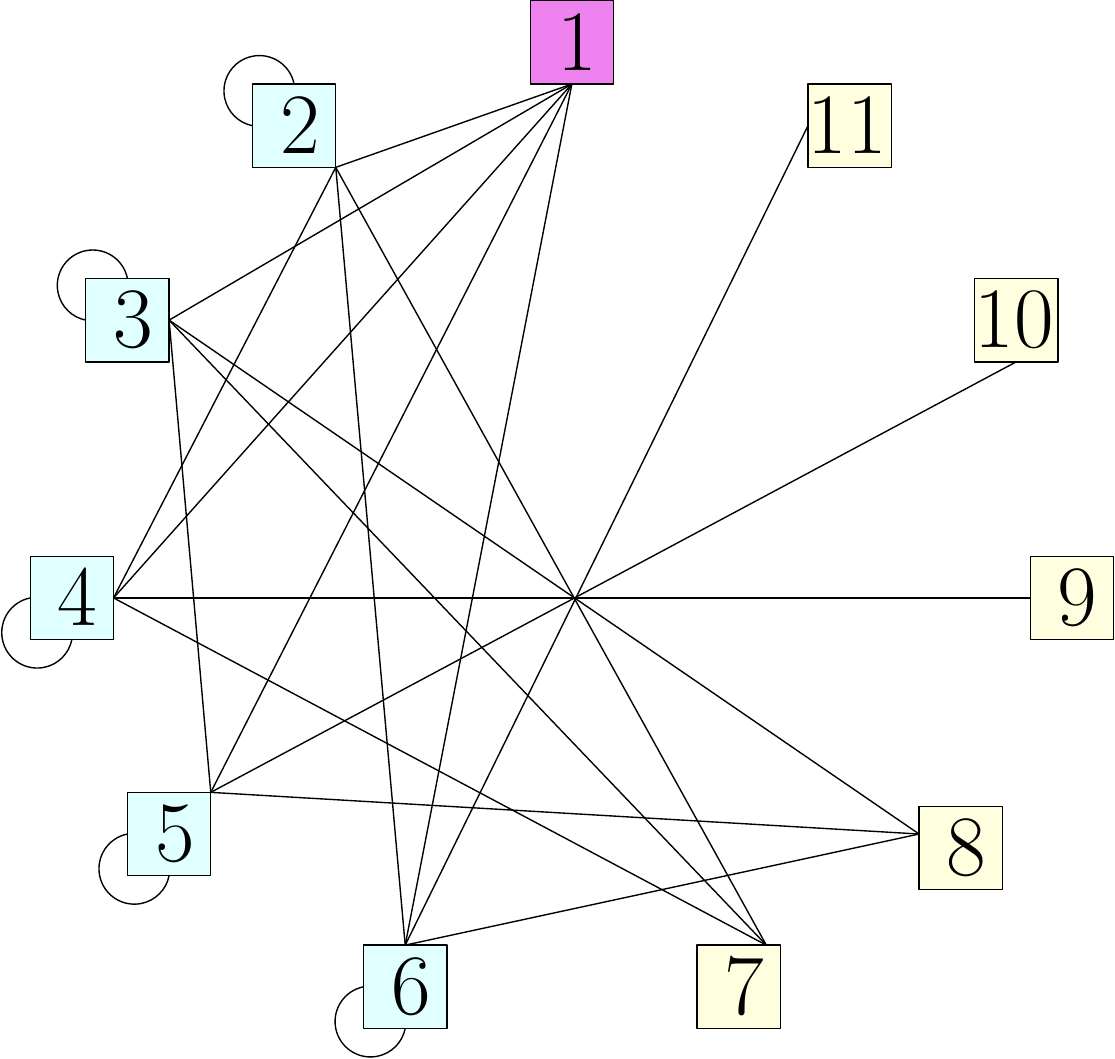}
\caption{One example of the  graph $G=(\Pc, \Ec)$ with $(t =5, n= 16)$, $\Pc= [1: 11]$, $i^{\star} =1$, $\Cc = \{2, 3, 4, 5, 6\}$, $k=2$ and $\Dc = \{2, 3, 4, 5, 6, 7, 8\}$.} 
\label{fig:graph}
\end{figure}

 \begin{figure}
\centering
\quad\quad \includegraphics[width=8cm]{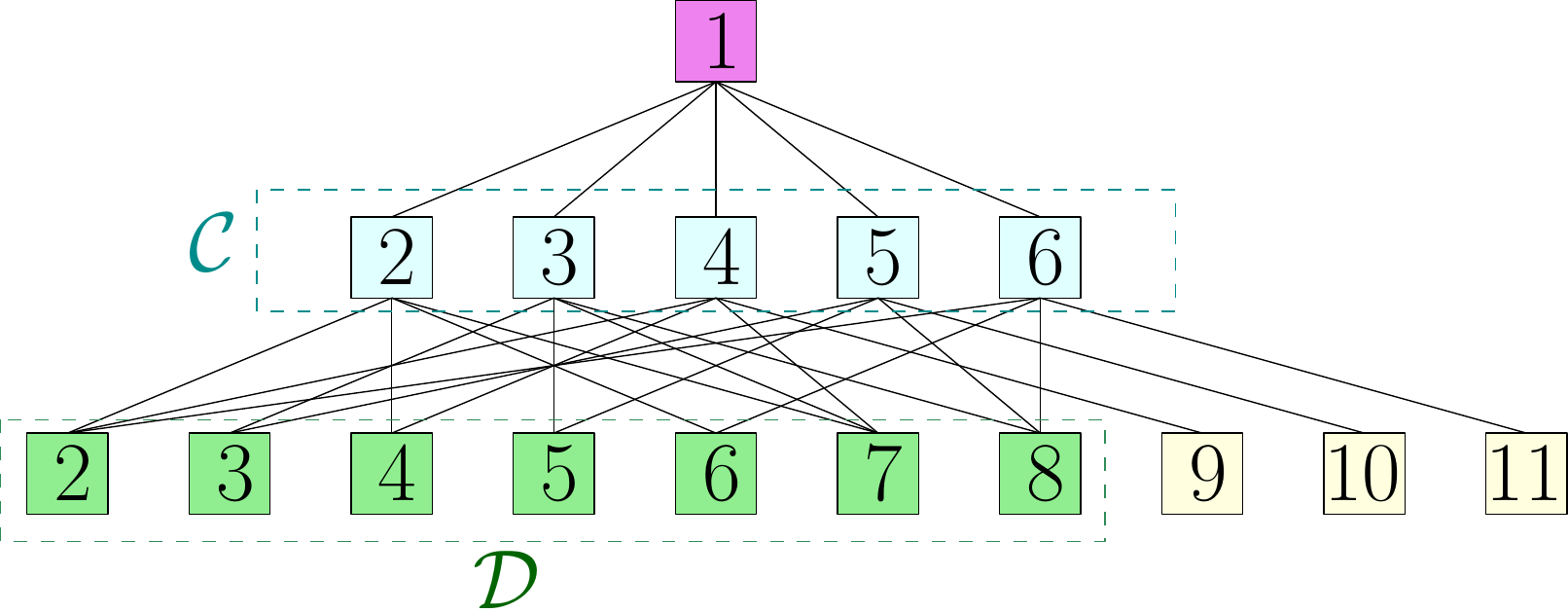}
\caption{A three-layer representation of the  graph $G=(\Pc, \Ec)$ appeared in Fig.~\ref{fig:graph},  for $(t =5, n= 16)$, $\Pc= [1: 11]$, $i^{\star} =1$, $\Cc = \{2, 3, 4, 5, 6\}$, $k=2$ and $\Dc = \{2, 3, 4, 5, 6, 7, 8\}$.} 
\label{fig:grapht5CD}
\end{figure}

For the graph  $G=(\Pc, \Ec)$ defined as above, let $\Dc \subseteq \Pc$ denote the set of vertices such that each vertex in $\Dc$ is connected with at least $k$ vertices in $\Cc$, that is,  
 \begin{align}
\Dc \defeq \Bigl\{i: \ \sum_{j\in \Cc }E_{i, j}  \geq  k , \   i \in  \Pc \setminus \{i^{\star}\} \Bigr\}  \label{eq:graphD01}   
 \end{align}
where $k$ is defined in \eqref{eq:qdef}. One example of the graph  $G=(\Pc, \Ec)$  is depicted in Fig.~\ref{fig:graph} and a  three-layer representation of this  graph is provided in Fig.~\ref{fig:grapht5CD}. For the  graph $G=(\Pc, \Ec)$ defined as above, the following lemma provides a result on bounding the size of $\Dc$.

\begin{lemma}    \label{lm:graph}
For any graph $G=(\Pc, \Ec)$ specified by \eqref{eq:graph01}-\eqref{eq:graph03} and for the set $\Dc \subseteq \Pc$ defined by \eqref{eq:graphD01}, and given $n\geq 3t+1$, it holds true that 
 \begin{align}
|\Dc| &\geq  n-9t/4-1 .    \label{eq:graphDr11}  
 \end{align}
  \end{lemma}
 \begin{proof}
 For a graph $G=(\Pc, \Ec)$ specified by \eqref{eq:graph01}-\eqref{eq:graph03}, there exists the following number of edges connected between $\Cc$ and $\Pc \setminus \{i^{\star}\}$
  \begin{align}
m_e \defeq \sum_{i\in \Cc} \sum_{j\in \Pc \setminus \{i^{\star}\}}E_{i, j}    \label{eq:graph02aa}   
 \end{align}
which is bounded by 
  \begin{align}
m_e  \geq  |\Cc|  \cdot (n-2t -1)   \label{eq:graph02aab}  
 \end{align}
 based on  \eqref{eq:graph01}-\eqref{eq:graph02}.
 For those $m_e$  edges connected between $\Cc$ and $\Pc \setminus \{i^{\star}\}$, it holds true that any vertex $j \in \Pc \setminus \{i^{\star}\}$ can have at most $|\Cc|$ edges, that is, 
 \begin{align}
\sum_{i\in \Cc}E_{i, j} &\leq  |\Cc|,  \quad     \forall  j\in \Pc \setminus \{i^{\star}\} . \label{eq:graph02bb}   
 \end{align}
 To  bound  the size of $\Dc$ that is defined  by \eqref{eq:graphD01}, we will consider an extreme scenario that has the minimum size of $\Dc$.  
 Recall that each vertex in $\Dc$ is connected with at least $k$, but at most $|\Cc|$, vertices in $\Cc$.  
 Also recall that, for a given graph $G=(\Pc, \Ec)$  satisfying \eqref{eq:graph01}-\eqref{eq:graph03}, it has $m_e$  edges connected from $\Cc$ to $\Pc \setminus \{i^{\star}\}$  (see \eqref{eq:graph02aa}). 
 When these $m_e$  edges are  connected from $\Cc$ to $\Pc \setminus \{i^{\star}\}$ in some different ways, it might result in some different graph scenarios, each graph scenario with a different size of $\Dc$  ---  these graph scenarios have the same $\Cc$ and have the same  $m_e$  number of edges connected from $\Cc$ to  $\Pc \setminus \{i^{\star}\}$ but have different sizes of $\Dc$. 
  Therefore, minimizing the size of $\Dc$ is a game of allocating  these $m_e$ edges connected from $\Cc$ to $\Pc \setminus \{i^{\star}\}$.  
 We argue that, given the same $\Cc$ and  the same $m_e$  number of edges connected from $\Cc$ to $\Pc \setminus \{i^{\star}\}$,  the extreme scenario satisfying the following conditions has the minimum size of $\Dc$:
 \begin{itemize}
\item  Condition (a): Every vertex in $\{\Pc \setminus \{i^{\star}\} \}\setminus \Dc$  is connected with  $k-1$ vertices in $\Cc$.    Assume that there exists one scenario with the minimum size of $\Dc$ such that  Condition (a) is not satisfied, i.e., at least one vertex in $\{\Pc \setminus \{i^{\star}\} \}\setminus \Dc$ has less than  $k-1$ edges connected from $\Cc$.  Then, given the fixed $m_e$  number of edges connected from $\Cc$ to $\Dc$ and  $\{\Pc \setminus \{i^{\star}\} \}\setminus \Dc$,  one can  increase the  number of edges connected to $\{\Pc \setminus \{i^{\star}\} \}\setminus \Dc$ for satisfying  Condition (a), by decreasing  the number of edges connected to $\Dc$, which will not increase the size of $\Dc$.
\item  Condition (b): $|\Dc| -1$ vertices in $\Dc$ are all fully  connected with  $|\Cc|$ vertices in $\Cc$. Assume that there exists one scenario with the minimum size of $\Dc$ such that  Condition (b) is not satisfied, i.e., at least two vertices in $\Dc$  are not fully  connected with  $|\Cc|$ vertices in $\Cc$.  
Then, one can move some edges from one vertex in $\Dc$ --- so that this vertex could possibly be removed from $\Dc$ when the number of  edges connected to this vertex is less than $k$ --- to the other  vertices  in $\Dc$  for satisfying  Condition (b), which will not increase the size of $\Dc$. 
  \end{itemize} 
Fig.~\ref{fig:worstgraph} describes  an example  of  achieving the minimum size of $\Dc$ by moving the edges to satisfy Condition (a) and Condition (b).

 \begin{figure}[t!]
\centering
 \[\begin{array}{cc}
   \includegraphics[width = 0.3\columnwidth]{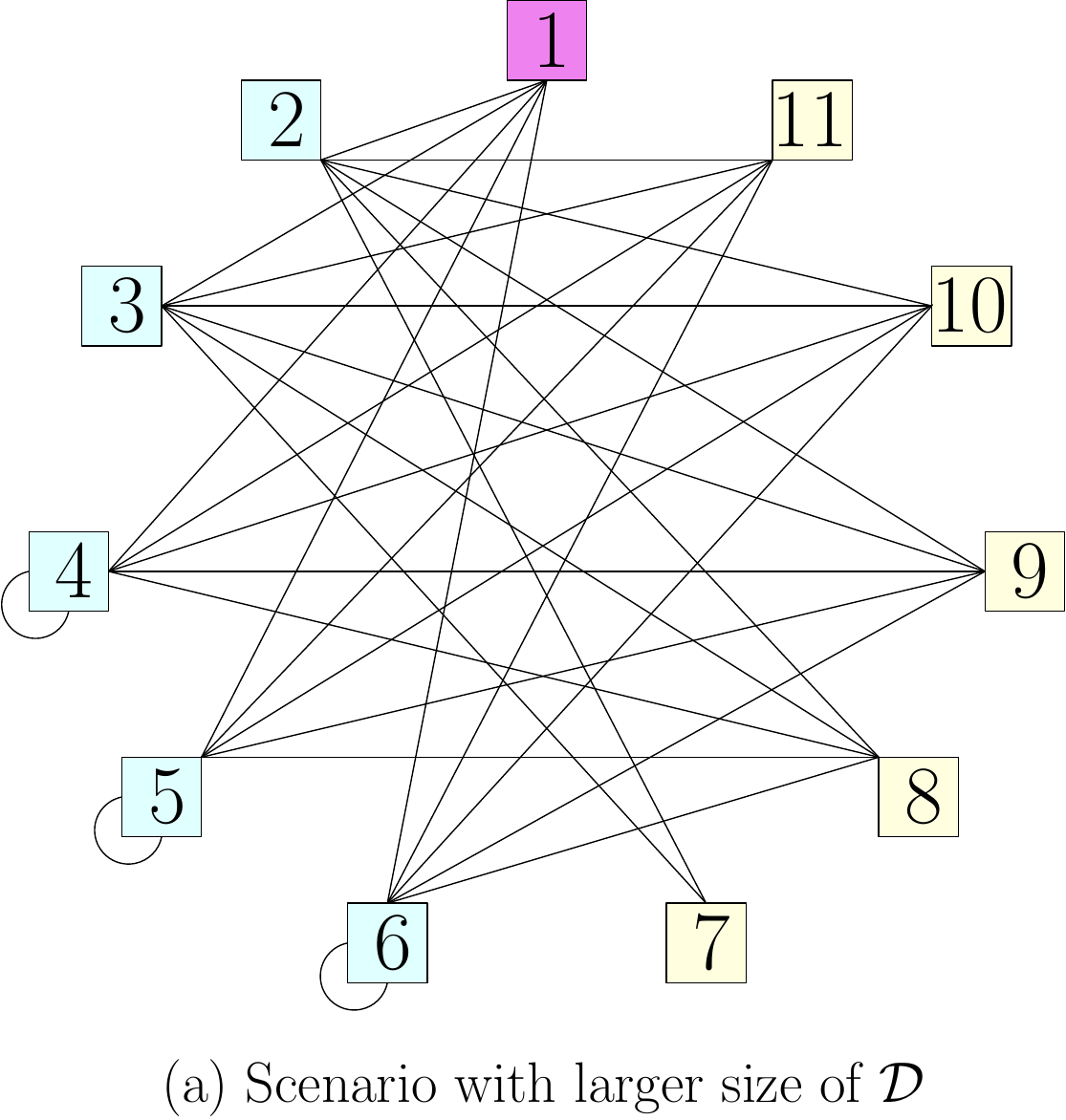}  \quad\quad  &  \quad \quad \includegraphics[width = 0.3\columnwidth]{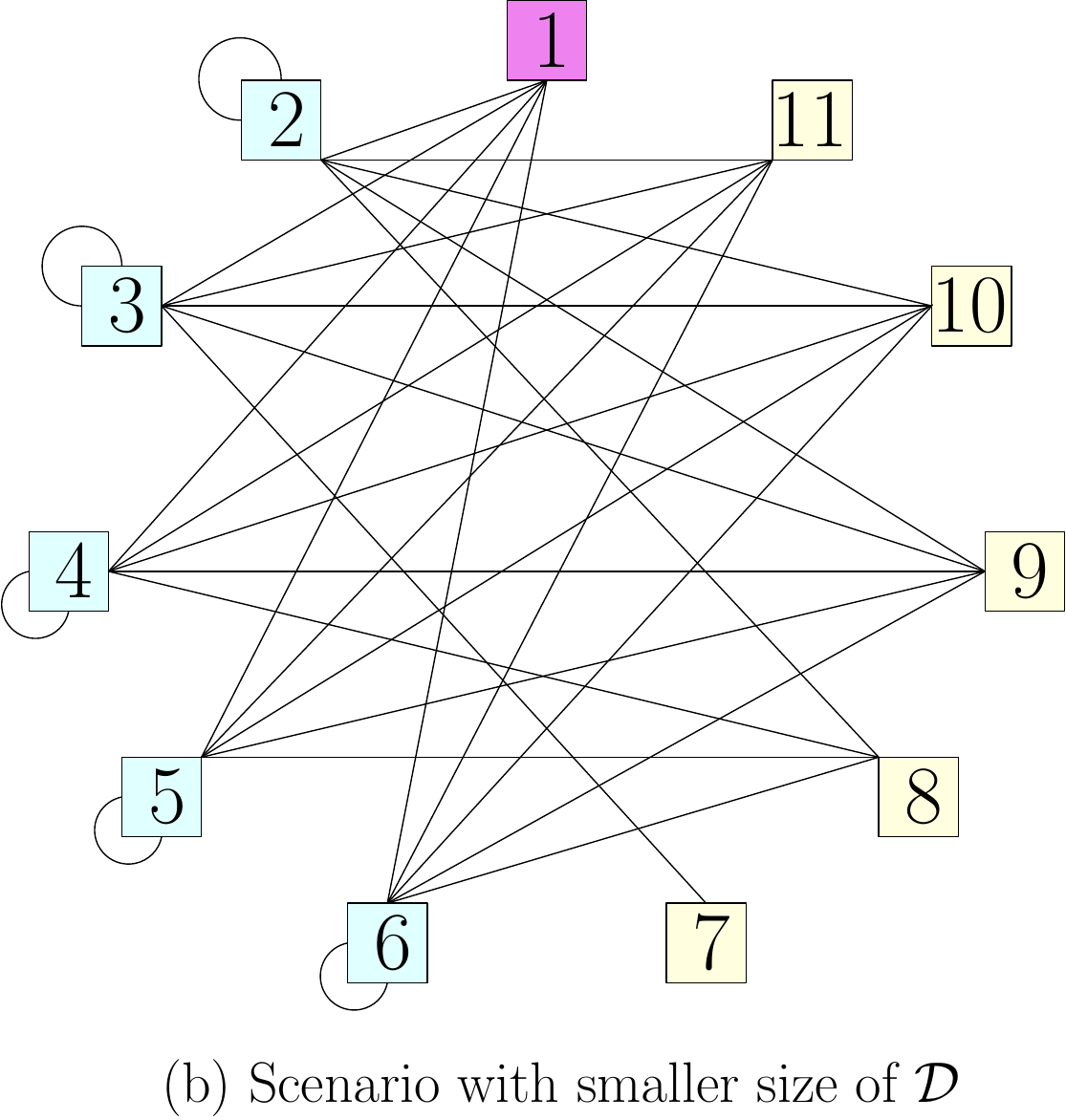}  
\end{array} \]
\caption{One example of achieving the minimum size of $\Dc$ by moving the edges to satisfy Condition (a) and Condition (b), for the graphs $G=(\Pc, \Ec)$ with $(t =5, n= 16)$, $\Pc= [1: 11]$, $i^{\star} =1$, $\Cc = \{2, 3, 4, 5, 6\}$ and $k=2$. For the scenario in the left size, the number of edges connected from $\Cc$ to the vertices $2, 3, \cdots, 11$ are $0, 0, 1, 1, 1, 2, 5, 5, 5, 5$, respectively, which indicates that $\Dc = \{7, 8, 9, 10, 11\}$. For the scenario in the right size, the number of edges connected from $\Cc$ to the vertices $2, 3, \cdots, 11$ are $1, 1, 1, 1, 1, 1, 4, 5, 5, 5$, respectively, which indicates that $\Dc = \{8, 9, 10, 11\}$. The second scenario satisfies  Condition (a) and Condition (b) and has a smaller size of  $\Dc$.} 
\label{fig:worstgraph}
\end{figure}

Let us consider the extreme scenario  satisfying  Condition (a) and Condition (b). We use $\Dc^*$ to denote the set defined as in \eqref{eq:graphD01} of this extreme scenario and use  $|\Dc^*|$ to denote the corresponding  size.
Given that the graph has $m_e$  edges connected from $\Cc$ to $\Dc^*$ and  $\{\Pc \setminus \{i^{\star}\} \}\setminus \Dc^*$, then the following equality holds true 
   \begin{align}
 (k-1) (n-t-1 - |\Dc^*|) +  \tau (|\Dc^*|-1)   +  \tau_0     =  m_e  \label{eq:setD01}  
 \end{align}
  for  $ \tau \defeq |\Cc|$ and for some $\tau_0$ such that $k\leq \tau_0 \leq \tau$. The first term in the left hand side of \eqref{eq:setD01} results from Condition (a), i.e., every vertex in $\{\Pc \setminus \{i^{\star}\} \}\setminus \Dc^*$  is connected with  $k-1$ vertices in $\Cc$, where $|\{\Pc \setminus \{i^{\star}\} \}\setminus \Dc^*| = (n-t-1 - |\Dc^*|)$. The remaining  terms in the left hand side of \eqref{eq:setD01} follow from Condition (b). 
From \eqref{eq:setD01}, it holds true that 
   \begin{align}
     |\Dc^*|      &=  \frac{m_e  +  \tau  -\tau_0 - (k-1) (n-t-1) }{\tau -  (k-1) }   \label{eq:setD02}  \\
     &\geq    \frac{m_e  - (k-1) (n-t-1) }{\tau -  (k-1) }   \label{eq:setD03}  \\
     &\geq    \frac{\tau (n-2t-1)  - (k-1) (n-t-1) }{\tau -  (k-1) }   \label{eq:setD04}  \\
     &=   (n-2t-1)  -  \frac{ (k-1)t  }{\tau -  (k-1) }   \non  \\
     & \geq   n-2t-1  -  \frac{ (k-1)t  }{n-2t-1 -  (k-1) }   \label{eq:setD05}  \\
     & =   n-2t-1  -  \frac{t  }{\frac{n-2t-1}{k-1} - 1 }   \non  \\
     & \geq    n-2t-1  -  \frac{t  }{\frac{n-2t-1}{   t/5  +1-1} - 1 }   \label{eq:setD06}  \\
     & \geq    n-2t-1  -  \frac{t  }{\frac{t}{   t/5 } - 1 }   \label{eq:setD07}  \\
     & =    n-2t-1  -  t /4    \label{eq:setD08}  
 \end{align}
for $ \tau = |\Cc|$,  where \eqref{eq:setD02} is derived from \eqref{eq:setD01};  \eqref{eq:setD03}  uses the condition that $\tau_0 \leq \tau$; 
 \eqref{eq:setD04} follows from \eqref{eq:graph02aab}; 
  \eqref{eq:setD05} uses the condition in \eqref{eq:graph03};
    \eqref{eq:setD06} results from the fact that  $k=   \bigl \lfloor   \frac{ t  }{5 } \bigr\rfloor    +1 \leq      \frac{ t  }{5 }     +1 $ (see \eqref{eq:qdef});
    \eqref{eq:setD07} follows from the condition that $n\geq 3t+1$. 
    One can see that the lower bound of $|\Dc^*|$ shown in \eqref{eq:setD08} does not depend on the values of $m_e$ and $|\Cc|$, as long as $m_e$ and $|\Cc|$ satisfy the conditions \eqref{eq:graph01}-\eqref{eq:graph03}.  
    Since $|\Dc^*|$ is the size of $\Dc^*$ of the extreme scenario, then the size of $\Dc$  of any other scenario of the graph $G=(\Pc, \Ec)$ specified by \eqref{eq:graph01}-\eqref{eq:graph03} satisfies the following inequalities    
     \begin{align}
      |\Dc|   \geq    |\Dc^*|   \geq  n-9t/4-1   \label{eq:setDfinal}  
 \end{align}
 where the last inequality results from \eqref{eq:setD08}. The above result holds for any $m_e$ and $\Cc$ that satisfy the conditions \eqref{eq:graph01}-\eqref{eq:graph03}.  At this point we complete the proof of Lemma~\ref{lm:graph}.
     \end{proof}

 \subsection{Lemma~\ref{lm:sizem} and its proof}  \label{sec:sizem}

In this sub-section we  provide a lemma that will be used later for the analysis of the proposed protocol.  The result of Lemma~\ref{lm:graph} will be used in this proof.

\begin{lemma}    \label{lm:sizem}
When $\gn^{[2]}\geq 1$, it holds true that $|\Ac_{\Lin}| \geq n-9t/4$, for any  $\Lin \in [1: \gn^{[2]}]$.
\end{lemma}
 \begin{proof}
 The proof consists of the following steps: 
 \begin{itemize}
\item  Step (a): Transform the network into a graph that is within the family of graphs satisfying \eqref{eq:graph01}-\eqref{eq:graph03} (see Section~\ref{sec:graph}), for a fixed $i^{\star}$ in  $\Ac_{\Lin^{\star}}^{[2]}$ and $\Lin^{\star} \in [1: \gn^{[2]}]$.
\item  Step (b): Bound the size of a group of honest $\Pros$, denoted by $\Dc'$ (with the same form as  in \eqref{eq:graphD01}), using the result of Lemma~\ref{lm:graph}, i.e., $|\Dc'| \geq n-9t/4-1$.
\item Step (c): Argue that every $\Pro$ in $\Dc'$ has the same initial message as $\PRO$~$i^{\star}$.
\item  Step (d): Conclude from Step (c) that $\Dc'$ is a subset of $\Ac_{\Lin^{\star}}$, i.e., $\Dc'\cup \{i^{\star}\} \subseteq \Ac_{\Lin^{\star}}$ and conclude that the size of $\Ac_{\Lin^{\star}}$ is bounded by the number determined in Step (b), i.e., $|\Ac_{\Lin^{\star}}| \geq |\Dc'| +1 \geq n-9t/4-1+1$, for  $\Lin^{\star} \in [1: \gn^{[2]}]$.
  \end{itemize}

 \emph{Step (a):} The first step of the proof is to transform the network into a graph that is within the family of graphs defined in Section~\ref{sec:graph}. We will consider the case of $\gn^{[2]}\geq 1$. 
 Recall that,   when $\gn^{[2]}\geq 1$,  we have $|\Ac_{\Lin}^{[2]}| \geq 1$ for any  $\Lin \in [1: \gn^{[2]}]$, where $\Ac_{\Lin}^{[2]}$ is defined as  $\Ac_{\Lin}^{[2]} =  \{  i:  \Ry_i^{[2]} =1, \Me_i =  \Meg_{\Lin},    i \notin  \Fc,  i \in [1:n]\}$ (see \eqref{eq:Aell03}). 
 Let us consider a fixed $i^{\star}$ for $i^{\star}\in \Ac_{\Lin^{\star}}^{[2]}$ and $\Lin^{\star} \in [1: \gn^{[2]}]$. 
   Based on the definition in \eqref{eq:Aell03}, it holds true that 
  \begin{align}
\Ry_{i^{\star}}^{[2]} =1      \label{eq:sindicator3lm01} 
  \end{align}
   for $i^{\star}\in \Ac_{\Lin^{\star}}^{[2]}$ and $\Lin^{\star} \in [1: \gn^{[2]}]$. 
 Recall that  $\Ry_{i^{\star}}^{[2]}$ denotes the value of the success indicator $\Ry_{i^{\star}}$ updated in   Phase~2.  The event of $\Ry_{i^{\star}}^{[2]} =1$ implies that 
   \begin{align}
\sum_{j =1}^n \Lk_{i^{\star}}^{[2]} (j)  \geq n - t     \label{eq:sindicator3lm02} 
  \end{align}
(see Step~2 of Phase~2 and \eqref{eq:sindicator}),  where $\Lk_{i^{\star}}^{[2]} (j)$  (resp. $\Lk_{i^{\star}}^{[1]} (j)$) denotes the value of the link indicator $\Lk_{i^{\star}} (j)$  updated in Phase~2 (resp. in Phase~1). The event of $\Lk_{i^{\star}}^{[2]} (j) =1$ implies that $\Ry_j^{[1]} =1$, $\Lk_{i^{\star}}^{[1]} (j) =1$,  and  $(y_{i^{\star}}^{(j)}, y_j^{(j)}) = (y_{i^{\star}}^{(i^{\star})}, y_j^{(i^{\star})})$ (see \eqref{eq:lkindicator} and \eqref{eq:Thetaset2826}). 
Since $\Ry_j^{[1]} =0, \forall j \notin \cup_{p=1}^{\gn^{[1]}}\Ac_{p}^{[1]}$, it is true that 
   \begin{align}
\Lk_{i^{\star}}^{[2]} (j) = 0, \quad \forall j \notin \cup_{p=1}^{\gn^{[1]}}\Ac_{p}^{[1]}     \label{eq:sindicator3lm03} 
  \end{align}
 based on  \eqref{eq:Thetaset2826} and \eqref{eq:Aell01}. 
By combining \eqref{eq:sindicator3lm02} and \eqref{eq:sindicator3lm03}, the event of $\Ry_{i^{\star}}^{[2]} =1$ implies that
   \begin{align}
\sum_{j \in \Fc \cup \{\cup_{p=1}^{\gn^{[1]}}\Ac_{p}^{[1]} \}} \Lk_{i^{\star}}^{[2]} (j)  &\geq n - t     \label{eq:sindicator3lm04}  \\
\sum_{j \in \{\cup_{p=1}^{\gn^{[1]}}\Ac_{p}^{[1]}\}\setminus \{i^{\star}\}} \Lk_{i^{\star}}^{[2]} (j)  &\geq n - t -t  -1    \label{eq:sindicator3lm052} 
  \end{align}
  and that
     \begin{align}
\sum_{j \in \{\cup_{p=1}^{\gn^{[1]}}\Ac_{p}^{[1]}\}\setminus \{i^{\star}\}} \Lk_{i^{\star}}^{[1]} (j)  \geq n - t -t  -1    \label{eq:sindicator3lm05} 
  \end{align}
  where the last inequality uses the fact that  $\Lk_{i^{\star}}^{[2]} (j) = \Lk_{i^{\star}}^{[1]} (j),  \forall   j \in \cup_{p=1}^{\gn^{[1]}}\Ac_{p}^{[1]}$. 
  Based on \eqref{eq:sindicator3lm01} and \eqref{eq:sindicator3lm05},  for $i^{\star}\in \Ac_{\Lin^{\star}}^{[2]}$ and $\Lin^{\star} \in [1: \gn^{[2]}]$, it holds true that $\PRO~i^{\star}$ receives at least $n-2t -1$ number of  matched observations, i.e., $\{(y_{i^{\star}}^{(j)}, y_j^{(j)}):    (y_{i^{\star}}^{(j)}, y_j^{(j)}) = (y_{i^{\star}}^{(i^{\star})}, y_j^{(i^{\star})}),   \Lk_{i^{\star}}^{[1]} (j) =1,  j \in \{\cup_{p=1}^{\gn^{[1]}}\Ac_{p}^{[1]}\}\setminus \{i^{\star}\} \}$    (see \eqref{eq:lkindicator}),   from the $\Pros$ in $\{\cup_{p=1}^{\gn^{[1]}}\Ac_{p}^{[1]}\}\setminus \{i^{\star}\}$.  
 Recall that in Phase~1 $\PRO$~$i^{\star}$ sets a link indicator as  $\Lk_{i^{\star}}^{[1]} (j)= 1$   if  the received observation $(y_{i^{\star}}^{(j)}, y_j^{(j)})$  is matched with its  observation $(y_{i^{\star}}^{(i^{\star})}, y_j^{(i^{\star})})$ (see \eqref{eq:lkindicator}).  For $i^{\star}\in \Ac_{\Lin^{\star}}^{[2]}$ and $\Lin^{\star} \in [1: \gn^{[2]}]$, let us define a subset of $\{\cup_{p=1}^{\gn^{[1]}}\Ac_{p}^{[1]}\}\setminus \{i^{\star}\}$ of honest $\Pros$ as 
     \begin{align}
\Cc' \defeq  \{j:    \Lk_{i^{\star}}^{[1]} (j) =1, j \in \{\cup_{p=1}^{\gn^{[1]}}\Ac_{p}^{[1]}\}\setminus \{i^{\star}\}  \} .    \label{eq:cdef01} 
  \end{align}
  The above $\Cc'$ can be interpreted  as a subset of $\{\cup_{p=1}^{\gn^{[1]}}\Ac_{p}^{[1]}\}\setminus \{i^{\star}\}$ of  honest $\Pros$, in which each $\Pro$ sends a matched observation to $\PRO$~$i^{\star}$. 
 Based on \eqref{eq:sindicator3lm05} and \eqref{eq:cdef01}, the following conclusions are true  
 \begin{align}
\Lk_{j}^{[1]} (i^{\star}) &=1,  \quad \forall j  \in  \Cc'  \label{eq:cdef02}  \\
|\Cc'| &\geq     n - 2t  -1.  \label{eq:cdef03} 
 \end{align}
Note that  in our setting it holds true that $ \Lk_{i}^{[1]} (j)  =  \Lk_{j}^{[1]} (i)$, $\forall i, j  \in  \cup_{\Lin=1}^{\gn}\Ac_{\Lin}$ (see \eqref{eq:lkindicator}).

 Since  $\Cc'$ is a subset of $\cup_{p=1}^{\gn^{[1]}}\Ac_{p}^{[1]}$ (see \eqref{eq:cdef01}), it implies   
  \begin{align}
\Ry_j^{[1]} =1, \quad  \forall j \in  \Cc'    \label{eq:cdef04} 
 \end{align}
 based on the definition of $\Ac_{p}^{[1]}$  (see \eqref{eq:Aell01}). 
 The fact in \eqref{eq:cdef04} further implies   
   \begin{align}
\sum_{p =1}^n \Lk_{j}^{[1]} (p)  \geq n - t ,   \quad \forall j    \in  \Cc'   \label{eq:sindicator3lm0211} 
  \end{align}
(see \eqref{eq:sindicator}) and that 
   \begin{align}
\sum_{p  \in  \cup_{\Lin=1}^{\gn}\Ac_{\Lin}}  \Lk_{j}^{[1]} (p)  \geq n - 2t ,   \quad \forall j    \in  \Cc'   \label{eq:sindicator3lm0222} 
  \end{align}
where $\cup_{\Lin=1}^{\gn}\Ac_{\Lin}=[1:n] \setminus \Fc$ (see \eqref{eq:Aell00}).  In other words, for any $j    \in  \Cc'$, $\PRO~j$ receives at least $n-2t$ number of  matched observations, i.e., $\{(y_{j}^{(p)}, y_p^{(p)}):    (y_{j}^{(p)}, y_p^{(p)}) = (y_{j}^{(j)}, y_p^{(j)}), \Lk_{j}^{[1]} (p) =1,  p  \in  \cup_{\Lin=1}^{\gn}\Ac_{\Lin} \}$,  from the $\Pros$ in $\cup_{\Lin=1}^{\gn}\Ac_{\Lin}$.  Let us define a subset of $\{\cup_{\Lin=1}^{\gn}\Ac_{\Lin}\} \setminus \{i^{\star}\}$ of honest $\Pros$ as 
 \begin{align}
\Dc' \defeq \Bigl\{p: \   \sum_{j    \in  \Cc' }  \Lk_{j}^{[1]} (p)    \geq  k , \   p  \in  \{\cup_{\Lin=1}^{\gn}\Ac_{\Lin}\} \setminus \{i^{\star}\} \Bigr\}  \label{eq:graphD01prime}   
 \end{align}
where $k$ is defined in \eqref{eq:qdef}.  
 The above $\Dc'$ can be interpreted as a set of  honest $\Pros$ in which each $\Pro$ sends at least $k$ matched observations to the $\Pros$ in $ \Cc'$.

Now we map the network into a graph by considering the honest $\Pros$ as the vertices and considering the link indicators as edges. 
In our setting, it is true that $ \Lk_{i}^{[1]} (j)  =  \Lk_{j}^{[1]} (i)$ for $i, j  \in  \cup_{\Lin=1}^{\gn}\Ac_{\Lin}$ (see \eqref{eq:lkindicator}). 
Let $\Pc' \defeq \cup_{\Lin=1}^{\gn}\Ac_{\Lin}$ and let $E_{i,j}= \Lk_{i}^{[1]} (j), \forall i, j  \in \Pc'$. 
At this point, the network can be transformed into a graph $G=(\Pc', \Ec')$, where $\Pc'$ is a set of $n -t$  vertices, and $\Ec'$ is a set of edges defined by $E_{i,j}= \Lk_{i}^{[1]} (j), \forall i, j  \in \Pc'$, i.e.,  $\Lk_{i}^{[1]} (j)= 1$ indicates that an edge exists between vertex $i$ and vertex $j$. 
For the graph $G=(\Pc', \Ec')$ considered here,  there exists a set  $\Cc' \subseteq \Pc' \setminus \{i^{\star}\}$ such that the  conditions in  \eqref{eq:cdef02}, \eqref{eq:cdef03} and \eqref{eq:sindicator3lm0222}  (similar to the conditions in \eqref{eq:graph01}, \eqref{eq:graph03} and \eqref{eq:graph02} respectively) are satisfied, for a given $i^{\star}\in \Ac_{\Lin^{\star}}^{[2]} \subseteq \Pc'$.  
This graph $G=(\Pc', \Ec')$ falls into a family of graphs satisfying \eqref{eq:graph01}-\eqref{eq:graph03} (see Section~\ref{sec:graph}).

\emph{Step (b):} Since  the graph $G=(\Pc', \Ec')$ falls into a family of graphs satisfying \eqref{eq:graph01}-\eqref{eq:graph03}, the result of Lemma~\ref{lm:graph} reveals that the size of $\Dc'$ defined in \eqref{eq:graphD01prime} (with the same form as in \eqref{eq:graphD01}) satisfies the following inequality
 \begin{align}
|\Dc'| &\geq  n-9t/4-1 .    \label{eq:graphDr11prime}  
 \end{align}

\emph{Step (c):}  Now we argue that every $\Pro$ in $\Dc'$ has the same initial message as $\PRO$~$i^{\star}$. Note that in the graph $G=(\Pc', \Ec')$, every vertex in $\Cc'$ has an edge  connected  with vertex $i^{\star}$  (see \eqref{eq:cdef01}). This identity implies  that $\Lk_{i^{\star}}^{[1]} (j) =1, \forall j \in \Cc'$ and that 
 \begin{align}
  y_j^{(j)}  &=   y_j^{(i^{\star})}   \quad  \forall j \in \Cc'   \label{eq:cdef0122b} 
  \end{align}
  (see \eqref{eq:lkindicator}). 
Also note that in the graph $G=(\Pc', \Ec')$, every vertex in  $\Dc'$ is connected with at least $k$ vertices in $\Cc'$ (see \eqref{eq:graphD01prime}). For $p\in \Dc'$, let us define the set of vertices in  $\Cc'$ connected with vertex $p$ as 
 \begin{align}
\Cc_p' \defeq \Bigl\{j: \     \Lk_{j}^{[1]} (p) =1 ,    \   j  \in  \Cc' \Bigr\}, \quad  p\in \Dc'    \label{eq:graphCj01}   
 \end{align}
  which satisfies the following condition 
 \begin{align}
|\Cc_p'| \geq k,   \quad  \forall p\in \Dc'    \label{eq:graphCj01eq}   
 \end{align}
 based on the definition in \eqref{eq:graphD01prime}. 
For any $p\in \Dc'$,  since  every vertex in $\Cc_p'$ is connected  with vertex $p$ (equivalently,  $\Lk_{j}^{[1]} (p) =1, \forall j \in \Cc_p'$), it is true that 
  \begin{align}
 y_{j}^{(p)}  &= y_{j}^{(j)}     \quad  \forall j \in \Cc_p' , \   p\in \Dc'  \label{eq:cdef0122aj} 
  \end{align}
  (see \eqref{eq:lkindicator}). 
 Since $\Cc_p'$ is a subset of $\Cc'$, the result in \eqref{eq:cdef0122b} gives 
  \begin{align}
  y_j^{(j)}  &=   y_j^{(i^{\star})}   \quad  \forall j \in \Cc_p'  , \   p\in \Dc'   \label{eq:cdef0122b33} 
  \end{align}
 which,  together with \eqref{eq:cdef0122aj}, implies that 
   \begin{align}
    y_j^{(i^{\star})}  =y_{j}^{(p)}  \quad  \forall j \in \Cc_p' , \   p\in \Dc' .   \label{eq:cdef0122bff} 
  \end{align}
 Based on the definition in \eqref{eq:yi11}, the conclusion in   \eqref{eq:cdef0122bff} can be rewritten as 
   \begin{align}
  \hv_j^\T ( \Me_{i^{\star}} - \Me_{p})   = 0 \quad  \forall j \in \Cc_p' , \   p\in \Dc'    \label{eq:cdefhequ} 
  \end{align}
  or 
     \begin{align}
  \Hm_{p} ( \Me_{i^{\star}} - \Me_{p})   = {\bf 0} \quad  \forall   p\in \Dc'    \label{eq:cdefhequH} 
  \end{align}
where $ \Hm_{p}$ is  an  $|\Cc_p'| \times k$ matrix  defined as   $\Hm_{p}\defeq  \Bmatrix{ \hv_{i_1}, \  \hv_{i_2}, \ \cdots ,\  \hv_{i_{|\Cc_p'|}} }^\T $ for $i_{1}, i_{2}, \cdots,  i_{|\Cc_p'|}    \in  \Cc_p'$ and     $i_{1} < i_{2} < \cdots < i_{|\Cc_p'|}$.    
  Since all of the rows of $\Hm_{p}$ are different and generated as in \eqref{eq:zidefh}, $\Hm_{p}$ is a full rank matrix.  
Given that $ \Hm_{p}$ is  an  $|\Cc_p'| \times k$ full matrix, and with the result of $|\Cc_p'| \geq k$,   $\forall p\in \Dc' $   (see \eqref{eq:graphCj01eq}), it is true that the only solution to  \eqref{eq:cdefhequH} is    
\begin{align}
\Me_{p} =  \Me_{i^{\star}}   \quad  \forall   p\in \Dc'   . \label{eq:cdefhequHs} 
  \end{align}
 Therefore, every $\Pro$ in $\Dc'$ has the same initial message as $\PRO$~$i^{\star}$, for  $i^{\star} \in \Ac_{\Lin^{\star}}^{[2]}$ and $\Lin^{\star} \in [1: \gn^{[2]}]$.
 
 \emph{Step (d):}  Since every $\Pro$ in $\Dc'$ has the same initial message as $\PRO$~$i^{\star}$ (see \eqref{eq:cdefhequHs}) for  $i^{\star} \in \Ac_{\Lin^{\star}}^{[2]}$,  it can be concluded  that $\Dc' \cup \{i^{\star}\}$ is a subset of $\Ac_{\Lin^{\star}}$, i.e., 
 \begin{align}
\Dc' \cup \{i^{\star}\} \subseteq \Ac_{\Lin^{\star}}    \label{eq:DAc} 
  \end{align}
   for  $\Lin^{\star} \in [1: \gn^{[2]}]$. Note that the $\Pros$ in $\Ac_{\Lin^{\star}}^{[2]}$ and the $\Pros$ in $\Ac_{\Lin^{\star}}$ have the same initial message (see \eqref{eq:Aell00} and \eqref{eq:Aell03}). 
Finally, by combining the result in \eqref{eq:DAc} and the result in \eqref{eq:graphDr11prime}, i.e., $|\Dc'| \geq  n-9t/4-1$, it can be concluded that 
 \begin{align}
|\Ac_{\Lin^{\star}}| \geq |\Dc'| +1 \geq n-9t/4   \label{eq:DAcsbound} 
  \end{align}
   for $\Lin^{\star} \in [1: \gn^{[2]}]$. The above result holds true for any $i^{\star} \in \Ac_{\Lin^{\star}}^{[2]}$ and  any $\Lin^{\star} \in [1: \gn^{[2]}]$, when $\gn^{[2]}\geq 1$.  At this point we complete the proof of this lemma. 
\end{proof}

 \subsection{Lemma~\ref{lm:eta1231} and its proof}  \label{sec:eta1231}

In this sub-section we  provide a lemma that will be used later for the analysis of the proposed protocol. 
Lemma~\ref{lm:sizeboundMatrix} will be used in this proof.

\begin{lemma}    \label{lm:eta1231}
For the proposed  COOL with $n\geq 3t+1$, if $\gn^{[1]} = 2$  then it holds true that $\gn^{[3]} \leq   1$.
\end{lemma}
 \begin{proof}
Given $\gn^{[1]} = 2$,   the definitions in \eqref{eq:Aell00}-\eqref{eq:All11} imply that 
 \begin{align}
 \Ac_{1}^{[1]} =&   \{  i:  \Ry_i^{[1]} =1, \Me_i =  \Meg_{1},  \  i \notin  \Fc, \ i \in [1:n]\}       \label{eq:Aellbig10455}  \\ 
  \Ac_{2}^{[1]} =&   \{  i:  \Ry_i^{[1]} =1, \Me_i =  \Meg_{2},  \  i \notin  \Fc, \ i \in [1:n]\}      \label{eq:Aellbig1045522}    \\ 
\Ac_{1,2}^{[1]} =   & \{  i:  \   i\in  \Ac_{1}^{[1]},   \   \hv_i^\T  \Meg_{1}  = \hv_i^\T  \Meg_2 \}    \label{eq:Alj2995}  \\
\Ac_{1,1}^{[1]} =   & \Ac_{1}^{[1]} \setminus \Ac_{1,2}^{[1]} = \{  i:  \   i\in  \Ac_{1}^{[1]},   \   \hv_i^\T  \Meg_{1}  \neq  \hv_i^\T  \Meg_2 \}  \label{eq:All2955251} \\
\Ac_{2,1}^{[1]} =   & \{  i:  \   i\in  \Ac_{2}^{[1]},   \   \hv_i^\T  \Meg_{2}  = \hv_i^\T  \Meg_1 \}   \label{eq:Alj299535}  \\
\Ac_{2,2}^{[1]} =   & \Ac_{2}^{[1]} \setminus \Ac_{2,1}^{[1]} = \{  i:  \   i\in  \Ac_{2}^{[1]},   \   \hv_i^\T  \Meg_{2}  \neq  \hv_i^\T  \Meg_1 \}  \label{eq:All29559386}  \\
\Bc^{[1]} =  & \{  i:  \Ry_i^{[1]} =0, \  i \notin  \Fc, \ i \in [1:n] \}=  \{  i:    i \in  [1:n], \ i\notin  \Fc \cup  \Ac_{1}^{[1]} \cup  \Ac_{2}^{[1]}  \}   .  \label{eq:BdefB1}         
 \end{align}

 In the following we will complete the proof by focusing on the following three cases
\begin{align}
\text{Case~1:} \quad | \Ac_{1}^{[1]}|  +     |\Bc^{[1]}| \geq  & t+1     \label{eq:A1reqset10198g} \\
  | \Ac_{2}^{[1]}|  +    |\Bc^{[1]}|  < &   t +1   \label{eq:A1reqset10299l} 
 \end{align}
 \begin{align}
 \text{Case~2:} \quad  | \Ac_{1}^{[1]}|  +     |\Bc^{[1]}| <   &  t +1 \label{eq:A1reqset10198l} \\
  | \Ac_{2}^{[1]}|  +    |\Bc^{[1]}|  \geq &  t+1  \label{eq:A1reqset10299g} 
 \end{align}
 \begin{align}
\text{Case~3:} \quad   | \Ac_{1}^{[1]}|  +     |\Bc^{[1]}| \geq   &  t+1   \label{eq:A1reqset10198gg} \\
  | \Ac_{2}^{[1]}|  +    |\Bc^{[1]}|  \geq & t+1 .    \label{eq:A1reqset10299gg} 
 \end{align}
 Note that  the following case 
  \begin{align}
\text{Case~4:} \quad   | \Ac_{1}^{[1]}|  +     |\Bc^{[1]}| <   &  t +1 \label{eq:A1reqset10198lll} \\
  | \Ac_{2}^{[1]}|  +    |\Bc^{[1]}|  < &  t+1  \label{eq:A1reqset10299lll} 
 \end{align}
 does not exist. The argument  is given as follows. Given the assumption in \eqref{eq:A1reqset10198lll} (which is equivalent to $ | \Ac_{1}^{[1]}|  +     |\Bc^{[1]}| \leq    t $),   it then implies that 
   \begin{align}
  |\Ac_{2}^{[1]}|   = & n - |\Fc| - |\Ac_{1}^{[1]}|-   |\Bc^{[1]}|  \label{eq:A1reqset13431}  \\
  \geq  & n - |\Fc| - t  \label{eq:A1reqset28859}  \\
   \geq  & t+1   \label{eq:A1reqset78285}  
 \end{align}
 where  \eqref{eq:A1reqset13431} follows from \eqref{eq:sumn01};
 \eqref{eq:A1reqset28859}  uses the assumption of  $ | \Ac_{1}^{[1]}|  +     |\Bc^{[1]}| \leq    t $  as in \eqref{eq:A1reqset10198lll};
 \eqref{eq:A1reqset78285} results from the condition of $n\geq 3t+1$ and $|\Fc|=t$.  The conclusion in \eqref{eq:A1reqset78285} contradicts with \eqref{eq:A1reqset10299lll}, which suggests that the case described in \eqref{eq:A1reqset10198lll} and \eqref{eq:A1reqset10299lll} does not exist.

\subsubsection{Analysis for Case~1}  Let us consider  Case~1 described in \eqref{eq:A1reqset10198g} and \eqref{eq:A1reqset10299l}.  
Recall that  in the first step of Phase~2, $\PRO$~$i$, $i  \in \Ss_1$,  sets  $\Lk_i^{[2]} (j) = 0$,  $\forall j \in \Ss_0$, where $\Lk_i^{[p]} (j)$  denotes the value of $\Lk_i (j)$ updated in Phase~$p$, $p\in \{1,2,3\}$.   In this step, $\Bc^{[1]}$ is in the list of  $\Ss_0$ because 
 \begin{align}
 \Ry_i^{[1]} =0, \quad \forall i\in \Bc^{[1]}  \label{eq:BdefBaa} 
 \end{align}
 (see \eqref{eq:BdefB1}), while $\Ac_{2,2}^{[1]}$ is in the list of  $\Ss_1$ (see \eqref{eq:All29559386} and \eqref{eq:vr0}). Therefore, in this step $\PRO$~$i$ sets 
 \begin{align}
\Lk_i^{[2]} (j) = 0 ,  \quad   \forall j \in \Bc^{[1]},   \ i  \in \Ac_{2,2}^{[1]}.   \label{eq:Thetaset2826aa} 
 \end{align}
Furthermore, it is true that
 \begin{align}
\Lk_i^{[2]} (j) = \Lk_i^{[1]} (j) = 0 ,   \quad  \forall j \in  \Ac_{1},  \   i  \in \Ac_{2,2}^{[1]}    \label{eq:A1A220} 
 \end{align}
 because of the condition $\hv_i^\T  \Meg_{2}  \neq  \hv_i^\T  \Meg_1, \forall i \in \Ac_{2,2}^{[1]}$ (see \eqref{eq:All29559386}),  indicating that  $ (y_i^{(j)}, y_j^{(j)}) \neq (y_i^{(i)}, y_j^{(i)})$, $\forall j \in  \Ac_{1},   i  \in \Ac_{2,2}^{[1]}$  (see \eqref{eq:yi11} and \eqref{eq:lkindicator}). 
 By combining \eqref{eq:Thetaset2826aa} and \eqref{eq:A1A220}, it gives 
  \begin{align}
\Lk_i^{[2]} (j) = 0 ,   \quad  \forall j \in \Ac_{1}\cup\Bc^{[1]}, \   i  \in \Ac_{2,2}^{[1]}.   \label{eq:B1A1A220} 
 \end{align} 
With the condition in \eqref{eq:A1reqset10198g},  i.e.,  $ | \Ac_{1}^{[1]}|  +     |\Bc^{[1]}| \geq   t+1$, and with  \eqref{eq:B1A1A220}, it implies that  in the second step of Phase~2, $\PRO$~$i$ sets 
   \begin{align}
\Ry_i^{[2]} =0 , \quad  \forall    i  \in \Ac_{2,2}^{[1]}  \label{eq:eindicator112} 
 \end{align}
 (see \eqref{eq:sindicator}), meaning that the number of mismatched observations is more than $t$. 
 With the outcome in \eqref{eq:eindicator112} and after exchanging the success indicators, the set of $\Ac_{2,2}^{[1]}$, as well as  $\Bc^{[1]}$, will be in the list of $\Ss_0$ at each honest $\Pro$ at the end of Phase~2. 
Note that a subset of $\Ac_{2,1}^{[1]}$ (i.e., $\Ac_{2,1}^{[1]}\cap \{p: \Ry_p^{[2]} =0\}$) is in the list of $\Ss_0$, while the rest of $\Ac_{2,1}^{[1]}$ (i.e., $\Ac_{2,1}^{[1]}\cap \{p: \Ry_p^{[2]} =1\}$) is still  in list of $\Ss_1$ at the end of Phase~2. 
Below we will argue that the complete set of $\Ac_{2,1}^{[1]}$ will be in the list of $\Ss_0$  at the end of Phase~3.

 In the first step of Phase~3, $\PRO$~$i$, $i  \in \Ss_1$,  sets  $\Lk_i^{[3]} (j) = 0$,  $\forall j \in \Ss_0$. In this step, since $\Ac_{2,2}^{[1]} \subseteq \Ss_0$ and $\{\Ac_{2,1}^{[1]}\cap \{p: \Ry_p^{[2]} =1\} \}\subseteq  \Ss_1$,  $\PRO$~$i$ sets  
 \begin{align}
\Lk_i^{[3]} (j) = 0 ,   \   \forall j \in \Ac_{2,2}^{[1]}, \   i  \in  \Ac_{2,1}^{[1]}\cap \{p: \Ry_p^{[2]} =1\}.   \label{eq:Thetaset2826P3bb} 
 \end{align}
Furthermore, it is true that
 \begin{align}
\Lk_i^{[3]} (j)  = 0 ,   \quad  \forall j \in \Bc^{[1]}, \    i  \in  \Ac_{2,1}^{[1]}\cap \{p: \Ry_p^{[2]} =1\}     \label{eq:A1A220p3} 
 \end{align}
 (see \eqref{eq:BdefBaa}). 
It is also true that
 \begin{align}
\Lk_i^{[3]} (j) = \Lk_i^{[2]} (j)= \Lk_i^{[1]} (j) = 0 ,   \quad  \forall j \in \Ac_{1,1}^{[1]}, \     i  \in  \Ac_{2,1}^{[1]}\cap \{p: \Ry_p^{[2]} =1\}     \label{eq:A11A2} 
 \end{align}
 because of the condition $\hv_j^\T  \Meg_{1}  \neq  \hv_j^\T  \Meg_2, \forall j \in \Ac_{1,1}^{[1]}$ (see \eqref{eq:All2955251}),  indicating that  $ (y_i^{(j)}, y_j^{(j)}) \neq (y_i^{(i)}, y_j^{(i)})$, $\forall j \in  \Ac_{1,1}^{[1]},   i  \in  \{\Ac_{2,1}^{[1]}\cap \{p: \Ry_p^{[2]} =1\}\}   \subseteq \Ac_{2}$  (see \eqref{eq:yi11} and \eqref{eq:lkindicator}). 
 By combining \eqref{eq:Thetaset2826P3bb}-\eqref{eq:A11A2}, it gives 
   \begin{align}
\Lk_i^{[3]} (j) = 0 ,   \quad  \forall j \in \Ac_{1,1}^{[1]}\cup\Ac_{2,2}^{[1]}\cup\Bc^{[1]}, \    i  \in  \Ac_{2,1}^{[1]}\cap \{p: \Ry_p^{[2]} =1\}.   \label{eq:B1A11A22atA21} 
 \end{align} 
 The size of $\Ac_{1,1}^{[1]}\cup\Ac_{2,2}^{[1]}\cup\Bc^{[1]}$ can be bounded by 
   \begin{align}
|\Ac_{1,1}^{[1]}\cup\Ac_{2,2}^{[1]}\cup\Bc^{[1]} |  = &  |\Ac_{1,1}^{[1]}| +  |\Ac_{2,2}^{[1]}| +  |\Bc^{[1]}|       \label{eq:orp3aa6254} \\
     = &n- |\Fc| -   |\Ac_{1,2}^{[1]}|   - |\Ac_{2,1}^{[1]}|    \label{eq:orp3aa83876} \\
     \geq & n- |\Fc| -  (k-1)   \label{eq:orp3aa9775} \\
      \geq & 2t+1 - (k-1)    \label{eq:orp3aa43636} \\
      \geq &t+1    \label{eq:orp3aa9187875} 
 \end{align}
where \eqref{eq:orp3aa6254} uses the disjoint property between $\Ac_{1,1}^{[1]}$,    $\Ac_{2,2}^{[1]}$ and $\Bc^{[1]}$; 
\eqref{eq:orp3aa83876} is from \eqref{eq:sumn01} and the disjoint property between $\Ac_{1,1}^{[1]}$,   $\Ac_{1,2}^{[1]}$,  $\Ac_{2,1}^{[1]}$, $\Ac_{2,2}^{[1]}$ and $\Bc^{[1]}$;
\eqref{eq:orp3aa9775} follows from Lemma~\ref{lm:sizeboundMatrix},  which implies that $ |\Ac_{1,2}^{[1]}|   + |\Ac_{2,1}^{[1]}| <  k$ $($or equivalently $|\Ac_{1,2}^{[1]}|   + |\Ac_{2,1}^{[1]}| \leq   k -1)$;
\eqref{eq:orp3aa43636} uses the condition that $n \geq 3t+1$ and $|\Fc| =t$;
\eqref{eq:orp3aa9187875} results from the fact that $t \geq k-1$ based on our design of $k$ (see \eqref{eq:qdef}).
From \eqref{eq:B1A11A22atA21} and \eqref{eq:orp3aa9187875}, it suggests that the number of  mismatched observations  is  more than $t$ at  $\PRO$~$i$, $\forall i \in  \Ac_{2,1}^{[1]}\cap \{p: \Ry_p^{[2]} =1\}$. Therefore, at the second step of Phase~3 $\PRO$~$i$    updates its success indicator (see \eqref{eq:sindicator}) as
  \begin{align}
\Ry_i^{[3]} =0 , \quad  \forall  i \in  \Ac_{2,1}^{[1]}\cap \{p: \Ry_p^{[2]} =1\} . \label{eq:p321aa} 
\end{align}

With \eqref{eq:p321aa} and \eqref{eq:eindicator112}, and with the fact $\Ac_{2,1}^{[1]}\cap \{p: \Ry_p^{[2]} =0\} \subseteq \Ss_0$, it can be concluded that the complete set of $\Ac_{2}^{[1]} =\Ac_{2,1}^{[1]} \cup \Ac_{2,2}^{[1]}$ is in the list of $\Ss_0$, that is, 
  \begin{align}
\Ac_{2}^{[1]}  \subseteq \Ss_0  \label{eq:case1A2S0} 
\end{align}
  at the end of Phase~3.  In other words,    at the end of Phase~3 of COOL  there exists \emph{at most}  1  group of honest $\Pros$, where the honest $\Pros$ within this group  have   the same  non-empty updated  message (with  success indicators as ones), and  the  honest $\Pros$ outside this group  have  the same  empty updated  message (with  success indicators as zeros), that is, $\gn^{[3]} \leq 1$, for Case~1.

\subsubsection{Analysis for Case~2}  
Due to the symmetry between Case~1 and Case~2, one can follow from the proof steps for Case~1 and interchange the roles of Groups $\Ac_{1}$ and  $\Ac_{2}$ (as well as the roles of Groups  $\Ac_{1}^{[p]}$ and $\Ac_{2}^{[p]}$ accordingly for $p\in\{1,2,3\}$),  to show for Case~2 that  at the end of Phase~3 the following conclusion is true 
   \begin{align}
\Ac_{1}^{[1]}  \subseteq \Ss_0.  \label{eq:case1A1S0} 
\end{align}
  In other words,  for Case~2,   at the end of Phase~3 of COOL  there exists \emph{at most}  1  group of honest $\Pros$, where the honest $\Pros$ within this group  have   the same  non-empty updated  message, and  the  honest $\Pros$ outside this group  have  the same  empty updated  message.
   Note that   for Case~2 the condition we use for deriving \eqref{eq:case1A1S0} is  \eqref{eq:A1reqset10299g}, while for Case~1 the condition we use for deriving \eqref{eq:case1A2S0} is  \eqref{eq:A1reqset10198g}.

\subsubsection{Analysis for Case~3}   Note that the condition in \eqref{eq:A1reqset10198gg} for Case~3 is the same as the condition in   \eqref{eq:A1reqset10198g} for Case~1.  By following the proof steps for Case~1, one can  show  for Case~3 that  the following conclusion is true
   \begin{align}
\Ac_{2}^{[1]}  \subseteq \Ss_0     \label{eq:p321aaf33} 
 \end{align}
   at the end of Phase~3 of  the  proposed COOL. 
   Also note that  the condition in \eqref{eq:A1reqset10299gg} for Case~3 is the same as the condition in   \eqref{eq:A1reqset10299g} for Case~2. Thus, one can also show  for Case~3 that the following conclusion is true
   \begin{align}
\Ac_{1}^{[1]}  \subseteq \Ss_0     \label{eq:p321aaf11333} 
 \end{align}
   at the end of Phase~3 of  the  proposed COOL. 
  In other words,  for Case~3,  none of the honest $\Pros$ could  have a non-zero success indicator  (or a non-empty value on the updated message)  at the end of Phase~3, that is, $\gn^{[3]} = 0$. 
  
By combining the analysis for the above three cases,  at this point we complete the proof of Lemma~\ref{lm:eta1231}.  
  \end{proof}

 \subsection{Proof of Lemma~\ref{lm:ph32group}}  \label{sec:ph32group}

We are now ready to prove Lemma~\ref{lm:ph32group}. Specifically we will prove that, given $n\geq 3t+1$, at the end of Phase~3 of COOL there exists \emph{at most}  1  group of honest $\Pros$,  
where the honest $\Pros$ within this group  have the same  non-empty updated  message, and  the  honest $\Pros$ outside this group  have  the same  empty updated  message. Based on our definition in \eqref{eq:Aell03}, it is equivalent to prove   
\[  \gn^{[3]} \leq 1 .\] 
In the first step we will prove  $\gn^{[2]} \leq 2$, which implies  $\gn^{[3]} \leq \gn^{[2]} \leq 2$ (using the fact that $\gn^{[3]} \leq \gn^{[2]}$, see \eqref{eq:Aell00}-\eqref{eq:Aell03new}). Clearly,  it is true  $\gn^{[3]} \leq 1$ when  $\gn^{[2]} \leq 1$.  
Then, in the second step we will prove   $\gn^{[3]} \leq 1$ when  $\gn^{[2]} = 2$.  
In our proof we will use the results of Lemma~\ref{lm:sizeboundMatrix}, Lemma~\ref{lm:sizem} and Lemma~\ref{lm:eta1231},   and will use the  proof by contradiction method. 
Proof by contradiction is one type of proof that establishes the truth or the validity of a
claim. The approach is to show that assuming the claim to be false leads to a
contradiction. 
The following lemma summarizes the result of the first step. 

\begin{lemma}    \label{lm:eta3bound2}
For the proposed  COOL with $n\geq 3t+1$, it holds true that $\gn^{[2]} \leq 2$.
\end{lemma}
 \begin{proof}
  Lemma~\ref{lm:sizem} and  proof by contradiction will be used in this proof.  Let us first assume that the claim in Lemma~\ref{lm:eta3bound2} is false.  Specifically let us assume that 
  \begin{align}
\gn^{[2]} \geq 3 .     \label{eq:etagn3} 
 \end{align}
Based on this assumption, and from Lemma~\ref{lm:sizem}, we have 
 \begin{align}
|\Ac_{1}| &\geq n-9t/4     \label{eq:acineq1} \\
|\Ac_{2}| &\geq n-9t/4     \label{eq:acineq2} \\
&\  \vdots       \label{eq:acineq3} \\
|\Ac_{\gn^{[2]}}| &\geq n-9t/4 .    \label{eq:acineq4} 
 \end{align}
 By combining the above $\gn^{[2]}$ bounds together we have 
  \begin{align}
\sum_{\Lin=1}^{\gn^{[2]}} |\Ac_{\Lin}| &\geq \gn^{[2]}(n-9t/4) \geq  3(n-9t/4) = (n - t) + (2n - 23t/4) >  (n - t)  \label{eq:acineq11} 
 \end{align}
 where the second inequality uses the assumption in  \eqref{eq:etagn3}; and the last inequality stems from the derivation that $2n - 23t/4 \geq  6t+2 - 23t/4  >0$ by using the condition of $n\geq 3t+1$. 
One can see that  the conclusion in \eqref{eq:acineq11} contradicts with the identify of $\sum_{\Lin=1}^{\gn^{[2]}} |\Ac_{\Lin}|  \leq \sum_{\Lin=1}^{\gn} |\Ac_{\Lin}|   =  n - t$ (see \eqref{eq:sumn001}), i.e.,  the total number of honest $\Pros$ should be $n - t$. Therefore, the assumption in \eqref{eq:etagn3} leads to a contradiction and thus $\gn^{[2]} $ should be bounded by  $\gn^{[2]} \leq 2$.
 \end{proof}

The result of Lemma~\ref{lm:eta3bound2}  reveals that $\gn^{[3]} \leq \gn^{[2]} \leq 2$. 
Since  $\gn^{[3]} \leq 1$ is true when  $\gn^{[2]} \leq 1$,  in the second step of the proof for Lemma~\ref{lm:ph32group}, we  will prove   $\gn^{[3]} \leq 1$ when  $\gn^{[2]} = 2$.  The result is summarized in the following Lemma~\ref{lm:eta3bound1} and the proof will use the results of  Lemma~\ref{lm:eta3212} (see below) and Lemma~\ref{lm:eta1231} (see Section~\ref{sec:eta1231}). 
 
 \begin{lemma}    \label{lm:eta3bound1}
For the proposed  COOL with $n\geq 3t+1$, it holds true that $\gn^{[3]} \leq 1$ when  $\gn^{[2]} = 2$.
\end{lemma}
 \begin{proof}
 Given $\gn^{[2]} = 2$, it is concluded from   Lemma~\ref{lm:eta3212} (see below) that $\gn^{[1]} = 2$. 
Furthermore, given $\gn^{[1]} = 2$,  it is  concluded from  Lemma~\ref{lm:eta1231} (see Section~\ref{sec:eta1231})  that $\gn^{[3]} \leq 1$. We then conclude that $\gn^{[3]} \leq 1$ when  $\gn^{[2]} = 2$. 
\end{proof}

\begin{lemma}    \label{lm:eta3212}
For the proposed  COOL with $n\geq 3t+1$,  it is true that $\gn^{[1]} = 2$ when $\gn^{[2]} = 2$.
\end{lemma}
 \begin{proof}
Lemma~\ref{lm:sizeboundMatrix} and Lemma~\ref{lm:sizem} will be used in this proof. 
Recall that $\gn^{[1]}$ denotes the number of the sets defined as in \eqref{eq:Aell01}, for $\gn^{[1]} \geq \gn^{[2]} \geq \gn^{[3]}$. Let us assume that  the claim in Lemma~\ref{lm:eta3212} is false.  Specifically, given $\gn^{[2]} = 2$, let us assume that $\gn^{[1]} > 2$.

Given $\gn^{[2]} = 2$, and from Lemma~\ref{lm:sizem}, we have $|\Ac_{1}| \geq n-9t/4$ and $|\Ac_{2}| \geq n-9t/4$, which imply that 
 \begin{align}
\sum_{\Lin=3}^{\gn} |\Ac_{\Lin}| = n- |\Fc| -|\Ac_{1}| -|\Ac_{2}|  \leq  n-t  -  2(n-9t/4)  \leq t/2 -1   \label{eq:acineq1122} 
 \end{align}
where the   equality  results from \eqref{eq:sumn001}; and the last inequality uses the condition of $n\geq 3t+1$. 

If $\gn^{[1]} > 2$, then there exists an $i^{\star} \in \Ac_{\Lin^{\star}}^{[1]} \subseteq \Ac_{\Lin^{\star}}$ for $\Lin^{\star} \geq 3$, such that   
\begin{align}
\Ry_{i^{\star}}^{[1]} =1      \label{eq:sindicator1p11} 
  \end{align}
(see \eqref{eq:Aell01})  and that 
  \begin{align}
\sum_{j \in \cup_{\Lin=1}^{\gn}\Ac_{\Lin}  } \Lk_{i^{\star}}^{[1]} (j) & \geq n - t  - t    \label{eq:sindicator1p22} 
 \end{align}
(see \eqref{eq:sindicator}).  
The inequality of \eqref{eq:sindicator1p22} reveals that $\PRO$~$i^{\star}$,  $i^{\star} \in \Ac_{\Lin^{\star}}^{[1]} \subseteq \Ac_{\Lin^{\star}}$,  receives at least $n - 2t $ number of matched observations $\{(y_{i^{\star}}^{(j)}, y_j^{(j)})\}_j$ from the \emph{honest} $\Pros$, where a matched observation of $(y_{i^{\star}}^{(j)}, y_j^{(j)})$ indicates the value of $\Lk_{i^{\star}}^{[1]} (j) =1$ (see \eqref{eq:lkindicator}).
Based on the definition of $\Ac_{\Lin,j}$ in \eqref{eq:Alj}, it is true that $\Lk_{i^{\star}}^{[1]} (j) =0$, $\forall j  \notin \Ac_{\Lin^{\star}}\cup \{ \cup_{\Lin=1, \Lin \neq \Lin^{\star}}^{\gn}\Ac_{\Lin, \Lin^{\star}} \}$.  Therefore, the inequality in \eqref{eq:sindicator1p22} can be rewritten as 
  \begin{align}
\sum_{j \in  \Ac_{\Lin^{\star}}\cup \{ \cup_{\Lin=1, \Lin \neq \Lin^{\star}}^{\gn}\Ac_{\Lin, \Lin^{\star}} \} } \Lk_{i^{\star}}^{[1]} (j)  &  \geq n - t  - t  .    \label{eq:sindicator1p2233} 
 \end{align}
On the other hand, the size of $\Ac_{\Lin^{\star}}\cup \{ \cup_{\Lin=1, \Lin \neq \Lin^{\star}}^{\gn}\Ac_{\Lin, \Lin^{\star}} \}$ can be upper bounded by 
  \begin{align}
 | \Ac_{\Lin^{\star}}\cup \{ \cup_{\Lin=1, \Lin \neq \Lin^{\star}}^{\gn}\Ac_{\Lin, \Lin^{\star}} \} |  & \leq |\Ac_{1, \Lin^{\star}}| +  |\Ac_{2, \Lin^{\star}}| +    \sum_{\Lin=3}^{\gn} |\Ac_{\Lin}|   \label{eq:sindicator1p727}  \\
 & \leq |\Ac_{1, \Lin^{\star}}| +  |\Ac_{2, \Lin^{\star}}| +   t/2 -1   \label{eq:sindicator1p8256}  \\
  & \leq  (k -1) + (k -1) +   t/2-1   \label{eq:sindicator1p7256}  \\
    & = 2 (  \lfloor    t /5   \rfloor    +1 -1) +   t/2 - 1   \label{eq:sindicator1p8866}  \\
          &  \leq     t+1 - t/10  -2  \non  \\
               &  \leq   n-2t  - t/10 -2  \label{eq:sindicator1p998898} \\
                              &  <   n-2t    \label{eq:sindicator1p9988} 
 \end{align}
 for  $i^{\star} \in \Ac_{\Lin^{\star}}^{[1]} \subseteq \Ac_{\Lin^{\star}}$ and $\Lin^{\star} \geq 3$,  where \eqref{eq:sindicator1p727} uses the fact that $ \Ac_{\Lin^{\star}}\cup \{ \cup_{\Lin=1, \Lin \neq \Lin^{\star}}^{\gn}\Ac_{\Lin, \Lin^{\star}} \}  \subseteq  \Ac_{1, \Lin^{\star}}\cup\Ac_{2, \Lin^{\star}} \cup  \{ \cup_{\Lin=3}^{\gn} \Ac_{\Lin} \}$;
 \eqref{eq:sindicator1p8256} follows from the result in \eqref{eq:acineq1122};
 \eqref{eq:sindicator1p7256} stems from the result in Lemma~\ref{lm:sizeboundMatrix}, that is,  $|\Ac_{\Lin,j}| + |\Ac_{j,\Lin}|  <    k$  $($or equivalently $|\Ac_{\Lin,j}| + |\Ac_{j,\Lin}|  \leq   k -1)$,    $ \forall   j \neq \Lin,  \  j, \Lin \in [1:\gn] $   (see \eqref{eq: Aljb00});
 \eqref{eq:sindicator1p8866} uses the definition of $k$ as in  \eqref{eq:qdef}, that is, $ k =    \bigl \lfloor   \frac{ t  }{5 } \bigr\rfloor    +1$;
 \eqref{eq:sindicator1p998898} follows from the condition of $n\geq 3t+1$.

 One can see that the conclusion in \eqref{eq:sindicator1p9988} contradicts with the conclusion in \eqref{eq:sindicator1p2233}.
 Therefore, the assumption of $\gn^{[1]} > 2$ leads to a contradiction and thus $\gn^{[1]}$ should be $\gn^{[1]} = 2$, given $\gn^{[2]} = 2$.
 \end{proof}

  Finally, with the results of  Lemma~\ref{lm:eta3bound2} and Lemma~\ref{lm:eta3bound1},  it is concluded  that   $\gn^{[3]} \leq 1$  for the proposed  COOL with $n\geq 3t+1$.  At this point we complete the proof of Lemma~\ref{lm:ph32group}.

\subsection{Proof of Lemma~\ref{lm:p4samem}}  \label{sec:p4samem}
        
We will prove Lemma~\ref{lm:p4samem} in this sub-section. Specifically, given $n\geq 3t+1$,  we will prove that all of the honest $\Pros$  reach  the same agreement in COOL. 

Recall that in the last step of Phase~3 in COOL, the system runs the one-bit consensus on the $n$ votes $\{\Vr_1, \Vr_2, \cdots, \Vr_n\}$ from all $\Pros$, where $\Vr_i= 1$ can be considered as a vote from  $\PRO$~$i$  for going to Phase~4, while $\Vr_i= 0$ can be considered as a vote from  $\PRO$~$i$  for stopping in Phase~3, for $i\in [1:n]$ (see \eqref{eq:vindicator}).  
The one-bit consensus  from \cite{BGP:92,CW:92} ensures that:   every honest $\Pro$  eventually outputs a consensus value;  all honest $\Pros$ should have the same consensus output; and if  all honest $\Pros$ have the same vote then the consensus should be the same as the vote of the honest $\Pros$. 
We will prove this lemma by considering each of the following two cases.  
\begin{itemize}
\item  Case~(a):  In Phase~3, the consensus of the votes $\{\Vr_1, \Vr_2, \cdots, \Vr_n\}$ is $0$.
\item  Case~(b):  In Phase~3, the consensus of the votes $\{\Vr_1, \Vr_2, \cdots, \Vr_n\}$ is $1$. 
\end{itemize}

\emph{Analysis for Case~(a):}   
  In Phase~3 of COOL, if the consensus of the votes $\{\Vr_1, \Vr_2, \cdots, \Vr_n\}$ is $0$, then all of the honest $\Pros$  will  set $\Me^{(i)} =\phi$ and consider $\phi$ as a final consensus and stop here.  In this case,  all of  the honest $\Pros$ agree on the same message. 
  
\emph{Analysis for Case~(b):}   
   In Phase~3 of COOL,  if the consensus of the votes $\{\Vr_1, \Vr_2, \cdots, \Vr_n\}$ is $1$, then all of the honest $\Pros$ will  go to Phase~4.  In this case, at least one of the honest $\Pros$  votes \[\Vr_i= 1, \quad  \text{for some} \quad i \notin \Fc.\] Otherwise, all of the honest $\Pros$ vote the same value such that $\Vr_i= 0$, $\forall i \notin \Fc$ and the consensus of the votes $\{\Vr_1, \Vr_2, \cdots, \Vr_n\}$ should be $0$, contradicting the condition of this case.  
In COOL, the condition of voting $\Vr_i =1 $ (see \eqref{eq:vindicator})
is that  $\PRO$~$i$ receives no less than $2t+1$ number of ones from $n$ success indicators $\{ \Ry_1^{[3]}, \Ry_2^{[3]}, \cdots, \Ry_n^{[3]} \}$, i.e., 
   \begin{align}
\sum_{j=1}^{n} \Ry_j^{[3]} \geq 2t+1 .   \label{eq:si2t1} 
 \end{align} 
Since $t$ dishonest $\Pros$ might send the success indicators as ones, the outcome in \eqref{eq:si2t1} implies that at least $t+1$ honest $\Pros$ send out success indicators as ones in Phase~3, i.e., 
  \begin{align}
\sum_{j \in \cup_{\Lin=1}^{\gn^{[3]}}\Ac_{\Lin}^{[3]}  }  \Ry_j^{[3]} \geq t+1 .   \label{eq:si2t122} 
 \end{align}

 Lemma~\ref{lm:ph32group} reveals that at the end of Phase~3 there exists \emph{at most}  1  group of honest $\Pros$, where the honest $\Pros$ within this group  have   the same  non-empty updated  message (with  success indicators as ones), and  the  honest $\Pros$ outside this group  have  the same  empty updated  message (with  success indicators as zeros), that is, $\gn^{[3]} \leq 1$. 
  Then, for this Case~(b), the conclusions in  \eqref{eq:si2t122} and  Lemma~\ref{lm:ph32group}  imply that at the end of Phase~3 there exists \emph{exactly}  1  group of \emph{honest} $\Pros$ with group size satisfying 
     \begin{align}
\text{group size} \  \geq t+1    \label{eq:si2t1285} 
 \end{align}
 where the honest $\Pros$ within this group  have   the same  non-empty updated  message, and  the  honest $\Pros$ outside this group  have  the same  empty updated  message.  In other words, for this Case~(b), we have 
   \begin{align}
   \gn^{[3]} &= 1    \label{eq:A1g2t1009a} \\
   |\Ac_1^{[3]}|  &\geq    t+1    \label{eq:A1g2t1} \\
       \Me_i &=\Meg_{1} , \quad  \forall i \in  \Ac_1^{[3]}  \label{eq:A1g2t1009} 
 \end{align}
 where \eqref{eq:A1g2t1009a} and \eqref{eq:A1g2t1} follow from \eqref{eq:si2t122}   and  Lemma~\ref{lm:ph32group};
  \eqref{eq:A1g2t1009} stems from the definition of $\Ac_{\Lin}^{[3]}$ in \eqref{eq:Aell03new}.

Recall that in Phase~4  $\PRO$~$i$, $i \in \Ss_0$, updates the value of $y_i^{(i)}$ as 
 \begin{align}
y_i^{(i)}    \leftarrow  \text{Majority}( \{y_i^{(j)}:   j \in \Ss_1\})     \non 
 \end{align}
 based on the majority rule, where the symbols of $y_i^{(j)}$ were received in Phase~1. In this step it is true that  $\Ac_1^{[3]} \subseteq \Ss_1$  and $ |\Ac_1^{[3]}|  >   |\Fc|$ (see \eqref{eq:A1g2t1}), which guarantees that $\PRO$~$i$, $i \in \Ss_0\setminus \Fc$, could use  the majority rule to update the value of $y_i^{(i)}$  as 
 \begin{align}
y_i^{(i)}   \leftarrow  \text{Majority}( \{y_i^{(j)}:   j \in \Ss_1\})   = \hv_i^\T   \Meg_{1} .  \non
 \end{align}
At the end of this step,  for any honest $\PRO$~$i$, $i \notin  \Fc$, the value of $y_i^{(i)}$  becomes $y_i^{(i)} = \hv_i^\T   \Meg_{1}$, which is encoded with $\Meg_{1}$. 
In the second step of Phase~4, $\PRO$~$i$,  $i \in \Ss_0$,  sends the updated value of $y_i^{(i)}$ to $\PRO$~$j$, $\forall j \in\Ss_0, j\neq i$.  
After this,  $\PRO$~$i$, $i \in \Ss_0$, decodes its message using the updated observations $\{y_1^{(1)}, y_2^{(2)}, \cdots, y_n^{(n)}\}$, where  $n-t$ observations of which are guaranteed to be $y_j^{(j)} = \hv_j^\T   \Meg_{1}, \forall j \notin  \Fc$.  Since the number of mismatched observations, i.e., the observations not encoded with the message $\Meg_{1}$, is no more than $t$, then $\PRO$~$i$, $i \in \Ss_0\setminus \Fc$, decodes its message and outputs $ \Me^{(i)} = \Meg_{1}$. At the same time,  $\PRO$~$i$, $i \in \Ss_1\setminus \Fc$, outputs the original value as $ \Me^{(i)} = \Meg_{1}$.
Thus,  all of the honest $\Pros$  successfully agree on   the same   message, i.e., $ \Me^{(i)} = \Meg_{1}, \  \forall i  \notin \Fc$.
  With this we complete the proof of Lemma~\ref{lm:p4samem}.

\subsection{Proof of Lemma~\ref{lm:p4samemsih}}  \label{sec:p4samemsih}

Now we will prove that, given $n\geq 3t+1$,  if  all honest $\Pros$ have  the same initial message, then at the end of COOL  all honest $\Pros$ agree on this  initial message.  One example is depicted in  Fig.~\ref{fig:coolexv}. 

Assume that all  honest $\Pros$  have the same  initial message as $ \Meg_{1}$, for some non-empty value $\Meg_{1}$ (see Fig.~\ref{fig:coolexv} as an example). 
For this scenario, in Phase~1 all honest $\Pros$ will set their success indicators as ones and keep their updated messages exactly the same as the initial message  $\Meg_{1}$, because  every honest $\Pro$ will have  no more than $t$   number of mismatched observations (see \eqref{eq:lkindicator} and \eqref{eq:sindicator}).  
Then, in Phase~2,  none of the honest $\Pros$ will be in the list of $\Ss_0$ defined in  \eqref{eq:vr0}, because  their success indicators are all ones.  
In Phase~3, all honest $\Pros$ will again keep their success indicators as ones and keep their updated messages exactly the same as the initial message  $\Meg_{1}$. 
In this scenario,  the consensus of the votes $\{\Vr_1, \Vr_2, \cdots, \Vr_n\}$ in Phase~3 should be $1$, i.e., all honest $\Pros$ will go to Phase~4, as described in Case~(b) in the previous sub-section.
Given the same  updated message $\Meg_{1}$ and the same nonzero success indicator as inputs at all of the honest $\Pros$,  in Phase~4 all honest $\Pros$  will output the messages  that are exactly the same as the initial message  $\Meg_{1}$.  Thus, for this scenario,  all honest $\Pros$ eventually agree on the  initial message $\Meg_{1}$. 
At this point we complete the proof of Lemma~\ref{lm:p4samemsih}.

\subsection{Proof of Lemma~\ref{lm:terminate}}  \label{sec:terminate}

 In the proposed COOL, given $n\geq 3t+1$,  it is guaranteed that  all  honest $\Pros$  eventually terminate together at the last step of Phase~3,  or terminate together at the last step of Phase~4.  It is also guaranteed that every honest $\Pro$ eventually  outputs  a message when it terminates, which completes the proof of Lemma~\ref{lm:terminate}.

\section{The extension of COOL to the Byzantine broadcast problem}  \label{sec:COOLbb}

 The proposed COOL designed for the BA setting can be adapted into the BB setting.  
 Note that  the known lower bounds on resilience, round complexity and communication complexity for error-free BA protocols, can also be applied to any error-free BB protocols. 

For the BB setting, we just add an additional step into the proposed COOL described in Section~\ref{sec:COOL}.
Specifically, for the BB setting, at first the leader is designed to send each   $\Pro$ an $\ell$-bit message that can be considered as the initial message in the BA setting. Then, COOL is applied into this BB setting from this step to achieve the consistent and validated consensus.  We call the proposed COOL with the additional step mentioned above as an adapted COOL for this BB setting. 
When $n\geq 3t+1$, this   adapted COOL is guaranteed to satisfy the termination, consistency and validity conditions in all executions in this BB setting. 
Note that for the added step, i.e., sending initial messages from the leader to all $\Pros$, the additional communication complexity is just $O(n \Lh)$ bits. 
Therefore, the total communication complexity of the adapted COOL for this BB setting is \[O(\max\{n\Lh, n t \log t \} + n \Lh) = O(\max\{n\Lh, n t \log t \}) \quad  \text{bits}. \]
Furthermore,  the round complexity of the adapted COOL for this BB setting is \[O(t + 1) = O(t)  \quad  \text{rounds}. \]
Hence, the  adapted COOL is an error-free signature-free information-theoretic-secure   BB  protocol that achieves the consensus on an $\Lh$-bit message  with optimal resilience,  asymptotically optimal round complexity, and asymptotically  optimal communication complexity  when $\ell \geq t   \log t$.  
This result   serves as the achievability proof of Theorem~\ref{thm:BBs}.

\section{Conclusion}   \label{sec:concl}

In this work, we proposed  COOL, a deterministic  BA protocol designed from coding theory, together with graph theory and linear algebra,  which achieves the BA consensus on an $\Lh$-bit message with optimal resilience, 
 asymptotically optimal round complexity,  and asymptotically optimal communication complexity when $\ell \geq t   \log t$, simultaneously.  
 Our protocol is guaranteed to be correct in all executions (\emph{error-free}) and it does not rely on cryptographic technique such as signatures, hashing, authentication and secret sharing (\emph{signature-free}). 
Furthermore, our protocol is robust even when the adversary has unbounded computational power (\emph{information-theoretic secure}). 
The adapted  COOL  also achieves the BB consensus  on an $\Lh$-bit message  with optimal resilience,  asymptotically optimal round complexity, and asymptotically optimal communication complexity when $\ell \geq t   \log t$, simultaneously.  
 Based on the performance of the proposed COOL and the known lower bounds, we characterized the optimal \emph{communication complexity exponent} achievable by \emph{any} error-free BA (or BB) protocol when $n\geq 3t+1$. 
This work reveals that coding is an effective approach for achieving the fundamental limits in resilience, round complexity and  communication complexity of Byzantine agreement and its variants.



\end{document}